\PassOptionsToPackage{prologue,usenames,dvipsnames,table}{xcolor}
\PassOptionsToPackage{bookmarks,unicode,colorlinks=true}{hyperref}
\pdfoutput=1
\newif\ifanonymous

\anonymousfalse

\ifanonymous
\documentclass[acmsmall,review,anonymous]{acmart}
\else
\documentclass[acmsmall,screen]{acmart}
\fi

\usepackage{fontawesome}
\usepackage{amsmath}
\usepackage{stmaryrd}
\usepackage{mathtools}
\usepackage{csquotes}
\usepackage{multirow}
\usepackage{xfrac}
\usepackage{nicefrac}
\renewcommand{\sfrac}[2]{\nicefrac{#1}{#2}}
\usepackage{xspace}
\usepackage{etoolbox}
\usepackage{scalerel}
\usepackage[inline]{enumitem}
\usepackage{newfloat}
\usepackage[skip=4pt]{caption}
\usepackage{subcaption}
\usepackage{blindtext}
\graphicspath{{./Images/}}

\usepackage{adjustbox}
\usepackage{makecell}
\usepackage{tikz}
\usetikzlibrary{patterns,calc,arrows,shapes,snakes,automata,backgrounds,petri,positioning,decorations.pathmorphing,intersections}
\usepackage{tkz-euclide}
\tikzset{mynode/.style={draw,solid,circle,inner sep=1pt}}
\usepackage[customcolors]{hf-tikz}
\hfsetfillcolor{gray!10}
\hfsetbordercolor{none}
\usepackage{pgfplots}
\usepackage{wrapfig}
\usepackage{bussproofs}
\usepackage{picinpar} 
\usepackage{pifont}

\usepackage[capitalize,nameinlink]{cleveref}
\Crefname{equation}{Eq.}{Eqs.}

\usepackage{footnote}
\makesavenoteenv{figure}
\makesavenoteenv{table}

\DeclareFloatingEnvironment[
fileext = los ,
listname = {List of programs} ,
name = Prog.
]{program}

\theoremstyle{plain}
\newtheorem*{thm*}{Theorem}
\newtheorem*{exmp*}{Example}
\newtheorem*{counterexmp*}{Counterexample}

\newtheoremstyle{ourstyle}{}{}{\itshape}{}{\bfseries}{.}{ }{\thmname{#1}\thmnumber{ #2}\thmnote{ (#3)}}
\theoremstyle{ourstyle}

\newtheorem{theorem}{Theorem}
\newtheorem{definition}[theorem]{Definition}
\newtheorem{lemma}[theorem]{Lemma}
\newtheorem{corollary}[theorem]{Corollary}
\newtheorem{proposition}[theorem]{Proposition}

\newtheoremstyle{exmpstyle}{}{}{}{}{\bfseries}{.}{ }{\thmname{#1}\thmnumber{ #2}\thmnote{ (#3)}}
\theoremstyle{exmpstyle}

\newtheorem{example}[theorem]{Example}

\newtheoremstyle{rmkstyle}{}{}{}{}{\bfseries}{.}{ }{\thmname{#1}\thmnote{ (#3)}}
\theoremstyle{rmkstyle}

\newtheorem{remark}{Remark}
\include{macros}

\setcopyright{rightsretained}
\acmDOI{10.1145/3649844}
\acmYear{2024}
\copyrightyear{2024}
\acmSubmissionID{oopslaa24main-p124-p}
\acmJournal{PACMPL}
\acmVolume{8}
\acmNumber{OOPSLA1}
\acmArticle{127}
\acmMonth{4}
\received{21-OCT-2023}
\received[accepted]{2024-02-24}

\makeatletter
\def\@fnsymbol#1{\ensuremath{\ifcase#1\or \dagger\or \ddagger\or
		\mathsection\or \mathparagraph\or \|\or **\or \dagger\dagger
		\or \ddagger\ddagger \else\@ctrerr\fi}}
\makeatother

\begin{document}
	
	\title[Exact Bayesian Inference for Loopy Probabilistic Programs using Generating Functions]{Exact Bayesian Inference for Loopy Probabilistic Programs using Generating Functions}
	
	
	\author{Lutz Klinkenberg}
	\authornotemark[1]
	\email{lutz.klinkenberg@cs.rwth-aachen.de}
	\orcid{0000-0002-3812-0572}
	\affiliation{%
		\institution{RWTH Aachen University}
		\city{Aachen}
		\country{Germany}
	}
	
	\author{Christian Blumenthal}
	\email{christian.blumenthal@rwth-aachen.de}
	\orcid{0009-0003-6427-0229}
	\affiliation{%
		\institution{RWTH Aachen University}
		\city{Aachen}
		\country{Germany}
	}

	\author{Mingshuai Chen}
	\authornote{The corresponding authors} 
	\email{m.chen@zju.edu.cn}
	\orcid{0000-0001-9663-7441}
	\affiliation{%
		\institution{Zhejiang University}
		\city{Hangzhou}
		\country{China}
	}

	\author{Darion Haase}
	\email{darion.haase@cs.rwth-aachen.de}
	\orcid{0000-0001-5664-6773}
	\affiliation{%
		\institution{RWTH Aachen University}
		\city{Aachen}
		\country{Germany}
	}

	\author{Joost-Pieter Katoen}
	\email{katoen@cs.rwth-aachen.de}
	\orcid{0000-0002-6143-1926}
	\affiliation{%
		\institution{RWTH Aachen University}
		\city{Aachen}
		\country{Germany}
	}

	\renewcommand{\shortauthors}{L.~Klinkenberg, C.~Blumenthal, M.~Chen, D.~Haase and J.-P.~Katoen}
	
	\begin{abstract}
	We present an exact Bayesian inference method for inferring posterior distributions encoded by probabilistic programs featuring possibly \emph{unbounded loops}. Our method is built on a denotational semantics represented by \emph{probability generating functions}, which resolves semantic intricacies induced by intertwining discrete probabilistic loops with \emph{conditioning} (for encoding posterior observations). We implement our method in a tool called \toolname{Prodigy}; it augments existing computer algebra systems with the theory of generating functions for the (semi-)automatic inference and quantitative verification of conditioned probabilistic programs. Experimental results show that \toolname{Prodigy} can handle various infinite-state loopy programs and exhibits comparable performance to state-of-the-art exact inference tools over loop-free benchmarks.
\end{abstract}


	
	\setlength{\floatsep}{1\baselineskip}
	\setlength{\textfloatsep}{1\baselineskip}
	\setlength{\intextsep}{1\baselineskip}
	
	\maketitle

	\section{Introduction}
\label{sec:intro}
Probabilistic programming is used to describe stochastic models in the form of executable computer programs. It enables fast and natural ways of designing statistical models without ever resorting to random variables in the mathematical sense. The so-obtained probabilistic programs \cite{Kozen,ACM:conf/fose/Gordon14,DBLP:journals/corr/abs-1809-10756,barthe_katoen_silva_2020,DBLP:journals/pacmpl/HoltzenBM20} are typically normal-looking programs describing posterior probability distributions. They intrinsically code up randomized algorithms \cite{DBLP:books/daglib/0012859} and are at the heart of approximate computing \cite{DBLP:journals/cacm/CarbinMR16} as well as probabilistic machine learning \cite[Chapter 8]{DBLP:journals/corr/abs-1809-10756}. 
%
One prominent example is \textsc{Scenic} \cite{scenic-mlj22} -- a domain-specific probabilistic programming language to describe and generate scenarios for, e.g., robotic systems, that can be used to train convolutional neural networks; 
\textsc{Scenic} features the ability to declaratively impose (hard and soft) constraints over the generated models by means of \emph{conditioning} via posterior observations. Moreover, a large volume of literature has been devoted to combining the strength of probabilistic and differentiable programming in a mutually beneficial manner; see \cite[Chapter 8]{DBLP:journals/corr/abs-1809-10756} for recent advancements in deep probabilistic programming.
%
%
%

Reasoning about probabilistic programs amounts to addressing various \emph{quantities} like assertion-violation probabilities \cite{DBLP:conf/pldi/WangS0CG21}, preexpectations \cite{DBLP:conf/cav/CKKMS20,DBLP:journals/pacmpl/HarkKGK20,FENG-OOPSLA2023}, moments \cite{DBLP:conf/pldi/Wang0R21,DBLP:journals/pacmpl/MoosbruggerSBK22}, expected runtimes \cite{DBLP:journals/jacm/KaminskiKMO18}, and concentrations \cite{DBLP:conf/cav/ChakarovS13,DBLP:conf/cav/ChatterjeeFG16}.
Probabilistic inference is one of the most important tasks in quantitative reasoning which aims to derive a program's posterior distribution.
In contrast to sampling-based approximate inference, 
inferring the \emph{exact} distribution has several benefits \cite{DBLP:conf/pldi/GehrSV20}, e.g., no loss of precision, natural support for symbolic parameters, and efficiency on models with certain structures.

Exact probabilistic inference, however, is a notoriously difficult task \cite{DBLP:journals/ai/Cooper90,DBLP:journals/acta/KaminskiKM19,DBLP:journals/toplas/OlmedoGJKKM18,ROTH1996273,Ackermann2019}; even for Bayesian networks, it is already \texttt{PP}-complete \cite{Kwisthout2009Computational,DBLP:journals/jair/LittmanGM98}.
The challenges mainly arise from three program constructs:
\begin{enumerate*}[label=(\roman*)]
	\item\label{issue:loop}unbounded \codify{while}-{\allowbreak}loops and/or recursion,
	\item\label{issue:support}infinite-support distributions, and 
	\item\label{issue:conditioning}conditioning.
\end{enumerate*}
Specifically, reasoning about probabilistic loops amounts to computing quantitative fixed points (see \cite{dahlqvist_silva_kozen_2020}) that are highly intractable in practice; admitting infinite-support distributions requires closed-form (i.e., finite) representations of program semantics; and conditioning \enquote{reshapes} the posterior distribution as per observed events thus yielding another layer of semantic intricacies (see \cite{DBLP:journals/toplas/OlmedoGJKKM18,Ackermann2019,DBLP:conf/esop/BichselGV18}).

This paper proposes to use \emph{probability generating functions} (PGFs) -- a subclass of \emph{generating functions} (GFs) \cite{generatingfunctionology} -- to do exact inference for \emph{discrete}, \emph{loopy}, \emph{infinite-state} probabilistic programs \emph{with conditioning}, thus addressing challenges \ref{issue:loop}, \ref{issue:support}, and \ref{issue:conditioning}, whilst aiming to push the limits of automation as far as possible by leveraging the strength of existing computer algebra systems like \sympy \cite{SymPy} and \ginac \cite{DBLP:journals/jsc/BauerFK02,ginac}. We extend the PGF-based semantics by \citet{LOPSTR}, which enables exact quantitative reasoning for, e.g., deciding probabilistic equivalence \cite{CAV22} and proving non-almost-sure termination \cite{LOPSTR} for certain programs \emph{without conditioning}.
Orthogonally, \citet{zaiser2023exact} recently employed PGFs to conduct exact Bayesian inference for conditioned probabilistic programs with infinite-support distributions yet \emph{no loops}.
Note that \emph{having loops and conditioning intertwined} incurs semantic intricacies; see \cite{DBLP:journals/toplas/OlmedoGJKKM18,DBLP:conf/esop/BichselGV18}. Let us illustrate our inference method and how it addresses such semantic intricacies by means of a number of examples of increasing complexity.

\begin{figure}[t]
	\begin{minipage}[b]{.365\linewidth}
		\begin{align*}
			& \pchoice{\assign{\progvar{w}}{0}}{\sfrac{5}{7}}{\assign{\progvar{w}}{1}}\,\fatsemi \\
			& \IFINLINE{\progvar{w} = 0}{\assign{c}{\poisson{6}}}\\
			& \ELSEINLINE{\assign{c}{\poisson{2}}}\,\fatsemi \\
			& \observe{\progvar{c} = 5}
		\end{align*}%
		\captionof{program}{The telephone operator.}
		\label{prog:telephone}
	\end{minipage}
	\hspace*{.5cm}
	\begin{minipage}[b]{.4\linewidth}
		\centering
		\pgfplotsset{width=3.6cm,compat=1.8}
		\begin{subfigure}[b]{0.42\linewidth}
			\centering
			\resizebox{\linewidth}{!}{
				\begin{tikzpicture}
					\begin{axis}[
						unit vector ratio*={2 1 1},
						ybar stacked,
						bar width=15pt,
						x=1cm,
						enlarge x limits={abs=.4cm},
						enlarge y limits={abs=0.05},
						legend style={at={(0.5,-0.20)},
							anchor=north,legend columns=-1},
						xlabel={$n$},
						ylabel={{$\textup{Pr}\left(\progvar{w} = n\right)$}},
						symbolic x coords={0, 1},
						xtick style={
							/pgfplots/major tick length=1.5pt,
						},
						ytick style={
							/pgfplots/major tick length=2pt,
						},
						ytick pos=left, 
						xtick pos=bottom,
						xtick=data,
						ymax=1.0,
						ymin=0,
						ytick={0.0,0.25,0.50,0.75,1.00},
						]
						\addplot+[ybar] plot coordinates {(0,0.71428571428) (1,0.28571428571)};
					\end{axis}
				\end{tikzpicture}
			}
			\vspace*{-.6cm}
			\caption{initial belief}
			\label{fig:telephone-before}
		\end{subfigure}
		\hspace*{.2cm}
		\begin{subfigure}[b]{0.42\linewidth}
			\centering
			\resizebox{\linewidth}{!}{
				\begin{tikzpicture}
					\begin{axis}[
						unit vector ratio*={2 1 1},
						ybar stacked,
						legend style={
							nodes={scale=0.75, transform shape},
							legend columns= 1,
							at={(.72,.97)},
							anchor=north,
							draw=none,
							fill=none
						},
						legend cell align={left}, 
						bar width=15pt,
						x=1cm,
						enlarge x limits={abs=0.4cm},
						enlarge y limits={abs=0.05},
						xlabel={$n$},
						ylabel={{$\textup{Pr}\left(\progvar{w} = n\right)$}},
						symbolic x coords={0, 1},
						xtick style={
							/pgfplots/major tick length=1.5pt,
						},
						ytick style={
							/pgfplots/major tick length=2pt,
						},
						ytick pos=left, 
						xtick pos=bottom,
						xtick=data,
						ymax=1.0,
						ymin=0,
						ytick={0.0,0.25,0.50,0.75,1.00},
						]			
						\addplot+[ybar] plot coordinates {(0,0.71428571428) (1,0.08246232078)};
						\addplot+[ybar] plot coordinates {(0,0.20325196494) (1,0)};
						\addplot+[ybar,gray!80,fill=gray!20,postaction={pattern=north east lines,pattern color=gray!80}] plot coordinates {(0,0) (1,0.20325196493)};
						
					\end{axis}
					\coordinate (a) at (1.19,.5);
					\coordinate (b) at (.7,1.2);
					\draw[->,>=stealth'] (a) [out=90, in=0] to (b);
				\end{tikzpicture}
			}
			\vspace*{-.6cm}
			\caption{updated belief}
			\label{fig:telephone-after}
		\end{subfigure}
		\caption{The distribution of $\progvar{w}$ in Prog.~\ref{prog:telephone}.}
		\label{fig:telephone}
	\end{minipage}
\end{figure}

\paragraph{Conditioning in loop-free programs.}
Consider the loop-free program Prog.~\ref{prog:telephone} producing an \emph{infinite-support} distribution. It describes a telephone operator who is unaware of whether today is a weekday or weekend.
The operator's initial belief is that with probability $\sfrac{5}{7}$ it is a weekday ($\progvar{w}=0$) and thus with probability $\sfrac{2}{7}$ weekend ($\progvar{w}=1$); see \cref{fig:telephone-before}.
Usually, on weekdays there are 6 incoming calls per hour on average; on weekends this rate decreases to 2 calls -- both rates are subject to a Poisson distribution.
The operator observes 5 calls in the last hour, and the inference task is to compute the posterior distribution in which the initial belief is updated based on the observation.
Our approach can automatically infer the updated belief (see \cref{fig:telephone-after}) with $\textup{Pr}(\progvar{w}=0) = \frac{1215}{1215 + 2 \cdot e^4} \approx 0.9175$.
(Detailed calculations of the PGF semantics for Prog.~\ref{prog:telephone} are given in \cref{ex:telephone} on page \pageref{ex:telephone}.)

\paragraph{Conditioning outside loops.}
Prog.~\ref{prog:odd_geo} describes an iterative algorithm that repeatedly flips a fair coin -- while counting the number of trials ($\progvar{t}$) -- until seeing tails ($\progvar{h} = 0$), and observes that this number is odd. In fact, the \codify{while}-loop produces a geometric distribution in $\progvar{t}$ (cf.\ Fig.~\ref{fig:odd-geo-before}), after which the \codify{observe} statement \enquote{blocks} all program runs where $\progvar{t}$ is even and normalizes the probabilities of the remaining runs (cf.\ Fig.~\ref{fig:odd-geo-after}). Note that Prog.~\ref{prog:odd_geo} features an unbounded looping behavior (inducing an infinite-support distribution) whose exact output distribution thus cannot be inferred by state-of-the-art inference engines, e.g., neither by ($\lambda$)\textsc{PSI} \cite{DBLP:conf/cav/GehrMV16,DBLP:conf/pldi/GehrSV20}, nor by the PGF-based approach in \cite{zaiser2023exact}.
%
%
	However, given a suitable loop invariant, our tool is able to derive 
	the posterior distribution of Prog.~\ref{prog:odd_geo} in an automated fashion: for any input with $\progvar{t}=0$, the posterior is represented as the closed-form PGF
	\begin{align*}
		\frac{3 \cdot T}{4-T^2} \eeq \sum_{n=0}^{\infty}\, \underbrace{-3 \cdot 2^{-2-n} \cdot \left(-1 + (-1)^n\right)}_{\textup{Pr}\left(\progvar{t} \,=\, n \,\wedge\, \progvar{h} \,=\, 0\right)}\, \cdot~T^n H^0~,
	\end{align*}%
	where $T,H$ are formal \emph{indeterminates} corresponding to the program variables $\progvar{t}$ and $\progvar{h}$, respectively. From this closed-form PGF, we can extract various quantitative properties of interest, e.g., the expected value of $\progvar{t}$ is $\mathbb{E}[\progvar{t}] = \left(\frac{\partial}{\partial T}\frac{3 \cdot T}{4-T^2}\right)[H/0, T/1] = \frac{5}{3}$, or compute concentration bounds (aka tail bounds) such as $\textup{Pr}(\progvar{t} > 100) \leq \frac{5}{3\cdot 100} = \frac{1}{60}$ \`{a} la Markov's inequality \cite{dubhashi2009concentration}.
%

\begin{figure}[t]
	\centering
	\begin{minipage}[b]{.365\linewidth}
		\centering
		\begin{align*}
			& \assign{\progvar{h}}{1}\,\fatsemi\\
			& \WHILE{\progvar{h} = 1}\\
			& \quad \PCHOICE{\assign{\progvar{t}}{\progvar{t} + 1}}{\sfrac{1}{2}}{\assign{\progvar{h}}{0}}\\
			& \}\,\fatsemi\\
			& \observe{\progvar{t} \equiv 1 \!\!\!\!\pmod{2}}
		\end{align*}
		\captionof{program}{The odd geometric distribution.}
		\label{prog:odd_geo}
	\end{minipage}
	\hfil
	\begin{minipage}[b]{.625\linewidth}
		\centering
		\pgfplotsset{width=7cm,compat=1.8}
		\begin{subfigure}[b]{0.46\linewidth}
			\centering
			\resizebox{\linewidth}{!}{
				\begin{tikzpicture}
					\begin{axis}[
						unit vector ratio*={1 6 1},
						ybar stacked,
						bar width=15pt,
						enlarge y limits={abs=0.05},
						legend style={at={(0.5,-0.20)},
							anchor=north,legend columns=-1},
						xlabel={$n$},
						ylabel={{$\textup{Pr}\left(\progvar{t} = n\right)$}},
						symbolic x coords={0, 1, 2, 3, 4, 5, 6, 7, 8},
						xtick style={
							/pgfplots/major tick length=3pt,
						},
						ytick style={
							/pgfplots/major tick length=3pt,
						},
						ytick pos=left, 
						xtick pos=bottom,
						xtick=data,
						ymax=0.8,
						ymin=0,
						ytick={0.0,0.2,0.4,0.6,0.8},
						]
						\addplot+[ybar] plot coordinates {(0,0.5) (1,0.25) (2,0.125) (3,0.0625) (4,0.03125) (5,0.015625) (6,0.0078125) (7,0.00390625) (8,0.001953125)};
					\end{axis}
				\end{tikzpicture}
			}
			\vspace*{-.6cm}
			\caption{before conditioning}
			\label{fig:odd-geo-before}
		\end{subfigure}
		\hfil
		\begin{subfigure}[b]{0.46\linewidth}
			\centering
			\resizebox{\linewidth}{!}{
				\begin{tikzpicture}
					\begin{axis}[
						unit vector ratio*={1 6 1},
						ybar stacked,
						legend style={
							nodes={scale=0.75, transform shape},
							legend columns= 1,
							at={(.72,.97)},
							anchor=north,
							draw=none,
							fill=none
						},
						legend cell align={left}, 
						bar width=15pt,
						enlarge y limits={abs=0.05},
						xlabel={$n$},
						ylabel={{$\textup{Pr}\left(\progvar{t} = n\right)$}},
						symbolic x coords={0, 1, 2, 3, 4, 5, 6, 7, 8},
						xtick style={
							/pgfplots/major tick length=3pt,
						},
						ytick style={
							/pgfplots/major tick length=3pt,
						},
						ytick pos=left, 
						xtick pos=bottom,
						xtick=data,
						ymax=0.8,
						ymin=0,
						ytick={0.0,0.2,0.4,0.6,0.8},
						]
						\addplot+[ybar,gray!80,fill=gray!20,postaction={pattern=north east lines,pattern color=gray!80}] plot coordinates {(0,0.5) (1,0) (2,0.125) (3,0) (4,0.03125) (5,0) (6,0.0078125) (7,0) (8,0.001953125)};
						\pgfplotsset{cycle list shift=-1} 
						\addplot+[ybar] plot coordinates {(0,0) (1,0.25) (2,0) (3,0.0625) (4,0) (5,0.015625) (6,0) (7,0.00390625) (8,0)};
						\addplot+[ybar] plot coordinates {(0,0) (1,0.5) (2,0) (3,0.125) (4,0) (5,0.03125) (6,0) (7,0.0078125) (8,0)};
						
						\legend{blocked probability, original probability, normalized probability}
					\end{axis}
				\end{tikzpicture}
			}
			\vspace*{-.6cm}
			\caption{after conditioning}
			\label{fig:odd-geo-after}
		\end{subfigure}
		\caption{Snippets of the distribution of $\progvar{t}$ in Prog.~\ref{prog:odd_geo}.}
		\label{fig:odd-geo}
	\end{minipage}
\end{figure}

\paragraph{Conditioning inside i.i.d.\ loops.}
As argued by \citet{DBLP:journals/toplas/OlmedoGJKKM18} and \citet{DBLP:conf/esop/BichselGV18}, having
{\makeatletter
	\let\par\@@par
	\par\parshape0
	\everypar{}
	\begin{wrapfigure}{r}{0.36\textwidth}
		\vspace*{-6mm}
		\begin{minipage}{1\linewidth} 
		\begin{align*}
			& \assign{\progvar{h}}{1}\,\fatsemi\\
			& \WHILE{\progvar{h} = 1}\\
			& \quad \PCHOICE{\assign{\progvar{t}}{\progvar{t} + 1}}{\sfrac{1}{2}}{\assign{\progvar{h}}{0}}\,\fatsemi\\
			& \quad \observe{\progvar{h} = 1}\\
			& \}
		\end{align*}
	\end{minipage}
		\captionof{program}{\codify{observe} inside loop.}
		\label{prog:div_and_obsfail}
	\end{wrapfigure}
	\noindent
	loops and conditioning intertwined incurs semantic intricacies: Consider Prog.~\ref{prog:div_and_obsfail} -- a variant of Prog.~\ref{prog:odd_geo} where instead we observe $\progvar{h}=1$ \emph{inside} the \codify{while}-loop.
	Prog.~\ref{prog:div_and_obsfail} features an \emph{i.i.d.\ loop}, i.e., the set of states reached upon the end of different loop iterations are \emph{independent and identically distributed}. This program is interesting since it conditions to a \emph{zero-probability event}, i.e., the probability of infinitely often ignoring $\assign{\progvar{h}}{0}$ is zero, which is important yet non-trivial to
	\par}%
\noindent detect in general. Assigning a meaningful semantics to Prog.~\ref{prog:div_and_obsfail} is delicate: Intuitively, the \codify{observe} statement prevents the \codify{while}-loop from terminating since we always observe that we have taken the left branch, and therefore never set the termination flag $\progvar{h}=0$.
As a consequence, all runs that eventually would terminate are invalid as they violate the observation criterion.
The single run that does satisfy the criterion in turn is never able to exit the loop (cf.\ \cref{subsec:conditioned-semantics}).
In previous work on using PGFs (without conditioning) \cite{LOPSTR,CAV22}, the semantics of non-termination is represented as subprobability distributions where the \enquote{missing} probability mass captures the probability of divergence. \citet{zaiser2023exact} circumvent such semantic intricacies by syntactically imposing certainly terminating programs (due to the absence of loops and recursion). In our work, we distinguish non-termination behaviors from \codify{observe} violations \cite{DBLP:journals/toplas/OlmedoGJKKM18,DBLP:conf/esop/BichselGV18}, which allows us to show that the \codify{while}-loop in Prog.~\ref{prog:div_and_obsfail} is in fact equivalent to
\begin{align*}
	\ITE{\progvar{h} = 1}{\observe{\FALSE}}{\pskip}
\end{align*}%
which in turn reduces to $ \observe{h\neq1}$.

For certain programs, conditioning inside loops can be treated differently. 
These approaches 
include
\begin{enumerate*}[label=(\arabic*)]
	\item \emph{hoisting} \cite{DBLP:journals/toplas/OlmedoGJKKM18} that removes observations completely from conditioned probabilistic programs, which however relies on  intractable fixed point computations to hoist \codify{observe} statements inside loops;
	\item the \emph{pre-image transformation} \cite{DBLP:conf/aaai/NoriHRS14} that propagates observations backward through the program, which however cannot hoist the \codify{observe} statement through probabilistic choices, as in Prog.~\ref{prog:div_and_obsfail};
	\item the \emph{ad hoc solution} that simply pulls the \codify{observe} statement outside the loop, which however works only for special i.i.d.\ loops like Prog.~\ref{prog:div_and_obsfail}: The  \codify{observe} statement in Prog.~\ref{prog:div_and_obsfail} can be equivalently moved downward to the outside of the loop, but such transformation does not generalize to non-i.i.d.\ loops (which may have data flow across different loop iterations) as exemplified below.
\end{enumerate*}

\paragraph{Conditioning inside non-i.i.d.\ loops.}
The probabilistic loop in Prog.~\ref{prog:bit-flipping} models a discrete sampler which keeps tossing two fair coins ($\progvar{h}_1$ and $\progvar{h}_2$) until they both turn tails. The \codify{observe} statement in this program conditions to the event that at least one of the coins yields the same outcome as in the
{\makeatletter
	\let\par\@@par
	\par\parshape0
	\everypar{}
	\begin{wrapfigure}{r}{0.37\textwidth}
		\vspace*{-7.6mm}
		\begin{minipage}{1\linewidth} 
			\begin{align*}
				& \assign{\progvar{n}}{0}\,\fatsemi\\
				& \assign{\progvar{h}_1}{1}\,\fatsemi \assign{\progvar{h}_2}{1}\,\fatsemi \assign{\progvar{h}'_1}{1}\,\fatsemi \assign{\progvar{h}'_2}{1}\,\fatsemi\\
				& \WHILE{\neg\left(\progvar{h}_1 = 0 \wedge \progvar{h}_2 = 0\right)}\\
				& \quad \assign{\progvar{h}_1}{\bernoulli{\sfrac{1}{2}}}\,\fatsemi\\
				& \quad \assign{\progvar{h}_2}{\bernoulli{\sfrac{1}{2}}}\,\fatsemi\\
				& \quad \observe{\progvar{h}_1 = \progvar{h}'_1 \vee \progvar{h}_2 = \progvar{h}'_2}\,\fatsemi\\
				& \quad \assign{\progvar{h}'_1}{\progvar{h}_1}\,\fatsemi\\
				& \quad \assign{\progvar{h}'_2}{\progvar{h}_2}\,\fatsemi\\
				& \quad \assign{\progvar{n}}{\progvar{n}+1}\\
				& \}
			\end{align*}
		\end{minipage}
		\captionof{program}{The non-i.i.d.\ discrete sampler.}
		\label{prog:bit-flipping}
	\end{wrapfigure}
	\noindent
	 previous iteration, thereby imposing the global effect to \enquote{reset} the counter $\progvar{n}$ and restart the program upon observation violations. This way of conditioning -- that induces data dependencies across consecutive loop iterations -- renders the loop non-i.i.d.\ and, as a consequence, no known tactic can be employed to pull the observation outside the loop. However, given a suitable invariant -- in the form of a conditioned \emph{loop-free} program that is equivalent to the loop -- our method automatically infers that the posterior distribution is $-\frac{7 \cdot N^2}{N^2+8 \cdot N - 16}$, where $N$ is the formal indeterminate of the counter $\progvar{n}$ (note that $\progvar{h}_1 = \progvar{h}_2 = \progvar{h}'_1 = \progvar{h}'_2 = 0$ on termination). Furthermore, our inference framework admits \emph{parameters} in both programs and invariants for, e.g., encoding distributions with unknown probabilities like $\bernoulli{p}$ with 
	\par}%
\noindent 
  $p \in (0,1)$; it is capable of determining possible valuations of these parameters such that the given invariant is equivalent to the loop in question. The support of parameters in our approach enables template-based invariant synthesis (see, e.g., \cite{DBLP:conf/tacas/BatzCJKKM23}) and model repair (cf.\ \cite{DBLP:conf/birthday/CeskaD0JK19}), as detailed in \Cref{sec:synthesis}.

\newsavebox{\priorDist}
\sbox{\priorDist}{%
	\pgfplotsset{width=7cm,compat=1.8}
	\begin{tikzpicture}
		\hspace*{-.1cm}
		\begin{axis}[
			unit vector ratio*={1 6 1},
			ybar stacked,
			bar width=15pt,
			xticklabel=\empty,
			yticklabel=\empty,
			enlarge y limits={abs=0.05},
			legend style={at={(0.5,-0.20)},
				anchor=north,legend columns=-1},
			xlabel={},
			ylabel={},
			xtick style={
				/pgfplots/major tick length=3pt,
			},
			ytick style={
				/pgfplots/major tick length=3pt,
			},
			ytick pos=left, 
			xtick pos=bottom,
			xtick=data,
			ymax=1,
			ymin=0,
			ytick={0.0,0.2,0.4,0.6,0.8, 1},
			]
			\addplot+[ybar] plot coordinates {(0,1) (1,0) (2,0) (3,0) (4,0) (5,0) (6,0) (7,0) (8,0)};
		\end{axis}
	\end{tikzpicture}
}

\newsavebox{\postDist}
\sbox{\postDist}{%
	\pgfplotsset{width=7cm,compat=1.8}
	\begin{tikzpicture}
		\hspace*{-.1cm}
		\begin{axis}[
			unit vector ratio*={1 6 1},
			ybar stacked,
			xticklabel=\empty,
			yticklabel=\empty,
			legend cell align={left}, 
			bar width=15pt,
			enlarge y limits={abs=0.05},
			xlabel={},
			ylabel={},
			xtick style={
				/pgfplots/major tick length=3pt,
			},
			ytick style={
				/pgfplots/major tick length=3pt,
			},
			ytick pos=left, 
			xtick pos=bottom,
			xtick=data,
			ymax=1,
			ymin=0,
			ytick={0.0,0.2,0.4,0.6,0.8,1},
			]
			\addplot+[ybar] plot coordinates {(0,0) (1,0.75) (2,0) (3,0.1875) (4,0) (5,0.046875) (6,0) (7,0.01171875) (8,0)};
			
		\end{axis}
	\end{tikzpicture}
}

\newsavebox{\lightningDist}
\sbox{\lightningDist}{%
	\pgfplotsset{width=7cm,compat=1.8}
	\begin{tikzpicture}
		\hspace*{-.1cm}
		\begin{axis}[
			unit vector ratio*={1 6 1},
			ybar stacked,
			xticklabel=\empty,
			yticklabel=\empty,
			legend cell align={left}, 
			bar width=15pt,
			enlarge y limits={abs=0.05},
			xlabel={},
			ylabel={},
			xtick style={
				/pgfplots/major tick length=3pt,
			},
			ytick style={
				/pgfplots/major tick length=3pt,
			},
			ytick pos=left, 
			xtick pos=bottom,
			xtick=data,
			ymax=1,
			ymin=0,
			ytick={0.0,0.2,0.4,0.6,0.8,1},
			]
			\addplot+[ybar,gray!80,fill=gray!20] plot coordinates {(0,0.5) (1,0) (2,0.125) (3,0) (4,0.03125) (5,0) (6,0.0078125) (7,0) (8,0.001953125)};
			\pgfplotsset{cycle list shift=-1} 
			\addplot+[ybar] plot coordinates {(0,0) (1,0.25) (2,0) (3,0.0625) (4,0) (5,0.015625) (6,0) (7,0.00390625) (8,0)};
			
		\end{axis}
	\end{tikzpicture}
}

\newsavebox{\normDist}
\sbox{\normDist}{%
	\pgfplotsset{width=7cm,compat=1.8}
	\begin{tikzpicture}
		\hspace*{-.1cm}
		\begin{axis}[
			unit vector ratio*={1 6 1},
			ybar stacked,
			xticklabel=\empty,
			yticklabel=\empty,
			legend cell align={left}, 
			bar width=15pt,
			enlarge y limits={abs=0.05},
			xlabel={},
			ylabel={},
			xtick style={
				/pgfplots/major tick length=3pt,
			},
			ytick style={
				/pgfplots/major tick length=3pt,
			},
			ytick pos=left, 
			xtick pos=bottom,
			xtick=data,
			ymax=1,
			ymin=0,
			ytick={0.0,0.2,0.4,0.6,0.8,1},
			]
			\addplot+[ybar,gray!80,fill=gray!20,postaction={pattern=north east lines,pattern color=gray!80}] plot coordinates {(0,0.5) (1,0) (2,0.125) (3,0) (4,0.03125) (5,0) (6,0.0078125) (7,0) (8,0.001953125)};
			\pgfplotsset{cycle list shift=-1} 
			\addplot+[ybar] plot coordinates {(0,0) (1,0.25) (2,0) (3,0.0625) (4,0) (5,0.015625) (6,0) (7,0.00390625) (8,0)};
			\addplot+[ybar] plot coordinates {(0,0) (1,0.5) (2,0) (3,0.125) (4,0) (5,0.03125) (6,0) (7,0.0078125) (8,0)};
			
		\end{axis}
	\end{tikzpicture}
}

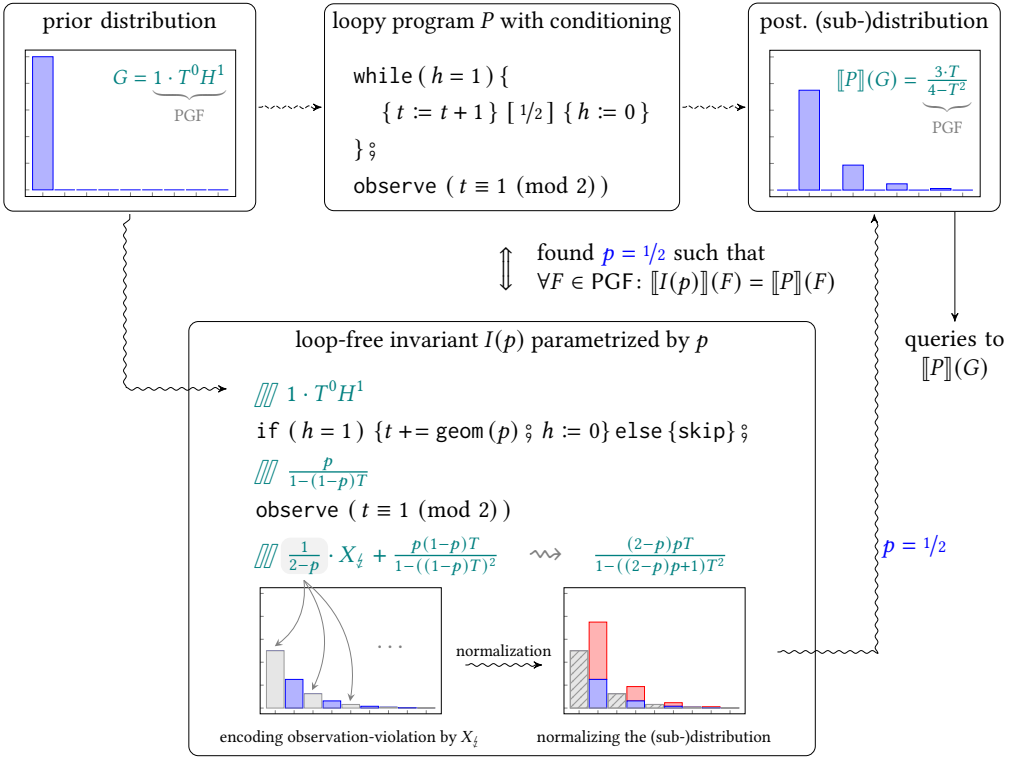
\begin{figure}[b]
	\centering
	\def\radius{.6mm} 
	\resizebox{1\linewidth}{!}{
		\begin{tikzpicture}[
			box/.style args = {#1/#2}{draw, align=flush center, rounded corners, text width=#1, minimum height=#2},
			box/.default = 1.4cm/10mm,
			unframedbox/.style args = {#1/#2}{align=flush center, text width=#1, minimum height=#2},
			unframedbox/.default = 1.4cm/10mm,
			foo/.style={%
				->,
				>=stealth',
				shorten >=1pt,
				shorten <=1pt,
				decorate,
				decoration={%
					snake,
					segment length=1.64mm,
					amplitude=0.2mm,
					pre length=2pt,
					post length=2pt,
				}
			},
			foodash/.style={%
				->,
				>=stealth',
				dash pattern=on 2pt off .75pt,
				shorten >=1pt,
				shorten <=1pt,
				decorate,
				decoration={%
					snake,
					segment length=1.64mm,
					amplitude=0.2mm,
					pre length=2pt,
					post length=2pt,
				}
			}
			]
			\tikzstyle{sarrow} = [draw, ->,>=stealth'],
			\tikzstyle{darrow} = [draw, <->,>=stealth']
			
			\begin{scope}
				\node[box=5cm/2.8cm] (program) at (0,0) {loopy program $P$ with conditioning\\[-.2cm]
					\begin{align*}
						& \WHILE{\progvar{h} = 1}\\
						& \quad \PCHOICE{\assign{\progvar{t}}{\progvar{t} + 1}}{\sfrac{1}{2}}{\assign{\progvar{h}}{0}}\\
						& \}\,\fatsemi\\
						& \observe{\progvar{t} \equiv 1 \!\!\!\!\pmod{2}}
					\end{align*}
				};
				\node[below right=.4cm and -2.8cm of program] (equivsymbol) {$\Big \Updownarrow$};
				\node[text width=5cm,align=left,right=.1cm of equivsymbol] (equiv) {$\textnormal{found}~\textcolor{blue}{p = \sfrac{1}{2}}~\textnormal{such that}$\\[-.1mm]$\forall F \in \pgf\colon \sem{I(p)}(F) = \sem{P}(F)$};
				\node[box=9cm/5.8cm,below=1.6cm of program] (invariant) {loop-free invariant $I(p)$ parametrized by $p$\\[-.2cm]
					\begin{align*}
						& \annotate{~1 \cdot T^0H^1}     \\
						& \IF{\progvar{h} = 1} \incrasgn{\progvar{t}}{\geometric{p}}\fatsemi\,\assign{\progvar{h}}{0} \ELSE \pskip\}\,\fatsemi\\
						& \annotate{~\tfrac{p}{1-(1-p) T}}\\
						& \observe{\progvar{t} \equiv 1 \!\!\!\!\pmod{2}}\\
						& \annotate{~\tikzmarkin[left offset={-0.08}, right offset={0.08}, above offset={0.35}]{vioprob}\tfrac{1}{2-p}\tikzmarkend{vioprob} \cdot X_\lightning + \tfrac{p(1-p)T}{1-((1-p) T)^2}} ~\quad\textcolor{gray}{\scaleto{\rightsquigarrow}{6pt}}~\quad \annocolor{\tfrac{(2-p)pT}{1-((2-p)p+1)T^2}}
					\end{align*}%
				
					\resizebox{1\linewidth}{!}{
						\begin{minipage}[t]{1\linewidth}
							\centering
							\usebox{\lightningDist}
							
							\LARGE{encoding observation-violation by $X_\lightning$}
						\end{minipage}
						\begin{minipage}[t]{1\linewidth}
							\centering
							\usebox{\normDist}
							
							\LARGE{normalizing the (sub-)distribution}
						\end{minipage}
					}
				};
				\node[box=3.3cm/3.01cm,left=1cm of program,text centered,text width=.25\textwidth,text depth=2.55cm] (priorDist) {prior distribution\\[.2cm]
					\resizebox{1\linewidth}{!}{
						 \usebox{\priorDist}
					}
				};
				\node[box=3.3cm/3.01cm,right=1cm of program,text centered,text width=.25\textwidth,text depth=2.55cm] (postDist) {post.~(sub-)distribution\\[.2cm]
					\resizebox{1\linewidth}{!}{
						\usebox{\postDist}
					}
				};
			
				\node[text width=2cm,align=center,below right=1.6cm and -1.8cm of postDist] (query) {queries to\\$\sem{P}(G)$};
				\draw[->,stealth'-] (query.north) to (query.north |- postDist.south);
			
				\node (ghostLightningDist) at (-.7,-8.12) {};
				\node (ghostNormDist) at (.75,-8.12) {};
				\draw[foo] (ghostLightningDist) to node[above,midway] {\scriptsize normalization} (ghostNormDist);
				\draw[foodash] (priorDist) to (program);
				\draw[foodash] (program) to (postDist);
				\path[draw,foo,stealth'-] (invariant.west)++(0.6,2.18) -| (priorDist.south);
				\path[draw,foo] (invariant.east)++(-0.6,-1.65) -| (postDist.south) node [right,pos=0.62] {\textcolor{blue}{$p = \sfrac{1}{2}$}};
				
				\node at (-4.9,0.2) (priorText) {\annocolor{\small$G = \grayunderbrace{1\cdot T^0 H^1}{\textnormal{PGF}}$}};
				\node at (5.92,0.1) (priorText) {\annocolor{\small$\sem{P}(G) = \grayunderbrace{\tfrac{3 \cdot T}{4-T^2}}{\textnormal{PGF}}$}};
			\end{scope}
%
%
		\end{tikzpicture}
		\begin{tikzpicture}[remember picture, overlay]
			\coordinate (viop) at (-10.80,2.56);
			\coordinate (t0prob) at (-11.23,1.6);
			\path[stealth-,gray] (t0prob) edge [bend right]  (viop);
			
			\coordinate (t1prob) at (-10.68,0.96);
			\path[stealth-,gray] (t1prob) edge [bend right]  (viop);
			
			\coordinate (t2prob) at (-10.12,0.80);
			\path[stealth-,gray] (t2prob) edge [bend right]  (viop);
			
			\node at (-9.5,1.56) {\textcolor{gray}{$\cdots$}};
		\end{tikzpicture}
	}
	\caption{A bird's-eye perspective of our approach.}
	\label{fig:overview}
\end{figure}


\paragraph{\bf Approach.}%
	\cref{fig:overview} sketches an overview of our inference approach:
	Given a prior distribution $G$ and a loopy probabilistic program $P$ with conditioning (at any place), our primary goal is to infer the posterior (sub-)distribution $\sem{P}(G)$ as depicted by the upper row. Here, we interpret $P$ as a distribution transformer $\sem{P}(\cdot)$ that transforms $G$ into $\sem{P}(G)$, both represented as PGFs to encode possibly infinite-support distributions. To deal with the unbounded \codify{while}-loop in $P$, we provide an invariant $I(p)$ in the form of a loop-free program parametrized by $p \in \R$, and aim to synthesize parameter values under which $I(p)$ is semantically equivalent to $P$, i.e., they transform every possible prior distribution into the same posterior (sub-)distribution. We show that checking the program equivalence $\sem{I(p)} = \sem{P}$ (together with parameter synthesis) is decidable -- via an extended technique of second-order PGFs -- when $I(p)$ and $P$ are restricted to a syntactic class of 
	programs called \credip preserving closed-form PGFs. Once the equivalence is concluded, we can simply push the prior distribution $G$ through the \emph{loop-free} program $I(p)$ -- as illustrated by the downward path in \Cref{fig:overview} -- and obtain the posterior (sub-)distribution $\sem{P}(G)$, from which various quantitative queries can be addressed. To tackle conditioning, a key technical ingredient in our approach is to extend PGFs with an extra term $X_\lightning$ keeping track of observation violations, which will eventually be normalized off to achieve the final, normalized (sub-)distribution.
	%

\paragraph{\bf Contributions.}
The main results of this paper are as follows.
%
\begin{itemize}
	\item
	We present a PGF-based denotational semantics for discrete probabilistic \WHILESYMBOL-programs where conditioning can occur at any place in the program. The basic technical ingredient is to extend PGFs with an extra term encoding the probability of violating observations as proposed by \cite{DBLP:conf/esop/BichselGV18}. The semantics can treat conditioning in the presence of possibly diverging loops and captures conditioning on zero-probability events.
	\item
	This semantics extends the PGF-based semantics of \cite{CAV22,LOPSTR} for unconditioned programs and is shown to coincide with the Markov chain semantics in \cite{DBLP:journals/toplas/OlmedoGJKKM18}. These correspondences indicate the adequacy of our semantics.
	\item
	Our PGF-based semantics readily enables exact inference for loop-free programs. We identify a syntactic class of almost-surely terminating programs for which exact inference for a \WHILESYMBOL-loop coincides with inference for a straight-line program. Technically this is based on proving program equivalence.
	\item
	We show that, for this class of programs, our approach can be generalized towards parameter synthesis: Are a \WHILESYMBOL-loop and a loop-free program that (both may) contain some parametric probability terms 
	equivalent for some values of these unknown probabilities?
	\item We implement our method in a tool called \toolname{Prodigy}; it augments existing computer algebra systems with GFs for (semi-)automatic inference and quantitative verification of conditioned probabilistic programs. We show that \toolname{Prodigy} can handle many infinite-state loopy programs and exhibits comparable performance to state-of-the-art exact inference tools over benchmarks of loop-free programs.
\end{itemize}
%

%
\paragraph{Paper structure.}
\Cref{sec:preliminaries} presents preliminaries on generating functions. 
\Cref{sec:semantics} presents our extended PGF-based denotational semantics that allows for exact quantitative reasoning about probabilistic programs with conditioning.
We dedicate \Cref{sec:loops} to the exact Bayesian inference for conditioned programs with loops leveraging the notions of invariants and equivalence checking.
In \Cref{sec:synthesis}, we identify the class of parametrized programs and invariants for which the problem of parameter synthesis is shown decidable.
We report the empirical evaluation 
of \toolname{Prodigy} in \Cref{sec:prodigy} and discuss the limitations of our approach in \Cref{sec:bottlenecks}. An extensive review of related work in probabilistic inference is given in \Cref{sec:related_work}. The paper is concluded in \Cref{sec:conclusion}. Additional background materials, elaborated proofs, and details on the examples can be found in the appendix.
	\section{Preliminaries on Generating Functions}
\label{sec:preliminaries}
Generating functions (GFs) constitute a versatile mathematical tool with extensive applications across various fields of mathematics and beyond \cite{generatingfunctionology}. They provide a systematic and elegant means of representing and manipulating sequences of numbers, rendering them essential for solving a diverse spectrum of mathematical problems in, e.g., enumerative combinatorics \cite{DBLP:books/daglib/0023751} and (discrete) probability theory \cite{johnson2005univariate}.

\paragraph{Formal power series.}
Generating functions, at their core, are \emph{formal power series} (FPSs), which encode essential information about possibly infinite sequences of numerical values (of any type).
The underlying principle is to represent the sequence as terms within an FPS (amenable to algebraic operations). Generating functions are classified as uni- or multivariate based upon the number of indeterminates. 
A \emph{univariate generating function} takes the form
\begin{equation}\label{eq:univariategf}
	F \eeq \sum\nolimits_{n\in\N} a_n X^n
\end{equation}
where $a_n$ is the $n$-th number within the sequence and $X$ is a \emph{formal indeterminate}.
The \enquote{monomials} $X^n$ are merely \emph{position-holders} for the coefficients $a_n$ and do not have any particular meaning.
However, \`{a} la \citet{LOPSTR} and \citet{Zaiser23}, we interpret the indeterminate $X$ with the corresponding program variable $\progvar{x}$ and the exponent $n$ with values of $\progvar{x}$; in this case, $a_n$ is the probability of $\progvar{x} = n$.

\begin{example}[Geometric Distribution as an FPS]
	\label{ex:univariategf}
	Consider a discrete random (program) variable $\progvar{t}$ which is geometrically distributed over $\N$ with parameter $\sfrac 1 2$.
	The probability mass function of $\progvar{t}$ is given by $P_\progvar{t}(\progvar{t}=n) = \sfrac 1 2^{n+1}$.
	We tabulate $P_\progvar{t}$ using a sequence $(a_n)_{n\in\N} = (P_\progvar{t}(\progvar{t}=n))_n = \nicefrac{1}{2}, \nicefrac{1}{4}, \nicefrac{1}{8}, \ldots$. Encoding this sequence as a generating function in terms of FPSs via formal indeterminate $T$ yields
	\begin{equation}
	\label{eqn:unvariate_series}
	\frac{1}{2} ~+~ \frac{1}{4} T ~+~ \frac{1}{8} T^2 ~+~ \frac{1}{16} T^3 ~+~ \frac{1}{32} T^4 ~+~ \frac{1}{64} T^5 ~+~ \frac{1}{128} T^6 ~+~ \frac{1}{256} T^7 ~+~ \cdots
	\end{equation}%
	where we uniquely associate terms of the power series to values of the sequence, e.g., the term $\tfrac 1 8 T^2$ encodes the information that the probability of $\progvar{t} = 2$ is $\sfrac 1 8$.
\end{example}

In order to deal with multiple program variables $\progvar{x}_1,\ldots,\progvar{x}_k$, the form in \cref{eq:univariategf} is generalized to a \emph{multivariate} generating function of dimension $k \in \N$ as $F = \sum\nolimits_{\bvec{n}\in\N^k} a_{\bvec{n}} {\bvec{X}}^{\bvec{n}}$, where $\bvec{X} = (X_1, X_2, \ldots, X_k)$ is a vector of indeterminates and ${\bvec{X}}^{\bvec{n}}$ is the monomial $X_1^{n_1} X_2^{n_2} \cdots X_k^{n_k}$. Here, the term $a_{\bvec{n}} {\bvec{X}}^{\bvec{n}}$ encodes that $(x_1, x_2, \ldots, x_k) = (n_1, n_2, \ldots, n_k)$ with probability $a_{\bvec{n}}$. A $k$-dimensional GF $F$ is called a \emph{probability generating function} (PGF) if $\sum_{\bvec{n} \in \N^k} a_{\bvec{n}} \leq 1$ and $a_{\bvec{n}} \geq 0$ for all $\bvec{n} \in \N^k$ (cf.\ \cref{eqn:unvariate_series}).
A PGF with $\sum_{\bvec{n} \in \N^k} a_{\bvec{n}} < 1$ represents a subprobability distribution and is called a \emph{sub-PGF}.

\paragraph{Closed forms.}
The encoding as in \cref{ex:univariategf} enables us to compress the infinite power series into a \emph{closed form} using Taylor's theorem, that is, a finitely-represented function whose Taylor series developed at zero coincides with the GF. For instance, the closed form of \Cref{eqn:unvariate_series} is given by $T \mapsto \sfrac{1}{(2-T)}$ for all $\abs{T} < 2$, as the Taylor series of $\sfrac{1}{(2-T)}$ is precisely $\tfrac{1}{2} + \tfrac{1}{4} T + \tfrac{1}{8} T^2 +\cdots$. Many important operations on infinite sequences of numbers -- and their corresponding GF series -- can be simulated by manipulating the closed-form expression instead. Using algebraic operations, this allows for computing, e.g., expected values, variances, higher-order moments, point probabilities, and tail bounds. For instance, the formal derivative $\frac{\mathrm{d}}{\mathrm{d}T}\tfrac{1}{2-T} = \tfrac{1}{(2-T)^2}$ evaluated at $T=1$ yields the expected value $\mathbb{E}(\progvar{t})=\frac{1}{(2-1)^2}=1$. \cref{fig:cheatSheet} summarizes some basic operations on GFs and their corresponding effects on the infinite sequences.

To effectively manipulate closed forms, we embed them in an algebraic structure -- the (commutative) \emph{ring of FPSs} $(\R[[\bvec{X}]], +, \cdot, 0, 1)$. Here, $\R[[\bvec{X}]]$ is the set of FPSs (of fixed dimension $k$):
\[
	F \eeq \sum\nolimits_{\bvec{n}\in\N^k} \coef{\bvec{n}}{F} {\bvec{X}}^{\bvec{n}}
\]
with $\coef{\cdot}{F}\colon \N^k \! \to \R$, \enquote{$+$} (addition) and \enquote{$\cdot$} (multiplication) are binary operations defined as
\begin{align*}
	F+G \ddefeq \sum\nolimits_{\bvec{n}\in \N^k} \left(\coef{\bvec{n}}{F} + \coef{\bvec{n}}{G}\right)\bvec{X}^\bvec{n} \quad \text{and}\quad
	F \cdot G \ddefeq \sum\nolimits_{\bvec{n}_1, \bvec{n}_2\in \N^k} \left(\coef{\bvec{n}_1}{F} \cdot \coef{\bvec{n}_2}{G}\right)\bvec{X}^{\bvec{n}_1 + \bvec{n}_2}~,
\end{align*}%
and $0, 1 \in \R[[\bvec{X}]]$ are neutral elements w.r.t.\ addition and multiplication, respectively. 
The multiplication $F \cdot G$ is in fact the discrete convolution of the two sequences $F$ and $G$ (aka, the \emph{Cauchy product} of power series). Note that $F \cdot G$ is always well-defined because for all $\bvec{n} \in \N^k$ there are \emph{finitely} many $\bvec{n}_1 + \bvec{n}_2 = \bvec{n}$ in $\N^k$. Moreover, every $F \in \R[[\bvec{X}]]$ has an \emph{additive inverse} $-F \in \R[[\bvec{X}]]$ yet \emph{multiplicative inverses} $F^{-1} = 1/F$ need not always exist. 

\begin{table}[t]
	\centering
	\caption{
		GF cheat sheet.
		$f, g$ and $X, Y$ are arbitrary GFs and indeterminates, resp.~\cite{CAV22}.
	}
	\label{fig:cheatSheet}
	\begin{adjustbox}{max width=1.0\textwidth}
		\renewcommand{\arraystretch}{1.6}
		\setlength{\tabcolsep}{5pt}
		\hspace*{-.28cm}
		\begin{tabular}{l l l}
			\toprule
			Operation & Effect& Example\\ \midrule
			$f^{-1} = 1/f$ & \makecell[l]{multiplicative inverse of $f$ \\[-2pt] (if it exists)} & \makecell[l]{$\frac{1}{1 - XY} = 1 + XY + X^2Y^2 + \cdots$ \\[-1pt] because $(1- XY)(1 + XY + X^2Y^2 + \cdots) = 1$} \\ 
			$f\cdot X$ & shift in dimension $X$ & $\frac{X}{1 - XY} = X + X^2Y + X^3Y^2 + \cdots$ \\ 
			$\subsFPSVarFor{f}{X}{0}$ & drop terms containing $X$ & $\frac{1}{1 - 0Y} = 1$\\ 
			$\subsFPSVarFor{f}{X}{1}$ & projection\footnote{%
				Projection is not always well-defined, e.g., $\subsFPSVarFor{\frac{1}{1-X+Y}}{X}{1} = \frac{1}{Y}$ is ill-defined, as $Y$ is not invertible. It is, however,
				well-defined whenever used in this paper; in particular, projection is well-defined for (fully simplified) rational closed forms of PGFs.
			} on $Y$ & $\frac{1}{1 - 1Y} = 1 + Y + Y^2 + \cdots$ \\ 
			$f\cdot g$ & \makecell[l]{discrete convolution \\[-2pt] (or Cauchy product)}& $\frac{1}{(1 - XY)^2} = 1 + 2XY + 3X^2Y^2 + \cdots$ \\ 
			$\partial_X f$ & formal derivative in $X$ & $\partial_X \frac{1}{1 - XY} = \frac{Y}{(1-XY)^2} = Y +2XY^2 + 3X^2Y^3 + \cdots$ \\ 
			$f + g$ & coefficient-wise sum & $\frac{1}{1 - XY} + \frac{1}{(1 - XY)^2} = \frac{2-XY}{(1-XY)^2} = 2 + 3XY + 4 X^2Y^2 + \cdots$ \\ 
			$a\cdot f$ & \makecell[l]{coefficient-wise scaling\\[-2pt] (by scalar $a$)} & $\frac{7}{(1 - XY)^2} = 7 + 14XY + 21 X^2Y^2 + \cdots$ \\
			%
			%
			\bottomrule
		\end{tabular}
	\end{adjustbox}%
\end{table}

\begin{remark}
		Treating the closed form as a \emph{function}, say $T \mapsto \tfrac{1}{2-T}$, and computing its Taylor series imposes -- for the sake of well-definedness -- the radius of convergence of the resulting series, i.e., $\abs{T} < 2$. However, due to the underlying algebraic structure, we can safely write $\tfrac{1}{2-T} = \sum_{n \in \N} \tfrac{1}{2^{n+1}}T^n$ regardless of the fact whether $\abs{T} < 2$: the sequences $2 - 1T + 0 T^2 + \cdots$ and $\tfrac{1}{2} + \tfrac{1}{4} T + \tfrac{1}{8} T^2 +\cdots$ are \emph{multiplicative inverse elements} to each other in $\R[[T]]$, i.e., their product is $1$. We refer interested readers to \cite[Appendix D]{DBLP:journals/corr/abs-2205-01449} for more details on convergence-related issues.
	\qedT
\end{remark}

In this paper, we are mainly concerned with \emph{rational closed forms}, i.e., FPSs of the form $F = GH^{-1} = G/H$ where $G, H$ are \emph{polynomials} in $\R[[\bvec{X}]]$ (i.e., $G,H$ have finitely many non-zero coefficients).

%
%
%
%

	\section{Generating Function Semantics}
\label{sec:semantics}
\subsection{Semantics without Conditioning}
\label{ssec:nocond_semantics}

Given a fixed input, the semantics of a probabilistic program is captured by its (posterior) probability distribution over the final (terminating) program states. 
In \cite{LOPSTR}, the domain of discrete distributions is represented in terms of PGFs -- elements from $\R[[\bvec{X}]]$ 
-- and a (conditioning-free) program is interpreted denotationally as a \emph{distribution transformer} \`{a} la Kozen~\cite{Kozen}.
We recap this semantics for programs without conditioning by means of an example:

\begin{example}[PGF Semantics without Conditioning]\label{ex:pgcl}%
		Consider Prog.~\ref{prog:pgcl_intro_annotated} with input $G=1\cdot X^0C^0$, representing the joint prior distribution $\textup{Pr}\left(\progvar{x} = 0 \wedge \progvar{c} = 0\right) = 1$.
		The%
	{\makeatletter
		\let\par\@@par
		\par\parshape0
		\everypar{}
	\begin{wrapfigure}{r}{4cm}
		\centering
		\begin{minipage}{\linewidth}
			\vspace{-1.1cm}
			\begin{align*}
			& \annotate{1 \qquad\quad\quad\ \ \, \left(= 1 \cdot X^0C^0\right)}                                                                         \\
			& \ASSIGN{\progvar{x}}{1}                                                              \\
			& \annotate{X}                                                                         \\
			& \PCHOICE{\ASSIGN{\progvar{c}}{\progvar{c} + 5}}{\sfrac 1 3}{\ASSIGN{\progvar{c}}{3}} \\
			& \annotate{\sfrac{1}{3} \cdot XC^5 + \sfrac{2}{3} \cdot XC^3}                                       \\
			& \IF{\progvar{c} > 4}                                                                 \\
			& \quad \annotate{\sfrac{1}{3} \cdot XC^5}                                                    \\
			& \quad \ASSIGN{\progvar{x}}{\progvar{x} + \progvar{c}}                                \\
			& \quad \annotate{\sfrac{1}{3} \cdot X^6C^5}                                                  \\
			& \ELSE                                                                                \\
			& \quad \annotate{\sfrac{2}{3} \cdot XC^3}                                                    \\
			& \quad \decr{\progvar{x}}                                                             \\
			& \quad \annotate{\sfrac{2}{3} \cdot C^3}                                                     \\
			& \}                                                                                   \\
			& \annotate{\sfrac{1}{3} \cdot X^6 C^5 + \sfrac{2}{3} \cdot C^3} 
			\end{align*}
		\end{minipage}
			\captionof{program}{PGF semantics for an observation-free program.}
		\label{prog:pgcl_intro_annotated}
		\vspace*{-6mm}
	\end{wrapfigure}
	\noindent%
	 denotational PGF semantics of this program is computed in a \emph{forward} manner per the annotation style in \cite{kaminski2019advanced}. Below, we show step-by-step how the prior distribution $G$ is transformed into the joint posterior distribution $G' = \sfrac{1}{3} \cdot X^6 C^5 + \sfrac{2}{3} \cdot C^3$, indicating that $\textup{Pr}\left(\progvar{x} = 6 \wedge \progvar{c} = 5\right) = \sfrac{1}{3}$ and $\textup{Pr}\left(\progvar{x} = 0 \wedge \progvar{c} = 3\right) = \sfrac{2}{3}$.
	We start by interpreting the first instruction of the program, i.e., the assignment of 1 to variable $\progvar{x}$, which results in the intermediate distribution $1\cdot X^1$.
	Then, we descend into the left and right branches of the probabilistic choice statement.
	For the left branch, we interpret the semantics of $\ASSIGN{\progvar{c}}{\progvar{c} + 5}$ by multiplying the previous distribution $1\cdot X^1C^0$ with $C^5$ which encodes the effect of increasing $\progvar{c}$ by $5$.
	The right branch is handled analogously by setting $\progvar{c}$ to $3$;
	this is done by first marginalizing the distribution $1\cdot X^1C^0$ by substituting $1$ for $C$ and then multiplying the result with $C^3$.
	Now, we can combine the semantics for the two branches (left: $XC^5$; right: $XC^3$) via a weighted sum $\sfrac 1 3 \cdot  XC^5 + \sfrac 2 3 \cdot XC^3$ to represent the distribution after executing the probabilistic choice.
	Subsequently, we evaluate the conditional branching by recursively descending into the satisfying branch ($\ASSIGN{\progvar{x}}{\progvar{x} + \progvar{c}}$) and non-satisfying branch ($\decr{\progvar{x}}$) with their respective filtered inputs ($\progvar{c}>4$: $\sfrac{1}{3}\cdot XC^5$; $\progvar{c}\leq 4$: $\sfrac{2}{3}\cdot XC^3$).
	Finally, we combine the two sub-results of the conditional branches and thus
	\par}
\noindent
obtain $\sfrac{1}{3}\cdot  X^6 C^5 + \sfrac{2}{3}\cdot C^3$. See \cite{LOPSTR} for semantics of more program constructs.
\qedT
\end{example}

The representation of a posterior distribution in terms of a generating function comes with several benefits:
\begin{enumerate*}[label=(\arabic*)]
	\item it naturally encodes common, \emph{infinite-support} distributions 
	like the geometric or Poisson distribution in compact, \emph{closed-form} representations;
	\item it allows for compositional reasoning and, in particular, 
	in contrast to representations in terms of density or mass functions, the effective computation of (high-order) moments;
	\item tail bounds, concentration bounds, and other properties of interest can be extracted with relative ease from a PGF; and
	\item expressions containing parameters 
	are naturally supported.
\end{enumerate*}

\subsection{Semantics with Conditioning}
\label{ssec:conditioned_semantics}
We lift the approach to discrete, loopy probabilistic programs \emph{with conditioning} by extending the PGF semantics of \citet{LOPSTR} to cope with posterior observations. To define such a semantic model, we fix $k$ $\N$-valued program variables $\progvar{x}_1, \progvar{x}_2, \ldots , \progvar{x}_k$. The set of program state valuations is $\N^k$; for each \mbox{$\sigma = (\sigma_1, \ldots, \sigma_k) \in \N^k$}, $\sigma_i$ indicates the value of $\progvar{x}_i$.
%
We consider the \pgcl programming language \cite{DBLP:series/mcs/McIverM05} with the extended ability to specify posterior observations via the \codify{observe} statements \cite{ACM:conf/fose/Gordon14,DBLP:journals/toplas/OlmedoGJKKM18,DBLP:conf/aaai/NoriHRS14}:
\begin{definition}[\cpgcl]
	\label{def:cpgcl}%
	A program $P$ in the \emph{conditional probabilistic guarded command language (\cpgcl)} adheres to the grammar
	\begin{align*}
		P~~\!\Coloneqq~~& \pskip \mmid \assign{\progvar{x}}{E} \mmid \COMPOSE{P}{P} \mmid \!{\pchoice{P}{p}{P}}\! \mmid \observe{B} \!\mmid \\
		  &\ITE{B}{P}{P} \!\mmid \WHILEDO{B}{P}
	\end{align*}%
	where 
	$E\colon \N^k \to \N$ is an arithmetic expression, 
	$B \subseteq \N^k$ is a predicate, and $p \in [0,1]$.\footnote{We do not give an explicit syntax for $E$ and $B$ as it is irrelevant at this point. When dealing with \emph{automation}, we present a concrete syntax, cf.\ \cref{tab:redip} on page \pageref{tab:redip}.}
\end{definition}
The meaning of most \cpgcl program constructs is standard. The \emph{probabilistic choice} $\{P\}\,[\,p\,]\,\{Q\}$ executes $P$ with probability $p \in [0,1]$ and $Q$ with probability $1-p$. The \emph{conditioning statement} $\codify{observe}(B)$ \enquote{blocks} all program runs that violate the guard $B$ and normalizes the probabilities of the remaining runs. For example, in Prog.\ \ref{prog:telephone} on page \pageref{prog:telephone}, the telephone operator observes 5 calls in the last hour as indicated by $\observe{\progvar{c} = 5}$. To reflect this, all program states where $\progvar{c} \neq 5$ are assigned probability zero. The program's distribution is adjusted by normalizing the probability of runs satisfying $\progvar{c} =5$ by the total probability mass of all runs violating this condition.

To identify program runs violating the observations, we extend the domain of FPSs -- and thus the domain of PGFs -- with a dedicated indeterminate $X_\lightning$\! aggregating observation-violation probability:
\begin{definition}[\efps and \epgf]\label{def:efps}%
	Let $\bvec{X}$ and $X_\lightning$\! be indeterminates. For any program state valuations $\sigma \in \N^k$, 
	an \emph{extended} formal power series (eFPS) is of the form\footnote{The coefficients $\coef{\cdot}{F}$ range over $\R_{\geq 0}^\infty$ to enforce a complete lattice structure over $\efps$; see details in \cref{apx:domaintheory,apx:semantics}.}
	\[
		F \eeq \coef{\lightning}{F} X_\lightning + \sum\nolimits_{\sigma \in \N^k} \coef{\sigma}{F}\bvec{X}^\sigma \quad \textnormal{with} \quad \coef{\cdot}{F}\colon\, \N^k \cup \{\lightning\} \to \R_{\geq 0}~.
	\]
	We refer to $\coef{\lightning}{F} X_\lightning$\! as the \emph{observation-violation term} and call the set of all extended formal power series \efps.
	Let $\abs{F} \defeq \sum_{\sigma \in \N^k} \coef{\sigma}{F}$ denote the \emph{mass} of $F$. $F \in \efps$ is an \emph{extended PGF (ePGF)} iff $\abs{F} \leq 1$; in this case, $F$ encodes a (sub)probability distribution. Let $\epgf$ be the set of all \textnormal{ePGFs}.
	An \emph{ePGF transformer} is a function $\epgf \to \epgf$.
\end{definition}
We emphasize that $\abs{F}$ does not take the \codify{observe}-violation probability $\coef{\lightning}{F}$ into account. Another way to obtain $\abs{F}$ is through the substitution of indeterminates $\bvec{X}$ representing program variables by $\bvec{1}$ and the indeterminate $X_\lightning$ for the observation-violation by $0$.
Addition and scalar multiplication in \efps are to be understood coefficient-wise, that is, for any $F, G \in \efps$, 
\begin{align*}
	F+G &\ddefeq \left(\coef{\lightning}{F} + \coef{\lightning}{G}\right)X_\lightning + \sum\nolimits_{\sigma \in \N^k} \left(\coef{\sigma}{F} + \coef{\sigma}{G}\right)\bvec{X}^\sigma~,\\
	a \cdot F &\ddefeq \left(a \coef{\lightning}{F} X_\lightning\right) + \sum\nolimits_{\sigma \in \N^k}\left(a \coef{\sigma}{F}\right)\bvec{X}^\sigma \quad \textnormal{for} \quad a \in \R_{\geq 0}~.
\end{align*}

\begin{remark}
\efps is not closed under multiplication: $X_\lightning\! \cdot X_\lightning\! = X_\lightning^2 \not\in \efps$.
This is intended, as such monomial combinations do not have a valid interpretation in terms of probability distributions.
\qedT
\end{remark}

We endow ePGFs with the following ordering relations.
\begin{definition}[Orders over \textup{\epgf}]
	For all $F,G \in \epgf$, let
	\begin{align*}
	F \ppreceq G \qquad\ \textnormal{iff}\ \qquad \forall \sigma \in \N^k \cup \{\lightning\}.\ \  \coef{\sigma}{F} \lleq \coef{\sigma}{G}~.
	\end{align*}
	This order can be lifted to ePGF transformers, that is, for all $\phi, \psi \in (\epgf \to \epgf)$,
	\begin{align*}
	\phi \ssqsubseteq \psi \qquad\ \textnormal{iff}\ \qquad \forall F \in \efps.\ \  \phi(F) \ppreceq \psi(F)~.
	\end{align*}%
\end{definition}%
\noindent
In fact, $(\epgf, \preceq)$ and $(\epgf \to \epgf, \sqsubseteq)$ are $\omega$-complete partial orders (cf.\ \Cref{apx:semantics}).
To evaluate Boolean guards, we use the so-called \emph{filtering} function for eFPSs. The filtering of $F\in \efps$ by predicate $B$ is
	\begin{align*}
		\constrain{F}{B} \ddefeq \sum\nolimits_{\sigma \models B} \coef{\sigma}{F} \bvec{X}^\sigma~,
	\end{align*}%
i.e., $\constrain{F}{B}$ is the eFPS derived from $F$ by setting $\coef{\lightning}{F}$ and all $\coef{\sigma}{F}$ with $\sigma \not\models B$ to $0$.
In contrast to \cite{LOPSTR}, we cannot decompose $F$ into $\constrain{F}{B} + \constrain{F}{\neg B}$, but rather have to include the observation-violation term separately, yielding $F = \constrain{F}{B} + \constrain{F}{\neg B} + \coef{\lightning}{F}X_\lightning$.
Further properties of the \epgf domain are found in \Cref{apx:semantics}.

\subsection{Non-Normalized Semantics for \cpgcl}
Let $\sem{P}\colon \epgf \to \epgf$ be a (non-normalized) distribution transformer for \cpgcl program $P$. 
We define the \emph{non-normalized} semantics of $P$ by transforming an input ePGF $G$ to an output ePGF $\sem{P}(G)$ while \emph{explicitly} keeping track of the probability of violating the observations; see \Cref{tab:semanics}.

\renewcommand{\arraystretch}{1.2}
\begin{table}[t]
	\caption{The \emph{non-normalized} semantics for \cpgcl programs.
	}
	\label{tab:semanics}
	\centering
	\begin{tabular}{lcl}
		\toprule
		$P$ & &$\sem{P}(G)$ \\
		\midrule
		$\pskip$ & &$G$\\
		$\assign{\progvar{x}_i}{E}$ & &$\coef{\lightning}{G}X_{\lightning} + \sum_{\sigma} \coef{\sigma}{G}X_1^{\sigma_1}\cdots X_i^{E(\sigma)}\cdots X_k^{\sigma_k}$ \\
		\observe{B} & &$\left(\coef{\lightning}{G} + \abs{\constrain{G}{\neg B}}\right) X_\lightning + \constrain{G}{B}$\\[.1cm]
		$\pchoice{P_1}{p}{P_2}$ & &$p \cdot \sem{P_1}\left(G\right) + \left(1-p\right) \cdot \sem{P_2}\left(G\right)$\\
		$\ITE{B}{P_1}{P_2}$ & &$\coef{\lightning}{G} X_\lightning + \sem{P_1}\left(\constrain{G}{B}\right) + \sem{P_2}\left(\constrain{G}{\neg B}\right)$\\
		$\COMPOSE{P_1}{P_2}$ & &$\sem{P_2}\left(\sem{P_1}\left(G\right)\right)$\\[.1cm]
		$\WHILEDO{B}{P_1}$ & &$[\lfp~ \Phi_{B,P_1}]\left(G\right),$ where \\
		& &$\Phi_{B,P_1}(f) \eeq \lambda G.~ \coef{\lightning}{G}X_{\lightning} + \constrain{G}{\neg B} + f\left(\sem{P_1}\left(\constrain{G}{B}\right)\right)$\\
		\bottomrule
	\end{tabular}
\end{table}
\renewcommand{\arraystretch}{1}

The $\pskip$ statement
leaves the initial distribution $G$ unchanged, i.e., it \emph{skips} an instruction. 
The assignment $\ASSIGN{\progvar{x}_i}{E}$ 
updates the exponent of the corresponding indeterminate $X_i$ in every term of the ePGF by $\eval{\sigma}{}$ and the observation-violation term remains unchanged.
For instance, given $E=2\cdot \progvar{x}\progvar{y}^3 + 23$ and state valuation $\sigma = (x,y) = (1, 10)$, 
$\ASSIGN{\progvar{x}_i}{E}$ updatess the term $a X Y^{10}$ to $a X^{2023} Y^{10}$.
The semantics for $\codify{observe}(B)$
is defined in line with \cite{DBLP:conf/aaai/NoriHRS14,DBLP:journals/toplas/OlmedoGJKKM18,DBLP:journals/pacmpl/Jacobs21,DBLP:conf/esop/BichselGV18} as \emph{rejection sampling}, i.e., if the current program run satisfies $B$, it behaves like a $\pskip$ statement and the posterior distribution is unchanged; If the current run, however, violates the condition $B$, the run is rejected and the program restarts from the top in a reinitialized state. Hence, observing a certain guard $B$ just filters the prior distribution and accumulates the probability mass that violates the guard. For example, observing an even dice roll $\observe{\progvar{x} \equiv_2 0}$ out of a six-sided die $\frac 1 6 \left(X + X^2 + X^3 + X^4 + X^5 + X^6\right)$ yields $\frac 1 6 \left(X^2 + X^4 + X^6\right) + \frac 1 2 X_\lightning$.
The probabilistic branching statement $\PCHOICE{P_1}{p}{P_2}$
is interpreted as the convex $p$-weighted combination of the two subprograms $P_1$ and $P_2$.
The semantics of conditional branching $\ITE{B}{P_1}{P_2}$ combines the semantics of $P_1$ and $P_2$ conditionally based on $B$.
Sequential composition $P_1\fatsemi P_2$
composes programs in a \emph{forward} manner, i.e., we first evaluate $P_1$ and take the intermediate result as new input for $P_2$.
The semantics of a loop $\WHILEDO{B}{P_1}$ is defined as the \emph{least fixed point} ($\lfp$) of $\Phi_{B,P_1}$ (see domain theory in \cref{apx:domaintheory}). Here, $\Phi_{B,P_1}$ is known as the \emph{characteristic function} -- a monotonic operator mimicking the effect of unfolding the loop. Concretely, $\Phi_{B,P_1}$ guarantees the equivalence of $\WHILEDO{B}{P_1}$ and $\ITE{B}{P_1\fatsemi \WHILEDO{B}{P_1}}{\pskip}$. 

Note that the \codify{observe}-violation term $\coef{\lightning}{G}X_{\lightning}$ \enquote{passes through} all instructions but $\observe{B}$:
\begin{lemma}[Error Term Pass-Through] \label{lem:errpassthru}
	For every program $P$ and every $F\in\epgf$,
	\[
	\sem{P}(F) \eeq \sem{P}\left(\sum\nolimits_{\sigma\in\N^k} \coef{\sigma}{F} \mathbf{X}^\sigma + \coef{\lightning}{F} X_\lightning\right)
	\eeq
	\sem{P}\left(\sum\nolimits_{\sigma\in\N^k} \coef{\sigma}{F} \mathbf{X}^\sigma\right) + \coef{\lightning}{F} X_\lightning~.
	\]
\end{lemma}
\noindent
This renders the semantics as a conservative extension to \cite{CAV22}, as for \codify{observe}-free programs on initial distributions without $\coef{\lightning}{G}X_{\lightning}$, both semantics coincide.

Recall that in Prog.~\ref{prog:div_and_obsfail} on page \pageref{prog:div_and_obsfail}, all program runs which eventually would terminate violate the observation.
Since the (unnormalized) probability of non-termination is zero (as there is only a single infinite run), the final \emph{non-normalized, conditioned} ePGF semantics of this program is $1\cdot X_\lightning$.

\subsection{Normalized Semantics for \cpgcl}\label{subsec:conditioned-semantics}

The non-normalized semantics serves as an intermediate result to achieve our \emph{normalized} semantics, which further addresses \emph{normalization} of distributions. 
\begin{definition}[Normalization]\label{def:normalization}%
	The \emph{normalization operator} $\normalize$ is a partial function defined as\footnote{$\normalize$ in fact maps an ePGF to a PGF, i.e., $\coef{\lightning}{F} X_\lightning$ is pruned away by normalization.}
	\[
		\normalize\colon~ \epgf \ppto \epgf, \qquad F \mmapsto
		\begin{cases}
			\frac{\constrain{F}{\TRUE}}{1-\coef{\lightning}{F}}& \textnormal{if } \coef{\lightning}{F} < 1\,,\\
			\textnormal{undefined}& \textnormal{otherwise}\,.
		\end{cases}
	\]
\end{definition}%
\noindent Intuitively, normalizing an ePGF amounts to \enquote{distributing} the probability mass $\coef{\lightning}{F}$ pertaining to observation violations 
over its remaining (valid) program runs. We lift the operator and denote the \emph{normalized} semantics of $P$ by
\begin{align*}
	\normalize\left(\sem{P}\right) \ddefeq \lambda G.~ \normalize\left(\sem{P}(G)\right) \eeq  \lambda G.~ \frac{\constrain{\sem{P}(G)}{\TRUE}}{1-\coef{\lightning}{\sem{P}(G)}}~, \qquad \text{provided}~ \coef{\lightning}{\sem{P}_G} < 1.
\end{align*}%

\begin{remark}
	In contrast to the non-normalized semantics, the normalized semantics might not always be defined: Reconsider Prog.~\ref{prog:div_and_obsfail} for which the non-normalized semantics is $1 \cdot X_\lightning$; normalizing the semantics is not possible as it would lead to $\frac{\constrain{1\cdot X_\lightning}{\TRUE}}{1-\coef{\lightning}{F}} = \frac{0}{0}$, i.e., an undefined expression.
	This phenomenon can only be caused by \codify{observe}-violations but never by non-terminating behaviors. The following two programs reveal the difference between non-termination and \codify{observe} violation: $\PCHOICE{\ASSIGN{\progvar{x}}{1}}{\sfrac 1 2}{\observe{\FALSE}}$ has a normalized semantics of $1\cdot X^1$, whereas the normalized semantics for $\PCHOICE{\ASSIGN{\progvar{x}}{1}}{\sfrac 1 2}{\DIVERGE}$\footnote{$\DIVERGE$ is syntactic sugar for $\WHILEDO{\TRUE}{\pskip}$.} is $\frac{1}{2}\cdot X^1$.
	\qedT
\end{remark}


\begin{example}[Telephone Operator]%
	\label{ex:telephone}%
	Reconsider Prog.\ \ref{prog:telephone}, the  loop-free program generating an infinite-support distribution. It describes a telephone operator who lacks knowledge about whether today is a weekday or weekend.
	The operator's initial belief is that there is a probability of $\sfrac 5 7$ of it
	{\makeatletter
	\let\par\@@par
	\par\parshape0
	\everypar{}
	\begin{wrapfigure}{r}{0.36\linewidth}
		\vspace*{-4mm}
		\begin{minipage}{1\linewidth} 
		\begin{align*}
			& \annotate{1\qquad\quad\quad~~\,\left(=1 \cdot W^0C^0 + 0\cdot X_\lightning\right)}\\
			& \pchoice{\assign{\progvar{w}}{0}}{\sfrac{5}{7}}{\assign{\progvar{w}}{1}}\,\fatsemi \\
			& \annotate{\tfrac{5}{7}W^0 + \tfrac{2}{7}W^1}\\
			& \IF{\progvar{w} = 0}\\
			& \quad \annotate{\tfrac{5}{7}}\\
			& \quad \assign{c}{\poisson{6}}\\
			& \quad \annotate{\tfrac{5}{7}e^{-6 (1-C)} }\\
			& \ELSE\\
			& \quad \annotate{\tfrac{2}{7}W}\\
			& \quad \assign{c}{\poisson{2}}\\
			& \quad \annotate{\tfrac{2}{7} e^{-2 (1-C)}W}\\
			& \}\,\fatsemi \\
			& \annotate{\tfrac{5}{7}e^{-6 (1-C)} + \tfrac{2}{7} e^{-2 (1-C)}W}\\
			& \observe{\progvar{c} = 5}\\
			& \annotate{\tfrac{(4860 + 8e^4W)}{105e^6} C^5 + (1-\tfrac{4860 + 8e^4}{105e^6})X_\lightning} 
		\end{align*}%
		\end{minipage}
		\captionof{program}{Semantics for the tel.\ operator.}
		\label{prog:telephone-PGF}
	\end{wrapfigure}
	\noindent
	being a weekday ($\progvar{w}=0$) and a $\sfrac 2 7$ probability of it being a weekend ($\progvar{w}=1$). Typically, on weekdays, there are an average of 6 incoming calls per hour, while on weekends, this rate decreases to 2 calls.
	 Both rates are governed by a Poisson distribution.
	 The operator has observed 5 calls in the past hour, and the objective is to determine the updated distribution of the initial belief based on this posterior observation.
	We start the computation with prior distribution (ePGF) $1$, which initializes every program variable to 0 with probability 1. For the assignments to $c$ we use the closed-form PGF for a Poisson distribution with parameter $\lambda$, which is given by
		$\sum_{k \in \N_0}{\frac{\lambda^{k} e^{-\lambda}}{k!} C^k} = e^{-\lambda} \sum_{k \in \N_0}{\frac{(\lambda C)^{k}}{k!}} =  e^{-\lambda} e^{\lambda C} = e^{-\lambda (1-C)}.$
	By computing the transformations forward in sequence for each program instruction (see Prog.~\ref{prog:telephone-PGF}), we obtain the \emph{non-normalized} semantics:
	
	\medskip
	\begin{minipage}{1\linewidth}
		\[
			\sem{P}(G) \eeq \tfrac{(4860 + 8e^4W)}{105e^6} C^5 + (1-\tfrac{4860 + 8e^4}{105e^6})X_\lightning~.
		\]
	\end{minipage}
	\medskip
	
	\noindent
	Normalizing this yields
	
	\medskip
	\begin{minipage}{1\linewidth}
		\[
			\normalize\left(\sem{P}(G)\right) \eeq \frac{(1215 e^{-4}+2W)C^5}{2+1215 e^{-4}}~.\qedT\ \ \ \ 
		\]
	 \end{minipage}
	\par}%
\end{example}

\bigskip
Notably, the semantics in \Cref{tab:semanics} coincides with an operationally modeled semantics using \emph{countably infinite Markov chains} \cite{DBLP:journals/toplas/OlmedoGJKKM18} -- which in turn, for universally almost-surely terminating programs\footnote{Programs that terminate with probability 1 on all inputs; see \cref{ssec:invariants}.} is equivalent to the interpretation of Microsoft's probabilistic programming language R2 \cite{DBLP:conf/aaai/NoriHRS14}.
A Markov chain describing the semantics of a \cpgcl program consists of three ingredients: (1) the state space $\mathcal{S}$, (2) the initial state $\angles{P, \sigma}$, and (3) a transition matrix $\mathcal{P}\colon \mathcal{S} \times \mathcal{S}$. The states are pairs of the form $\angles{P, \sigma}$. Here, $P$ denotes the program left to be executed (with $\downarrow$ indicating the terminated program) and $\sigma$ the current state valuation. We use the dedicated state $\angles{\lightning}$ for denoting that some observe violations have occurred during the run of a program. The detailed construction of the Markov chain $\mathcal{R}_\sigma \sem{P}$ from a \cpgcl program $P$ with initial state valuation $\sigma$ is given in \Cref{apx:semantics}.
Regarding the equivalence between the two semantics, we are interested in the reachability probability of eventually reaching state $\langle\downarrow, \sigma\rangle$ conditioned to never visiting the \codify{observe}-violation state $\langle\lightning\rangle$.


\begin{theorem}[Equivalence of Semantics]%
	\label{thm:operational_equivalence}%
	For every \cpgcl program $P$, let $\mathcal{R}_\sigma\sem{P}$ be the Markov chain of\, $P$ starting with state valuation $\sigma \in\N^k$. Then, for any $\sigma' \in \N^k$,
	\begin{align}\label{eq:sem-equiv}
	\textup{Pr}^{\mathcal{R}_\sigma\sem{P}}\left(\diamondsuit\langle\downarrow,\sigma'\rangle\mid\neg\diamondsuit\angles{\lightning}\right) \!\eeq \coef{\sigma'}{\normalize\left(\sem{P}(\bvec{X}^\sigma)\right)}~,
	\end{align}%
	where the left term denotes the probability of eventually reaching the terminating state $\langle\downarrow,\sigma'\rangle$ in $\mathcal{R}_\sigma\sem{P}$ conditioned on avoiding the \codify{observe}-failure state $\angles{\lightning}$.
\end{theorem}

The coincidence captured in \cref{eq:sem-equiv} ensures the adequateness of our ePGF semantics for \cpgcl programs, which includes the case of \emph{undefined semantics}, i.e., the conditional probability (LHS) is not defined if and only if the normalized semantics (RHS) is undefined. Again, for \pgcl programs \emph{without conditioning},  the \emph{conditioned} semantic model is equivalent to that of \cite{LOPSTR} and thereby \cite{Kozen,DBLP:series/mcs/McIverM05}, since an \codify{observe}-free program never induces the violation term $\coef{\lightning}{F} X_\lightning$ and hence, the $\normalize$ operator has no effect.

	\section{Exact Bayesian Inference with Loops}
\label{sec:loops}

Loops significantly complicate inferring posterior distributions of probabilistic programs.
Computing the exact least fixed point of the characteristic function $\Phi_{B,P}$ (see \Cref{tab:semanics}) is in general highly intractable, and other techniques like \emph{invariant}-based reasoning are used.
Given the loop $\WHILEDO{B}{P}$, we call an ePGF transformer $I\colon \epgf \to \epgf$ an \emph{invariant} if $\Phi_{B,P}(I) = I$, i.e., it remains unchanged when pushed through one loop iteration.

Effectively, reasoning about loops is reduced to \emph{two challenges}:
\begin{enumerate*}
	\item finding an invariant candidate $I$, and 
	\item verifying that $I$ is indeed a valid invariant, i.e., deciding whether $\Phi_{B,P}(I) = I$.
\end{enumerate*}
Since the semantics of a program is also of type $\epgf \to \epgf$, we can describe such an invariant by means of a program.
To facilitate reasoning about such loop invariant programs, we consider a restricted set of \cpgcl programs, called \credip.
We further extend the program semantics to second-order ePGFs (eSOPs) to enable reasoning about multiple input distributions simultaneously.
We develop an eSOP-based equivalence checking technique for \credip programs to reason about loop invariants in a \emph{non-normalized} semantics.
This technique also enables invariant synthesis by solving equation systems yielding parameter values satisfying the invariant condition $\Phi_{B, P}(I) = I$.

\subsection{Program Equivalence}

Checking whether a loop-free program $I$ is an invariant of $\WHILEDO{B}{P}$ amounts to checking whether $\Phi_{B, P}(\sem{I}) = \sem{I}$.
Phrased in terms of generating functions, this reads
\begin{align}
	\label{eqn:eq_order_check}
	\forall G \in \epgf.~\;  \forall \sigma \in \N^k \cup \{\lightning\}. \quad \coef{\sigma}{\Phi_{B,P}(\sem{I})(G)} \eeq \coef{\sigma}{\sem{I}(G)}~.
\end{align}%
Namely, we need to check the equivalence of two loop-free programs.
As program equivalence is undecidable in general, we introduce a syntactic fragment of \cpgcl called \credip (conditional rectangular discrete probabilistic programs) for which equivalence of loop-free programs is decidable.

\begin{table}[t]
	\caption{Syntax (left) and the \emph{non-normalized} semantics (right) of \credip programs.}
	\label{tab:redip}
	\renewcommand{\arraystretch}{1.2}
	\centering
	\begin{tabular}{ll}
		$P$ & $\sem{P}(G)$\\
		\hline\hline
		$\assign{x}{n}$ & $G[X_\lightning /0, X/1] \cdot X^n+(G-G[X_\lightning /0])$\\
		$\decr{x}$ & $(G - G[X/0]) · X^{-1} + G[X/0]$\\
		$\incrasgn{x}{\iid{D}{y}}$ & $G[X_\lightning /0, Y/Y\sem{D}[T/X]]+(G-G[X_\lightning /0])$\\
		\hline
		$\ITE{\progvar{x} < n}{P_1}{P_2}$& $\sem{P_1}(G_{\progvar{x}<n}) + \sem{P_2}(G - G_{\progvar{x}<n}),$ where\\
		& $G_{\progvar{x}<n} = \sum_{i = 0}^{n-1} \tfrac{1}{i!}(\partial_X^i G[X_\lightning /0])[X/0]\cdot X^i$\\
		$\COMPOSE{P_1}{P_2}$& $\sem{P_2}(\sem{P_1}(G))$\\
		$\WHILEDO{\progvar{x} < n}{P_1}$& $(\lfp\, \Phi_{\progvar{x}<n, P_1} )(G),$ where\\
		& $\Phi_{\progvar{x}<n,P_1} (\psi) = \lambda F.~ (F\!-\!F_{\progvar{x}<n}) + \psi(\sem{P_1}(F_{\progvar{x}<n}))$ \\
		\hline
		$\observe{\FALSE}$& $G[\bvec{X}/\bvec{1}, X_\lightning /1]\cdot X_\lightning$
	\end{tabular}
	\renewcommand{\arraystretch}{1}
\end{table}

\paragraph{The \credip language.}  \Cref{tab:redip} describes the syntax and semantics of \credip. This fragment contains multiple statements to update the values of program variables.
	Intuitively, the updates are performed by extracting the parts of the ePGF that are affected by the update through substitution operations.
	For example, $\assign{x}{n}$ drops the observation-violation term and marginalizes w.r.t.\ $X$ (thus effectively setting $\progvar{x}$ to $0$ temporarily) and then performs a shift by $n$ in $X$.
	Finally, the unaffected part of the ePGF is added back to complete the transformation.

A prominent difference to \pgcl is the statement $\incrasgn{\progvar{x}}{\iid{D}{y}}$.
Intuitively, it can be interpreted as a bounded loop, namely $\codify{loop}(y)\{ \incrasgn{\progvar{x}}{\codify{sample}(D)}\}$ where the number of iterations is given by program variable $\progvar{y}$.
More specifically, $\incrasgn{\progvar{x}}{\iid{D}{y}}$ combines a series of operations: First independently sample $y$ many random variables from distribution $D$ and second, sum up the sampled values and increment $x$ by that amount.
For example, the program $P \coloneqq \assign{\progvar{y}}{10}; \assign{\progvar{x}}{0}; \incrasgn{\progvar{x}}{\iid{\bernoulli{\sfrac 1 2}}{\progvar{y}}}$ describes a binomial distribution in $X$ with parameters $n=10$ and $p=\sfrac 1 2$, i.e. $\sem{P} = Y^{10}\cdot (\sfrac 1 2 + \sfrac 1 2 X)^{10}$.

Moreover, we emphasize that Boolean guards in \credip can only be of the form $\progvar{x} < n$ where $n \in \N$ is a constant.
We denote by $G_{x<n}$ the PGF $G$ restricted to its terms with low enough order satisfying the guard $x<n$. The required elements of the PGF are collected by constructing the $i$-th formal derivative (for every $0 \leq i < n$) w.r.t.\ $X$ and extracting the constant monomials (in $X$), i.e. the coefficients of monomial $X^i$ in $G$.
By nesting of $\IFSYMBOL$-statements, axis-aligned hyper-rectangles can be identified, i.e., in this way we can express conjunction, disjunction and negation of guards.
The latter enables us to only consider $\codify{observe}(\FALSE)$ statements in our syntax, as we can reconstruct the full \enquote{rectangular} expressiveness for \codify{observe} statements.

%
%

The key feature of the \credip language is that its loop-free fragment preserves rational closed-form ePGF representations; see \cref{tab:redip} and \cite{CAV22}.
Hence, we can effectively compute the semantics of a loop-free \credip program given \emph{one} closed-form representation of the input distribution.
However, in order to decide program equivalence per \cref{eqn:eq_order_check}, we need to compute the semantics of \emph{infinitely or even uncountably many} input distributions.
\citet{CAV22} solved this issue by introducing second-order PGFs;
intuitively, these are FPS whose coefficients themselves are PGFs.
We extend this idea for programs with conditioning:

\begin{definition}[Second-Order ePGF]
	\label{def:sop}
	Let $\bvec{U} = (U_1, \ldots, U_k)$ be a tuple of formal indeterminates, that are pairwise distinct from $\bvec{X} = (X_1,\ldots,X_k)$ and $X_\lightning$ of \efps.
	A \emph{second-order} ePGF is a generating function of the form
	\[
     	G = \sum_{\sigma \in \N^k} G_\sigma U^\sigma = \sum_{\sigma \in \N^k} (\constrain{G_\sigma}{\TRUE} + \coef{\lightning}{G_\sigma} X_{\lightning} ) U^\sigma 
	       = \sum_{\sigma \in \N^k} \constrain{G_\sigma}{\TRUE} U^\sigma ~+~ \sum_{\sigma \in \N^k}\coef{\lightning}{G_\sigma} X_{\lightning}U^\sigma,
	\]
	where $G_\sigma \in \epgf$. We denote the set of second-order ePGFs by \esop.
\end{definition}
\noindent%
	{%
	\makeatletter
	\let\par\@@par
	\par\parshape0
	\everypar{}
	\begin{wrapfigure}{r}{0.36\linewidth}
		\vspace{-1.7em}
		\begin{minipage}{\linewidth}
			\centering
			\begin{align*}
				& \annotate{\,X^1 \cdot U^1 + X^2 \cdot U^2 + X^3 \cdot U^3}                                                             \\
				& \incrasgn{\progvar{x}}{\iid{\bernoulli{\sfrac 1 2}}{\progvar{x}}}                                                    \\
				& \annotate{\,\frac{1}{2} (X^1 + X^2) \cdot U^1 \ +} \\
				& \quad \annocolor{\,\frac{1}{4} (X^2 + 2 X^3 + X^4) \cdot U^2 \ +}\\
				& \quad \annocolor{\,\frac{1}{8} ( X^3 + 3 X^4 + 3 X^5 + X^6) \cdot U^3}                            \\
				& \OBSERVE{\progvar{x} < 3}                                                                                                \\
				& \annotate{\,\frac{1}{2} (X^1 + X^2) \cdot U^1 \ +}\\
				& \quad \annocolor{\,\frac{1}{4} (X^2 + 3 X_\lightning) \cdot U^2 \ +}\\
				&\quad \annocolor{\,X_\lightning \cdot U^3}
			\end{align*}
			\captionof{program}{A \credip program annotated with eSOP semantics.}
			\label{prog:credip_esop_annotated}
		\end{minipage}
	\end{wrapfigure}
\noindent
An eSOP hence represents, in a single formal power series, \emph{multiple} ePGFs as coefficients $G_\sigma$ of different monomials $\bvec{U}^\sigma$.
Intuitively one can interpret $\bvec{U}$ as eFPS formal indeterminates of additional program variables which do not occur in the program and whose sole purpose is to remember the actual program variables' initial values.
We can naturally extend the denotational semantics described in \cref{tab:redip} to \esop, as demonstrated by the following example.
	\par}%

\begin{example}[eSOP Semantics of \credip Program]
	Con-
	{%
		\makeatletter
		\let\par\@@par
		\par\parshape0
		\everypar{}
		\begin{wrapfigure}{r}{0.36\linewidth}
			\vspace{-1.7em}
			\begin{minipage}{\linewidth}
				\centering
				\begin{align*}
					& \phantom{\annotate{\,X^1 \cdot U^1 + X^2 \cdot U^2 + X^3 \cdot U^3} }                                                            \\
					& \phantom{\incrasgn{\progvar{x}}{\iid{\bernoulli{\sfrac 1 2}}{\progvar{x}}}}                                                    \\
					& \phantom{\annotate{\,\frac{1}{2} (X^1 + X^2) \cdot U^1 \ +} }\\
					& \phantom{\quad \annocolor{\,\frac{1}{4} (X^2 + 2 X^3 + X^4) \cdot U^2 \ +}}\\
					& \phantom{\quad \annocolor{\,\frac{1}{8} ( X^3 + 3 X^4 + 3 X^5 + X^6) \cdot U^3}       }                     \\
					& \phantom{\OBSERVE{\progvar{x} < 3}                                                                                         }       \\
					& \phantom{\annotate{\,\frac{1}{2} (X^1 + X^2) \cdot U^1 \ +}}\\
					& \phantom{\quad \annocolor{\,\frac{1}{4} (X^2 + 3 X_\lightning) \cdot U^2 \ +}}\\
					&\phantom{\quad \annocolor{\,X_\lightning \cdot U^3}}
				\end{align*}
			\end{minipage}
		\end{wrapfigure}
		\noindent
			sider the \credip program $P$ in Prog.~\ref{prog:credip_esop_annotated} together with the \esop input generating function $G = 1 X^1 \cdot U^1 + 1 X^2 \cdot U^2 + 1 X^3 \cdot U^3$, identifying indeterminate $X$ and meta-indeterminate $U$ for program variable $\progvar{x}$.
			This eSOP represents three Dirac distributions, i.e., $1 X^1, 1 X^2$, and $1 X^3$, where the purpose of $U$ is to remember the initial value of $\progvar{x}$.
			We now examine the computation of $\sem{P}(G)$ step-by-step, starting with the increment operation which only affects the indeterminate $X$ of the involved program variable $\progvar{x}$ and does not affect $U$.
			To 
		\par}%
	\noindent
	this end, we substitute $(\frac 1 2 + \frac 1 2 X) \cdot X$ for $X$, since $G$ contains no initial observation-violation term.
	Afterwards, to aggregate the states that violate the observation, the semantics also substitutes $1$ for indeterminate $X$ (and $X_\lightning$) and leaves the meta-indeterminates unchanged.
	As a result, we obtain
	\(
	\sem{P}(G) = \sfrac{1}{2} (X^1 + X^2) \cdot U^1 + \sfrac{1}{4} (X^2 + 3 X_\lightning) \cdot U^2 + X_\lightning \cdot U^3,
	\)
	and have computed all posterior distributions for initial state valuations $\progvar{x}=1, \progvar{x}=2, \progvar{x}=3$ in one shot.
	For instance, when starting with initial distribution $1X$ the posterior distribution is $\sfrac{1}{2} (X^1 + X^2)$ as indicated by the coefficient of $U^1$.
	Finally, we note that the meta-indeterminates just ``pass through'' the
	\esop semantics functional, i.e., it can be seen as the point-wise lifting of the \epgf semantics.
		\qedT
\end{example}

\begin{theorem}[eSOP Semantics]
	\label{thm:sop_semantics}
	Let\, $P$ be a loop-free \credip program.
	Let\, $G = \sum_{\sigma \in \N^k} G_\sigma \bvec{U}^\sigma \in \esop$.
	The \esop semantics $\sem{P}\colon \esop \to \esop$ of $P$ can be computed by
	\[
	\sem{P}(G) \eeq \sum\nolimits_{\sigma \in \N^k} \sem{P}(G_\sigma) \cdot \bvec{U}^\sigma~.
	\]
\end{theorem}


Since PGF semantics is an instance of the general framework of Kozen's measure transformer semantics \cite{Kozen,LOPSTR}, the posterior distribution of a \credip program $P$ is uniquely determined by its semantics on all possible Dirac distributions.
One can thus construct an eSOP from $P$ that represents all possible point-mass distributions for the program variables:
\begin{definition}[Equivalence-Witness eSOP]
	\label{def:testinput}
	Let $\hat{G}$ be an eSOP defined as
	\[
	\hat{G} \ddefeq \underbrace{\left(1-X_1U_1\right)^{-1}\cdots\left(1-X_kU_k\right)^{-1}}_{\text{rational closed form}} \eeq \sum\nolimits_{\sigma \in \N^k} \bvec{X}^\sigma\bvec{U}^\sigma \eeq 1 + (1\bvec{X})\bvec{U}+ (1\bvec{X}^2) \bvec{U}^2 + \cdots ~,
	\]%
	where the meta-indeterminates $\bvec{U}$ serve the purpose of \enquote{remembering} the initial state valuations.
\end{definition}

For the purpose of deciding program equivalence, $\hat{G}$ is particularly useful, since it represents Dirac distributions for all potential initial state valuations, with the exception of any \codify{observe}-violation probabilities.
This is, however, not a problem, as such observation-violation terms can be immediately removed from the equivalence check (by \Cref{lem:errpassthru}).
As a consequence, we can use $\hat{G}$ to characterize program equivalence of loop-free \credip programs using \esop.
This is expressed by the following lemma.
\begin{lemma}[\esop Characterization]
	\label{lem:sopequiv}
	Let $P_1$ and $P_2$ be loop-free \credip programs with $\text{Vars}(P_i) \subseteq \{\progvar{x}_1,\ldots, \progvar{x}_k\}$ for $i \in \{1, 2\}$. Further, consider a vector $\bvec{U} = (U_1,\ldots, U_k)$ of meta-indeterminates.
	Then, 
	\[
		\forall G \in \textup{\epgf}. \ \ \sem{P_1}(G)~=~\sem{P_2}(G) \qquad\textnormal{iff}\qquad \sem{P_1}(\hat{G})~=~\sem{P_2}(\hat{G}).
	\]
\end{lemma}

\noindent As we can compute $\sem{P}(\hat{G})$ for loop-free $P \in \credip$, the following consequence is immediate.
\begin{corollary}[Decidability of Equivalence]
	\label{cor:loopfree_decidability}
	Let $P_1, P_2$ be two loop-free \credip programs. Then,
	\[
		\forall G \in \epgf.\ \ \sem{P_1}(G) \eeq \sem{P_2}(G)\quad\text{is decidable.}
	\]
\end{corollary}
\begin{proof}
	By utilizing \Cref{lem:sopequiv}, we can rephrase the problem of determining program equivalence through the \esop characterization $\sem{P_1}(\hat{G}) = \sem{P_2}(\hat{G})$.
	It is worth noting that $\hat{G}$ represents a \emph{rational closed-form} eSOP $\hat{G} = \frac{1}{1 - X_1 U_1} \frac{1}{1 - X_2 U_2} \cdots \frac{1}{1 - X_k U_k} \in\R[[\bvec{X},\bvec{U}]].$
	For our purposes, we can disregard the portion of $\hat{G}$ that describes the initial observe violation behavior, as it immediately cancels out (see \cref{lem:errpassthru}).
	As $\hat{G}$ is in rational closed form, both $\sem{P_1}(\hat{G})$ and $\sem{P_2}(\hat{G})$ must also possess a rational closed form since loop-free \credip semantics preserve closed forms; see \cref{tab:redip} and \cite{CAV22}.
	Additionally, the effective computation of $\sem{P_1}(\hat{G}) = F_1/H_1$ and $\sem{P_2}(\hat{G}) = F_2/H_2$ is possible because both $P_1$ and $P_2$ are loop-free programs.

	In $\R[[\bvec{X}, X_\lightning, \bvec{U}]]$, the question of whether two formal power series represented as rational closed forms, namely $F_1/H_1$ and $F_2/H_2$, are equal can be decided:
	\begin{align*}
		\frac{F_1}{H_1} ~=~ \frac{F_2}{H_2} \qquad & \iff\qquad F_1 H_2 ~=~ F_2 H_1,
	\end{align*}
	since the latter equation concerns the equivalence of two polynomials in $\R[\bvec{X}, X_\lightning,\bvec{U}]$.
	Therefore, we can compute these two polynomials and verify whether their (finite number of) non-zero coefficients coincide.
	If they do, then $P_1$ and $P_2$ are equivalent (i.e., $\sem{P_1} = \sem{P_2}$), whereas if they do not, they are not equivalent.
	In the case of non-equivalence, we can generate a Dirac distribution that produces two distinct outcomes.
	This is achieved by taking the difference $F_1 H_2 - F_2 H_1$ and computing the first non-zero coefficient in $\R[\bvec{X}, X_\lightning\!]$.
	Then, extracting the exponent of the monomial describes an initial state valuation $\sigma$, with $\sem{P_1}(\sigma) \neq  \sem{P_2}(\sigma)$.
\end{proof}

\begin{remark}
	The proof of \cref{cor:loopfree_decidability} (on decidability of equivalence) relies on the fact that the \esop transformer $\sem{P}(\cdot)$ preserves rational closed-form eSOPs. \credip is a non-trivial fragment of \cpgcl for which we can show the preservation of rational closed forms for loop-free programs; but it is not necessarily the largest class of programs that features such a property. Investigating a more expressive fragment with decidability of equivalence is subject to future work.
%
\qedT
\end{remark}

\subsection{Invariant-Based Reasoning with Conditioning}
\label{ssec:invariants}

\credip is a fragment of \cpgcl for which the equivalence of loop-free programs is decidable.
We now exploit this result to reason about loops in \credip programs.
The key idea is to use loop-free \credip programs as potential invariant candidates.
Recall the two main challenges of invariant-based reasoning: first, find an invariant candidate, and second, verify that it is indeed an invariant, i.e., $\Phi_{B, P}(I) = I$.
In the remainder of this section, we focus on \emph{verifying} invariant candidates given in the form of \credip programs, while deferring finding invariants to \Cref{sec:synthesis}.

We first introduce the notion of lossless ePGF transformers to capture program termination:

\begin{definition}[Lossless ePGF Transformers]
	\label{def:lossless}
	An ePGF transformer $H\colon \epgf \to \epgf$ is \emph{lossless} for $F\in \epgf$ if

	\vspace*{-6mm}
	\[
		\abs{H(F)} + \coef{\lightning}{H(F)} \eeq \abs{F} + \coef{\lightning}{F}~.
	\]
	$H$ is \emph{universally lossless} if it is lossless for all $F$ in $\epgf$.
\end{definition}%
\noindent
Intuitively, a lossless ePGF transformer is a mapping that does not leak any probability mass.
Since the semantics of a program $P$ is an ePGF transformer, $\sem{P}$ being (universally) lossless coincides with $P$ being (universally) almost-surely terminating, abbreviated as (U)AST \cite{DBLP:conf/mfcs/Saheb-Djahromi78,DBLP:conf/rta/BournezG05}.
%
%
Given $L = \WHILEDO{B}{P}$, we can approximate its least fixed point $\lfp~\Phi_{B, P}$ leveraging domain theory, in particular, Park's lemma, namely, $\Phi_{B,P}(I) \sqsubseteq I$ implies $\sem{L} \sqsubseteq I$ \cite{park1969fixpoint}.
It enables reasoning about \codify{while}-loops in terms of over-approximations and -- in case a program is \UAST -- also about program equivalence.

\begin{theorem}[Loop Invariants]
	\label{thm:loop_invs}
	Given $L = \WHILEDO{B}{P}$ and a universally lossless ePGF transformer $I\colon \epgf \to \epgf$. We have
	\begin{enumerate}
		\item If \,$\Phi_{B,P}(I) \sqsubseteq I$, then $\normalize(\sem{L}(F)) \preceq \normalize(I(F))$ whenever $\normalize(I(F))$ is defined.
		\item If \,$L$ is \UAST, then $I$ is an invariant of $L$ if and only if
		      \[\sem{L} \eeq I \qquad\textnormal{and}\qquad \normalize(\sem{L}(F)) \eeq \normalize(I(F))~.\]
	\end{enumerate}
\end{theorem}
\begin{proof}

	For (1), we first prove that the normalization function is monotonic, whenever it is defined.
	Let $F,G \in \epgf$ such that $\normalize(F), \normalize(G)$ are defined. We have:
	\begingroup
	\allowdisplaybreaks
	\begin{align*}
		F \preceq G & \impliesqq \coef{\lightning}{F} \leq \coef{\lightning}{G} \quad \text{and} \quad \sum_{\sigma \in \N^k}\coef{\sigma}{F}\bvec{X}^\sigma \preceq \sum_{\sigma \in \N^k} \coef{\sigma}{G}\bvec{X}^\sigma                         \\
		            & \impliesqq 1-\coef{\lightning}{F} \geq 1- \coef{\lightning}{G} \quad \text{and} \quad \sum_{\sigma \in \N^k}\coef{\sigma}{F}\bvec{X}^\sigma \preceq \sum_{\sigma \in \N^k} \coef{\sigma}{G}\bvec{X}^\sigma                    \\
		            & \impliesqq \frac{1}{1-\coef{\lightning}{F}} \leq \frac{1}{1-\coef{\lightning}{G}} \quad \text{and} \quad \sum_{\sigma \in \N^k}\coef{\sigma}{F}\bvec{X}^\sigma \preceq \sum_{\sigma \in \N^k} \coef{\sigma}{G}\bvec{X}^\sigma \\
		            & \impliesqq \frac{1}{1-\coef{\lightning}{F}}  \cdot \sum_{\sigma \in \N^k}\coef{\sigma}{F}\bvec{X}^\sigma \quad\preceq\quad \frac{1}{1-\coef{\lightning}{G}} \cdot \sum_{\sigma \in \N^k} \coef{\sigma}{G}\bvec{X}^\sigma      \\
		            & \impliesqq \normalize(F) ~\preceq~ \normalize(G)~.
	\end{align*}%
	\endgroup
	It follows that $\normalize(\sem{\WHILEDO{B}{P}}(F)) \preceq \normalize(I(F))$, due to Park's lemma.

	For (2), first assume that $\sem{L} = I$ and $\normalize(\sem{L}(F)) = \normalize(I(F))$.
	As $I = \sem{L} = \lfp\, \Phi_{B,P}$, $I$ is trivially identified as an invariant.
	For the other direction, assume that $I$ is an invariant (i.e., a fixed point).
	Thus, $I$ must be at least $\lfp\ \Phi_{B,P} = \sem{\WHILEDO{B}{P}}$. Moreover, because $\WHILEDO{B}{P}$ is \UAST, it follows that
	\[
		\abs{\sem{\WHILEDO{B}{P}}(F)} + \coef{\lightning}{\sem{\WHILEDO{B}{P}}(F)} \eeq \abs{F}  + \coef{\lightning}{F} \eeq \abs{I(F)} + \coef{\lightning}{I(F)} \quad \text{for all } F \in \epgf~.
	\]
	The second equality arises from $I$ being universally lossless.
	Combining these results yields
	\begin{align*}
		\forall F \in \epgf.~ & \left(\sem{\WHILEDO{B}{P}}(F) ~\preceq~ I(F)\right.                                                                                          \\
		                      & \left.\text{and} ~\abs{\sem{\WHILEDO{B}{P}}(F)} + \coef{\lightning}{\sem{\WHILEDO{B}{P}}(F)} ~=~ \abs{I(F)} + \coef{\lightning}{I(F)}\right) \\
		                      & \!\implies  \forall F \in \epgf. ~ \sem{\WHILEDO{B}{P}}(F) = I(F)
		\,\Longleftrightarrow ~ \sem{\WHILEDO{B}{P}} = I~.
	\end{align*}
	Then, $\normalize(\sem{\WHILEDO{B}{P}}(F)) = \normalize(I(F))$ follows for all $F \in \epgf$.
\end{proof}

\begin{program}[t]
	\centering
	\begin{minipage}[b]{.515\textwidth}
		\begin{align*}
			& \WHILE{\progvar{y} = 1}\\
			& \quad \pchoice{\assign{y}{0}}{\sfrac{1}{2}}{\assign{y}{1}}\,\fatsemi\\
			& \quad \assign{\progvar{x}}{\progvar{x} + 1}\,\fatsemi\\
			& \quad \observe{\progvar{x} < 3}~\}
		\end{align*}
		\caption{A truncated geometric distribution generator.}
		\label{prog:trunc_geo}
	\end{minipage}
	\hfil
	\begin{minipage}[b]{.475\textwidth}
		\begin{align*}
			& \IF{\progvar{y} = 1}\\
			& \quad \incrasgn{\progvar{x}} { \iid{\geometric{1/2}} {y}	}\,\fatsemi\\
			& \quad \assign{\progvar{y}}{0}\,\fatsemi\\
			& \quad \observe{\progvar{x} < 3}~\} 
		\end{align*}
		\caption{A loop-free \credip invariant of Prog.\ \ref{prog:trunc_geo}.}
		\label{prog:trunc_geo_2}
	\end{minipage}
\end{program}

Combining the results from this section, we can state the decidability of checking invariant validity for loop-free \credip candidates.
\begin{theorem}
	Let $L=\WHILEDO{B}{P} \in \credip$ be UAST with loop-free body $P$ and $I$ be a loop-free \credip program. It is decidable whether $\sem{L} = \sem{I}$.
\end{theorem}
\begin{proof}
	The correctness is an immediate consequence of \Cref{thm:loop_invs} and \Cref{cor:loopfree_decidability}.
\end{proof}

We demonstrate our invariant-based reasoning technique by \Cref{ex:geom}.

\begin{example}[Geometric Distribution Generator]
	\label{ex:geom}
	Prog.\ \ref{prog:trunc_geo} describes an iterative algorithm that repeatedly flips a fair coin -- while counting the number of trials -- until seeing heads, and observes that the number of trials is less than 3.
	Assume we want to compute the posterior distribution for input $1\cdot Y^1X^0$ (i.e. $y=1$ and $x=0$).
	We first evaluate $\lfp~\Phi_{B,P}$.
	Using \Cref{thm:loop_invs}\,(2), we perform an \emph{equivalence check} on the invariant in Prog. \ref{prog:trunc_geo_2}.
	As Prog.\ \ref{prog:trunc_geo} and \ref{prog:trunc_geo_2} are equivalent, we substitute the loop-free program for the \codify{while}-loop and continue.
	The resulting posterior distribution for input $Y$ is $\sem{P}(Y) = \frac{4}{7} + \frac{2}{7}X + \frac{1}{7}X^2$.
	Since Prog.\ \ref{prog:trunc_geo} is UAST, this is its \emph{precise posterior distribution}.
	The step-by-step computation of the equivalence check can be found in \Cref{apx:prodigy}.
	\qedT
\end{example}
\noindent To summarize, \emph{reasoning about program equivalence using eSOPs enables exact Bayesian inference for \credip programs containing loops}.
We remark that \emph{nested loops} can be treated in a compositional manner: We first provide a loop-free invariant for the inner loop, prove its correctness (i.e., equivalence), and then replace the inner loop by its invariant and repeat the procedure for the outer loop.
This feature of compositional reasoning is a key benefit of reusing the loop-free fragment of cReDiP as a specification language to describe invariants.

\subsection{Equivalence of Normalized Semantics}
Our previous notion of equivalence $\sem{L} = \sem{I}$ describes the equivalence of the \emph{non-normalized} semantics, i.e., the \WHILESYMBOL-loop and the loop-free invariant generate exactly the same distributions and \codify{observe}-violation probabilities, which immediately entails also the equivalence of the normalized semantics, i.e., $\normalize(\sem{L}) = \normalize(\sem{I})$, but not necessarily the reverse. In practice, however, it is interesting to have a \emph{weaker} notion of equivalence which addresses only the normalized semantics, regardless of observation violations (as programmers may use different observation strategies to construct programs yielding the same output distribution). This weaker notion reads as
\begin{equation}\label{eqn:condequiv}
	P \sim Q \quad \text{iff} \quad \forall G \in \pgf.~ \normalize(\sem{P}(G)) = \normalize(\sem{Q}(G))~.
\end{equation}%
We aim to capture such equivalence again using eSOPs.
First, we lift the operator \normalize to eSOPs:
\begin{definition}[Conditioning on \esop]
	Let $G\in \esop$.
	The function 
	\[
	\cond\colon \esop \to \sop\,, \qquad G \mapsto \sum\nolimits_{\sigma \in \N^k} \normalize(G_\sigma)\bvec{U}^\sigma~.
	\]
	is called the \emph{conditioning function}.
\end{definition}
\noindent For simplicity, we assume that $\forall \sigma \in \N^k. ~G_\sigma \neq X_\lightning$ as otherwise $\normalize$ is not defined. 
Note that $\cond$ often \emph{cannot} be evaluated in a closed-form eSOP as there may be \emph{infinitely many} ePGF coefficients of the (non-normalized) eSOP that have different observation-violation probabilities. However, we present a sufficient condition under which $\cond$ can be evaluated on closed-form eSOPs:
\noindent\begin{proposition}
	Let $F_1, F_2 \in \epgf$, with $p \coloneqq \coef{\lightning}{F_1} = \coef{\lightning}{F_2}$.
	Then,
	\[
	\cond(F_1) + \cond(F_2) = \frac{\constrain{F_1}{\TRUE} + \constrain{F_2}{\TRUE}}{1-p} = \cond(F_1 + F_2)~.
	\]
\end{proposition}

\noindent Intuitively, addition distributes over $\cond$, i.e., $\cond$ behaves linearly. Generalizing this concept to finitely many equal \codify{observe}-violation properties we get the following.

\begin{corollary}[Partitioning]\label{cor:partitioning}
	Let $S$ be a finite partitioning of $\;\N^k = S_1 \uplus \cdots \uplus S_m$ with $\coef{\lightning}{G_\sigma} = \coef{\lightning}{G_{\sigma'}}$, for all $\sigma, \sigma' \in S_i, ~ 1 \leq i \leq m$. Then:
	\[
	G \eeq \sum\nolimits_{i = 1}^{m} \sum\nolimits_{\sigma \in S_i} (\coef{\lightning}{S_i} X_\lightning + \constrain{G_\sigma}{\TRUE}) \bvec{U}^\sigma~,
	\]
	where $\coef{\lightning}{S_i}$ denotes the observation-violation probability in $S_i$.
	For such $G$ we have:
	\[
	\cond(G) \eeq \sum\nolimits_{i=1}^{m} 
	\frac{
		\sum_{\sigma \in S_i}
		\constrain{G_\sigma}{\TRUE}
		\bvec{U}^\sigma
	}{
		1-\coef{\lightning}{S_i}
	}~.
	\]
\end{corollary}
Unfortunately, requiring a finite partitioning is quite restrictive.
Finite partitioning is impossible already for some loop-free programs, an
%
\noindent  example is provided in Prog.\ \ref{prog:loopfree}. Given an initial distribution for variable $\progvar{y}$, the program computes the sum of $\progvar{y}$-many independent and identically distributed Bernoulli variables with success probability $\sfrac{1}{2}$. This is equivalent to sampling from a binomial distribution with $y$ trials and probability $\sfrac{1}{2}$. Finally, it marginalizes the distribution by assigning $\progvar{y}$ to zero and conditions on the
event that $\progvar{x}$ is less than 1, resulting in $\sum_{i=0}^{\infty} \frac{\left(2^{-i} + (1-2^{-i})X_\lightning\right)V^i}{(1-U)}$.
We can deduce that for any initial state valuation $(x,y)$ we obtain a \emph{different} \codify{observe} violation probability $(1-2^{-y})$, hence we cannot finitely partition the state space into equal violation probability classes.

Another challenge when considering the equivalence of normalized distributions is:
Evaluating $\cond$ on (closed-form) eSOPs yields that $\cond(\sem{P}(\hat{G})) = \cond(\sem{Q}(\hat{G}))$.
This implies $\forall \sigma \in \N^k.~ \normalize(\sem{P}(\bvec{X}^\sigma)) = \normalize(\sem{Q}(\bvec{X}^\sigma))$, i.e., equivalence on point-mass distributions.
However, we do not necessarily have the precise equivalence as per \Cref{eqn:condequiv}, because the $\normalize$ operator used to define $\cond$ is a non-linear function\footnote{For the non-normalized semantics, general equivalence $\sem{P} = \sem{Q}$ follows from the linearity of the transformer.} and thus the point-mass distributions cannot be combined in a sensible way.
However, in many use cases we are only interested in the behavior of a specific initial state valuation where such a result on point-mass equivalence can still be useful.


	\section{Finding Invariants using Parameter Synthesis}
\label{sec:synthesis}

\begin{program}[t]
	\centering
	\vfill
	\begin{minipage}[t]{0.47\linewidth} 
		\begin{align*}
		& \annotate{(1-XU)^{-1}(1-YV)^{-1}} \\
		& \ASSIGN{\progvar{x}}{0}\fatsemi \\
		& \annotate{(1-U)^{-1}(1-YV)^{-1}} \\
		& \incrasgn{\progvar{x}}{\iid{\bernoulli{\sfrac{1}{2}}}{\progvar{y}}}\fatsemi \\
		%
		%
		& \annotate{2(1-U)^{-1}(2-(1+X)YV)^{-1}} \\
		& \ASSIGN{\progvar{y}}{0}\fatsemi \\
		& \annotate{2(1-U)^{-1}(2-(1+X)V)^{-1}} \\
		& \observe{\progvar{x}<1}\fatsemi \\
		& \annotate{ \frac{2 (1-X_\lightning)}{(1-U)(2-V)} + \frac{X_\lightning}{(1-U)(1-V)}} \\
		& \annotate{\sum_{i=0}^{\infty} \frac{\left(2^{-i} + (1-2^{-i})X_\lightning\right)V^i}{(1-U)}}
		\end{align*}
		\captionof{program}{Program with infinitely many observe violation probabilities.}
		\label{prog:loopfree}
	\end{minipage}
	\hspace*{.6cm}
	\begin{minipage}[t]{0.47\textwidth}
		\vspace{-1mm}
		\begin{align*}
		& \WHILE{\progvar{n} > 0} \\
		& \quad \pchoice{
			\assign{\progvar{n}}{\progvar{n}-1}
		} {\sfrac{q}{3}} {
			\assign{\progvar{c}}{\progvar{c}+1}
		}\\
		&\}
		\end{align*}
		\caption{$n$-geometric generator with success probability $\sfrac{q}{3}$ for $0 \leq q \leq 3$.}
		\label{prog:n_geom_param}
		\vspace{1.35cm}
		\begin{align*}
		& \pcomment{sums n geometric(p) samples }\\
		& \incrasgn{\progvar{c}}{\iid{\geometric{p}}{\progvar{n}}}\fatsemi \\
		& \pcomment{on termination n is zero}\\
		& \assign{\progvar{n}}{0}
		\end{align*}
		\caption{$n$-geometric invariant with parameter $p$.}
		\label{prog:n_geom_param_inv}
	\end{minipage}
\end{program}

In contrast to the previous section which aims at \emph{validating a given invariant}, in this section, we address the problem of \emph{finding} such invariants.
For related problems, e.g., finding invariants in terms of weakest preexpectations, there exist sound and complete synthesis algorithms for \emph{subclasses} of loops and properties that can be verified by piecewise linear templates \cite{DBLP:conf/tacas/BatzCJKKM23}.
We adopt the idea of template-based invariant synthesis and leverage the power of eSOPs to achieve decidability results for a subclass of invariant candidates.
Our templates are described by parametric loop-free cReDiP programs, e.g., $I_p = \PCHOICE{\ASSIGN{\progvar{x}}{1}}{p}{\ASSIGN{\progvar{x}}{0}}$ which models a Bernoulli distribution with symbolic parameter $p$.
We believe that %
\begin{enumerate*}
	\item using programs as templates is (in particular in the probabilistic case) intuitively easier than using first-order logic as typically used to express invariants, and
	\item finding suitable templates can be encoded as a program synthesis problem whose hardness may be precisely quantified.
\end{enumerate*} 
	Recall the invariant synthesis problem: Given a \WHILESYMBOL-loop $L = \WHILEDO{B}{P}$, find a loop-free \credip program $I$ such that $\Phi_{B,P}(\sem{I}) = \sem{I}$.
	Sometimes, the general shape of an invariant template $T_\bvec{p}$ (with a vector $\bvec{p}$ of parameters) is derivable from $L$, but finding a valid parameter valuation may be involved.
%
We illustrate the idea by \cref{ex:params}.

\begin{example}[$n$-Geometric Parameter Synthesis]
	\label{ex:params}
	Prog.\ \ref{prog:n_geom_param} (with loop body $P$) is a variant of Prog.\ \ref{prog:odd_geo}, where instead of requiring one success (setting $\progvar{h}=0$), we need $\progvar{n}$ successes to terminate.
	Furthermore, the individual success probability is $\frac q 3$, where $q$ is a symbolic parameter.
	It seems natural that this program encodes the $n$-fold geometric distribution\footnote{Sometimes also called negative binomial distribution.} with individual success probability $\frac{q}{3}$.
	This suggests to formulate the invariant template $Q_p$ given in Prog.\ \ref{prog:n_geom_param_inv},  where $\progvar{c}$ is a sum of $\progvar{n}$ geometric distributions with an unknown parameter $p$.
	Using \Cref{thm:loop_invs}, we can derive the equivalence of Prog.\ \ref{prog:n_geom_param} and Prog.\ \ref{prog:n_geom_param_inv} and obtain an equation in $p$ and $q$:

	\begin{align*}
		\Phi_{B,P}(\sem{Q_p})(\hat{G}) & =-\frac{(-3+qCU+3C-3pC-qU+3pU-3pCU)}{3(-1+CV)(-1+C-pC+pU)} \\
		\sem{Q_p}(\hat{G})             & =-\frac{(-1+C-pC)}{(-1+CV)(-1+C-pC+pU)}
	\end{align*}
	\[
		\text{Then} \qquad \Phi_{B,P}(\sem{Q_p})(\hat{G}) = \sem{Q_p}(\hat{G}) \qquad\text{iff}\qquad p = \frac{q}{3}~.
	\]

	\noindent The formal variable $C$ corresponds to program variable $\progvar{c}$, while $U$ and $V$ are meta-indeterminates corresponding to the variables $\progvar{n}$ and $\progvar{c}$.
	This result tells us, that for $p=\frac q 3$ our parametrized invariant program is an invariant of Prog.~\ref{prog:n_geom_param}.
	\qedT
\end{example}

This approach works in general as the following theorem describes:

\begin{theorem}[Decidability of Parameter Synthesis]
	Let $W$ be a \credip \WHILESYMBOL\ loop and $I_{\bvec{p}}$ be a parametrized loop-free \credip program.
	The problem whether there exist parameter values $\bvec{\rho}$ such that the instantiated template $I_\bvec{\rho}$ is an invariant, i.e.,
	\[
		\exists\, \bvec{p} \in \R^l. \quad \sem{W} \eeq \sem{I_{\bvec{p}}}
		\qquad \text{is decidable.}
	\]
\end{theorem}
\begin{proof}
	The proof is a variant of \cref{cor:loopfree_decidability}. Full details are provided in \cref{apx:synthesis}.
\end{proof}

\noindent Note that in this formulation, parameters may depend on other parameters, but are always \emph{independent} of all program variables and second-order indeterminates. Unfortunately, not every parametric invariant can be expressed by a loop-free \credip program as illustrated by the following example.



\begin{example}[Hypergeometric Invariant]
	Prog.\ \ref{prog:cond_and} encodes a biased 2-dimensional bounded random walk.
	In each turn, it decrements one of the variables with equal probability $\sfrac{1}{2}$ until either%
	{\makeatletter
		\let\par\@@par
		\par\parshape0
		\everypar{}
		\begin{wrapfigure}{r}{0.37\textwidth}
			\vspace*{-5mm}
			\begin{minipage}{1\linewidth} 
				\begin{align*}
					& \WHILE{\progvar{n} > 0 \wedge \progvar{m} > 0}                                                          \\
					& \quad \PCHOICE{\ASSIGN{\progvar{m}}{\progvar{m}-1}}{\sfrac{1}{2}}{\ASSIGN{\progvar{n}}{\progvar{n} -1}} \\
					& \}
				\end{align*}
			\end{minipage}
			\captionof{program}{Dependent negative binomial variables.}
			\label{prog:cond_and}
			\vspace*{-4mm}
		\end{wrapfigure}
	\noindent 
	the value of $\progvar{m}$ or $\progvar{n}$ arrives at 0.
	For any fixed program state valuation \mbox{$(0,0) \neq (\progvar{m},\progvar{n}) \in \N^2$}, the number of loop iterations is bounded by $\progvar{n}+\progvar{m}-1$. We are interested in the exact posterior distribution for arbitrary input distributions.
	Due to its finite nature for any particular input distribution with finite support, we can analyze this program automatically using \toolname{Prodigy} by unfolding the loop $\progvar{m}+ \progvar{n}-1$ times.
	For instance, the resulting distribution for an initial Dirac distribution describing the state valuation $(a,b)$, is
	$\sem{P}(M^a N^b) = \sum_{i=1}^a \frac{M^i}{2^{a+b-i}} \cdot \binom{a+b-i-1}{b-1} + \sum_{i=1}^b \frac{N^i}{2^{a+b-i}} \cdot \binom{a+b-i-1}{a-1}$.
	Using the simplification function in Mathematica \cite{Mathematica}, we derive the closed form,
	\[
		I(a,b) = 2^{1 - a - b} M \binom{-2 + a + b}{ -1 + b} {}_2F_1(1, 1 - a, 2 - a - b, 2 M) +
		2^{1 - a - b} N \binom{-2 + a + b}{-1 + a} {}_2F_1(1, 1 - b, 2 - a - b, 2 N).
	\]
	Here ${}_2F_1$ denotes the hypergeometric function\footnote{More about this closed form and algorithms to compute closed forms alike can be found in \cite{petkovsek1996b}.}.
	It shows that the distribution is in some sense linked to the hypergeometric distribution, indicated by the ${}_2F_1$ terms.
	Even though that function is quite complex, taking derivatives in $M$ or $N$ respectively is straightforward, i.e., $\frac{\partial}{\partial_x} {}_2F_1(p_1, p_2, p_3; x) = \frac{p_1 p_2}{c} {}_2F_1(p_1+1, p_2+1, p_3+1; M)$.
	Thus, extracting many properties of interest can still be computed exactly using the closed-form expression.
	It is unknown (to us) whether some loop-free \credip invariant program generates this closed-form distribution.
	However, the GF semantics enables us to prove that the precise semantics of Prog.\ \ref{prog:cond_and} is captured by checking $\forall a, b \in \N.~ (a,b) \neq (0,0) \implies I(a,b) = \Phi_{B, P}(I)(a,b)$, combined with the fact that it universally certainly terminates.
		\par}%
\end{example}

	\section{Empirical Evaluation of \toolname{Prodigy}}
\label{sec:prodigy}
We have implemented our approach in Python as an extension to \toolname{Prodigy}\footnote{\ifanonymous Repository link hidden due to double-blind reviews. \else \url{https://github.com/LKlinke/Prodigy}\fi}\ifanonymous.\else \cite{CAV22} -- Probability Distributions via GeneratingfunctionologY.\fi~The current implementation consists of about 6,000 LOC.
The two new features are the implementation of the observe semantics and normalization, as well as a parameter-synthesis approach for finding suitable parameters of distributions to satisfy the invariant condition.

\subsection{Implementation of \toolname{Prodigy}}

\begin{figure}[t]
	\centering
	\def\radius{.6mm} 
	\resizebox{.72\linewidth}{!}{
		\begin{tikzpicture}[
			box/.style args = {#1/#2}{draw, align=flush center, rounded corners, text width=#1, minimum height=#2},
			box/.default = 1.4cm/10mm,
			unframedbox/.style args = {#1/#2}{align=flush center, text width=#1, minimum height=#2},
			unframedbox/.default = 1.4cm/10mm
			]
			\tikzstyle{sarrow} = [draw, ->,>=stealth'],
			\tikzstyle{darrow} = [draw, <->,>=stealth']
			
			\begin{scope}
				\node[box=3.62cm/8mm] (engine) at (0,0) {exact inference engine};
				\node[box=3.62cm/8mm,below=6mm of engine] (interface) {distribution interface};
				\node[box=1cm/6mm, below left=4mm and -1.225cm of interface] (ginac) {\ginac};
				\node[box=1cm/6mm, below right=4mm and -1.225cm of interface] (sympy) {\sympy};
				\node[box=1.5cm/12mm, above left=6mm and -1.725cm of engine] (parser) {parser};
				\node[box=1.6cm/12mm, above right=6mm and -.6cm of engine] (eqcheck) {equivalence\\checker};
				\node[unframedbox=2.2cm/10mm, left=6mm of interface]  (prior) {prior dist.~$G$};
				\node[unframedbox=2.48cm/10mm]  (problem) at  (parser -| prior) {~\\[1mm]\cpgcl~program $P$\\[-1mm]$+$\\[-1mm]queries\\[-1mm]$+$\\[-1mm]invariant $I$};
				\node[unframedbox=2.6cm/12mm, right=16mm of engine] (output) {post.\ dist.~$\sem{P}(G)$\\[-1mm]$+$\\[-1mm]answer to queries};
				\node[unframedbox=1.4cm/6mm] (invalidinv) at  (eqcheck -| output) {$\sem{P} \neq \sem{I}$\\[-1mm]$+$\\[-1mm]counterexample};


				\path[sarrow] (problem) -- (parser);
				\path[sarrow] (parser) -- (eqcheck);
				\path[sarrow] (eqcheck) -- (invalidinv) node [above,pos=0.5] {\xmark};
				
				\path[sarrow] (prior) -- (interface);
				\path[darrow] (ginac.north) -- (ginac.north |- interface.south);
				\path[darrow] (sympy.north) -- (sympy.north |- interface.south);
				\path[darrow] (engine) -- (interface);
				
				\path[darrow,name path=arrow 1] (eqcheck.south)++(.45,0) |- (interface.east);
				\path[sarrow] (parser.south) -- (parser.south |- engine.north);
				\path[sarrow,stealth'-] (engine.east)++(0,.2) -| (eqcheck.south);
				\node[] (cmark) at (1.97,.65) {\cmark};
				
				\path[name path=arrow 2] (engine) -- (output);
				
				\path [name intersections={of = arrow 1 and arrow 2}];
				\coordinate (S)  at (intersection-1);
				
				\path[name path=circle,thin] (S) circle(\radius);
				
				\path [name intersections={of = circle and arrow 2}];
				\coordinate (I2)  at (intersection-2);
				\coordinate (I1)  at (intersection-1);
				
				\draw (engine) -- (I1);
				\draw[sarrow] (I2) -- (output);
				
				\tkzDrawArc[-,color=black,shorten >=0pt](S,I2)(I1);
				
			\end{scope}
			
			\begin{pgfonlayer}{background}
				\filldraw [line width=4mm,join=round,black!6]
				(eqcheck.north  -| output.east)  rectangle (sympy.south  -| problem.west);
				
				\filldraw [line width=4mm,join=round,blue!15]
				(eqcheck.north  -| eqcheck.east)  rectangle (ginac.south  -| ginac.west);
				\node[] (prodigy)  at (2.65,-2.85) {\toolname{Prodigy}};
			\end{pgfonlayer}
		\end{tikzpicture}
		}
	\caption{A sketch of the \toolname{Prodigy} workflow.}
	\label{fig:prodigy}
\end{figure}
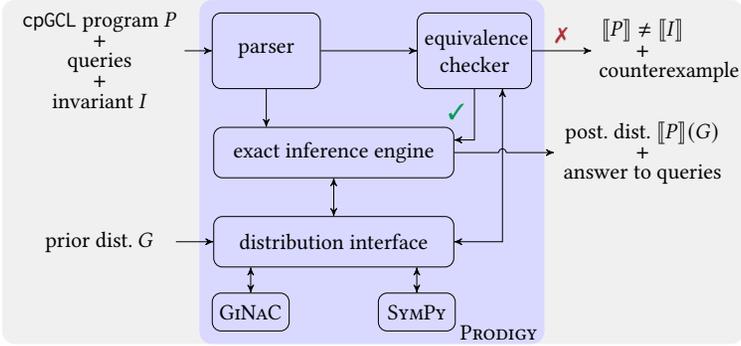

\toolname{Prodigy} implements exact inference for \cpgcl programs;
its high-level structure is depicted in \Cref{fig:prodigy}.
Given a \cpgcl program $P$ (optionally with queries to the output distribution, e.g., expected values, tail bounds and moments) together with a prior distribution $G$, \toolname{Prodigy} parses the program, performs PGF-based distribution transformations (via the inference engine), and finally outputs the posterior distribution $\sem{P}(G)$ (plus answers to the queries, if any). For the distribution transformation, \toolname{Prodigy} implements an internal interface acting as an abstract datatype for probability distributions in the form of formal power series.
Such an abstraction allows for an easy integration of alternative distribution representations (not necessarily related to generating functions) and various computer algebra systems (CAS) in the backend.
\toolname{Prodigy} currently supports \sympy \cite{SymPy} and \ginac \cite{DBLP:journals/jsc/BauerFK02,ginac}.
When (UAST) loops $L=\WHILEDO{B}{P'}$ are encountered, \toolname{Prodigy} asks for a user-provided invariant $I$ and then performs the equivalence check such that it can either infer the output distribution or conclude that $\sem{L} \neq \sem{I}$ while providing counterexamples $\sigma$ such that $\Phi_{B,P'}(\sem{I})(\sigma) \neq \sem{I}(\sigma)$.
In the absence of an invariant, \toolname{Prodigy} is capable of computing under-approximations of the posterior distribution by unfolding the loop up to a specified accuracy or number of loop unrollings.

\subsection{Benchmarks}

We collected a set of 37 benchmarks, 16 of them related to inferring distributions for loopy programs.
This set consists of examples provided by $\lambda$-PSI \cite{DBLP:conf/pldi/GehrSV20}, \textsc{Genfer} \cite{zaiser2023exact}, and \toolname{Prodigy}.
All experiments were evaluated on MacOS Sonoma 14.0 with a 2,4 GHz Quad-Core Intel Core i5 and 16GB RAM. 
For each benchmark, we run \toolname{Prodigy} with both CAS backends, i.e., \sympy and \ginac.
For loop-free benchmarks, \toolname{Prodigy} is compared against $\lambda$-PSI\footnote{We used the commit \texttt{9db68ba9581b7a1211f1514e44e7927af24bd398}.} and \textsc{Genfer}\footnote{We used the commit \texttt{5911de13f16bc3c28703f1631c5c4847f9ebac9a}.} -- the two closest tools (among those in \Cref{sec:related_work}). As \toolname{Prodigy} is an exact inference engine, all tools are run using \emph{exact arithmetic}.
The initial prior distribution is $1$ which means all variables are initialized to $0$ with probability $1$ and no \codify{observe}-violations have occurred.
All timings are averaged over 20 iterations per benchmark and we measured the time used for performing inference (computing the posterior distribution).
The experiments aim to answer questions in terms of 
\begin{enumerate*}
	\item \emph{Effectiveness:} Can \toolname{Prodigy} effectively do exact inference on the selected benchmarks, including equivalence checking and invariant synthesis for programs with loops?
	\item \emph{Efficiency:} How does \toolname{Prodigy} compare to the most related tools? How do the CAS backends \sympy and \ginac compare to each other? 
\end{enumerate*}

\subsection{Experimental Results}
\paragraph{General observations.}
\Cref{tab:benchmarks,tab:benchmarks_loop} summarize our experimental results.
Our approach is capable of computing posterior distributions for a variety of programs in less than 0.1 seconds.
For loop-free benchmarks, \emph{exact} Bayesian inference based on generating functions (\textsc{Genfer}, \toolname{Prodigy}) performs better than $\lambda$-PSI on \emph{discrete} probabilistic programs with \textsc{Genfer} being the fastest in most instances.
Regarding the timings for \toolname{Prodigy} only, the \ginac backend is generally about two orders of magnitude faster.
\toolname{Prodigy} is the only tool that is able to deal with unbounded loopy programs.

\begin{table}[t]
	\caption{Benchmarks of loop-free programs; timings are in seconds.
	}
	\label{tab:benchmarks}
	\centering
	\begin{tabular}{lcclllll}
		\toprule
		\multirow{2}{*}{Program}  & \multirow{2}{*}{$\infty$}& \multirow{2}{*}{$p$}& \multicolumn{2}{c}{\toolname{Prodigy}} & \multicolumn{2}{c}{\textsc{$\lambda$PSI}} & \multirow{2}{*}{\textsc{Genfer}}\\ 
		\cmidrule(lr){4-5}\cmidrule(lr){6-7}
		&&& \sympy & \ginac & symbolic & dp&\\
		\midrule
		\texttt{burgler\_alarm} & \phantom{\xmark} & \phantom{\xmark} & 1.988 & 0.012 & 0.055& 0.008 & \textbf{0.002}\\
		\texttt{caesar} & \phantom{\xmark}& \bmark &  8.377 & \textbf{0.025} & 1.152 & 0.051 & ---\\
		\texttt{digitRecognition} & \phantom{\xmark} & \phantom{\bmark} & Err.\footnote{Exceeding \sympy internal limits for parsing.} & 34.685 & 96.283 & 2.818 & \textbf{0.137} \\
		\texttt{dnd\_handicap} & \phantom{\xmark}& \phantom{\xmark} & 7.760 & 0.032& 0.094 & 0.039& \textbf{0.006}\\
		\texttt{evidence1}& \phantom{\xmark}& \phantom{\xmark} & 0.348 & 0.002& 0.011 & 0.002 & \textbf{<0.001}\\
		\texttt{evidence2}& \phantom{\xmark}& \phantom{\xmark} &0.413 & 0.003& 0.014 & 0.002 & \textbf{0.001}\\	
		\texttt{function}& \phantom{\xmark}& \phantom{\xmark} & 0.338 & 0.002& 0.001 &\textbf{<0.001}& 0.003\\
		\texttt{fuzzy\_or}& \phantom{\xmark}& \phantom{\xmark} & 67.048 & 0.227& 8.779 & 4.797& \textbf{0.025}\\
		\texttt{grass}& \phantom{\xmark}& \phantom{\xmark} & 6.706 & 0.021& 0.481 & 0.089 & \textbf{0.006}\\
		\texttt{infer\_geom\_mix}& \bmark& \phantom{\xmark} & 13.723 & 0.031 & 0.199 & \textbf{0.003}& 0.139\\
		\texttt{lin\_regression\_unbiased}& \phantom{\xmark}& \phantom{\xmark} & 6.700 &\textbf{0.014}& 0.056 & 0.016 & 0.918 \\
		\texttt{lucky\_throw}&\phantom{\xmark}& \phantom{\xmark} &Err.\footnote{Reached maximum recursion limit} &1.560& TO & 1.565&\textbf{0.455}\\
		\texttt{max}& \phantom{\xmark}&\phantom{\xmark} &0.618 & 0.005& 0.020 & 0.003 & \textbf{0.001}\\
		\texttt{monty\_hall}& \phantom{\xmark}&\phantom{\xmark} &2.927 & 0.033 &0.063 & \textbf{0.004} & 0.006\\
		\texttt{monty\_hall\_nested}& \phantom{\xmark}&\phantom{\xmark} & 15.694 &0.140&0.525 & 0.017 & \textbf{0.025}\\
		\texttt{murder\_mystery}& \phantom{\xmark}&\bmark &0.615 &0.004& 0.020 & \textbf{0.003}& ---\\
		\texttt{pi}& \phantom{\xmark}&\phantom{\xmark} &90.931 & \textbf{0.094} & TO & 0.103& ---\\
		\texttt{piranha}& \phantom{\xmark}&\phantom{\xmark} &0.379 & 0.003 & 0.011 & 0.002& \textbf{<0.001}\\
		\texttt{telephone\_operator}& \bmark& \phantom{\xmark} & 1.249 & 0.006& $0.058^\ast$ & Err.\footnote{The \texttt{-{}-dp} strategy produces $p(x,d)=0$ which is an incorrect result.}& 0.006\\
		\texttt{telephone\_operator\_param}& \bmark& \bmark &5.880& \textbf{0.017} & $0.108^\ast$ & 0.007 & --- \\
		\texttt{twocoins}& \phantom{\xmark}&\phantom{\xmark} &0.493 & 0.004&0.011 & 0.002 & \textbf{<0.001}\\
		\bottomrule
	\end{tabular}
\end{table}

\paragraph{Results for loop-free programs.}
Whereas our focus is on programs featuring unbounded loops, we compared \toolname{Prodigy} to $\lambda$-PSI and \textsc{Genfer} for loop-free benchmarks.
 \Cref{tab:benchmarks} lists the results.
The column Program lists the benchmarks.
The next column ($\infty$) marks the occurrence of samplings from infinite-support distributions in the benchmark.
Column $p$ indicates the presence of symbolic parameters.
Finally, columns \sympy, \ginac, symbolic, dp and \textsc{genfer} list run-times in seconds for the individual backends of \toolname{Prodigy}, \textsc{$\lambda$-PSI}, and the tool \textsc{genfer} respectively.
Here, dp represents the dynamic programming backend of \textsc{$\lambda$-PSI} invoked by using the option \texttt{{-}{-}dp}, and \textsc{Genfer} using exact arithmetic (\texttt{-{}-rational}).
The timing in boldface marks the fastest variant.
The acronym TO stands for time-out, i.e., did not terminate within the time limit of 90 seconds.
Entries consisting of \enquote{---} indicate the lack of support for this benchmark instance.
Timings marked with $^\ast$ refer to results by $\lambda$-PSI which contain integral expressions that we like to avoid, however $\lambda$-PSI is still able to compute all moments exactly.

Our experiments show that \textsc{Genfer} can be up to two orders of magnitude faster.
We emphasize that PSI and Genfer are symbolic engines tailored to solving loop-free inference tasks.
Despite this, it turns out that we oftentimes are on par.
For the \texttt{digitRecognition} example (the most prominent outlier), the speedup of \textsc{Genfer} mostly originates from an optimization in computing the observe-violation probabilities.
For loop-free programs, where termination is inherent by design, the necessity to precisely track \codify{observe}-violation probabilities is avoided.
Consequently, the observation-violation probability can be computed as the \enquote{missing} probability mass in the final distribution.
While this methodology is effective in loop-free scenarios, it does not apply to loopy programs and hence was not implemented in \toolname{Prodigy}.

\citet{Zaiser23} see automatic differentiation as the key ingredient enabling the fast results of \textsc{Genfer}.
Automatic differentiation in the sense of computing $n$-th derivatives at specific points is done by both \textsc{Genfer} and \toolname{Prodigy}.
Whereas \citet{Zaiser23} employ a custom implementation, we rely on well-established implementations from \sympy and \ginac.
In fact, the actual differentiation implementation can be exchanged freely.
\toolname{Prodigy}'s support for loops and parameter synthesis seamlessly integrate with any differentiation method while maintaining the functionality and capitalizing on potential speed enhancements.
Moreover, \textsc{Genfer} is unable to deal with non-linear observations as in the \texttt{pi} benchmark.
The same holds for instances with symbolic parameters. \toolname{Prodigy} outperforms the symbolic engine of $\lambda$-PSI on almost every instance whilst \toolname{Prodigy} has a comparable performance to the dynamic programming strategy of $\lambda$-PSI.

\renewcommand{\arraystretch}{1.05}
\begin{table}[t]
	
	\caption{Exact inference results for loopy probabilistic programs (those with parameter synthesis are marked by \texttt{\_param}); timings are given in seconds.}
	\label{tab:benchmarks_loop}
	\centering
	\begin{tabular}{lllll}
		\toprule
		Program & \sympy & & \ginac &\\
		\midrule
		\texttt{dep\_bern} & 13.354 &&\textbf{0.457}&\\
		\texttt{endless\_conditioning} & 1.148 & &\textbf{0.012}&\\
		\texttt{geometric} & 3.757 & &\textbf{0.031} &\\
		\texttt{ky\_die} & 21.562 & &\textbf{0.209} &\\ 
		\texttt{n\_geometric} & 3.050 & &\textbf{0.038}&\\
		\texttt{random\_walk} & 3.439 & &\textbf{0.047}&\\
		\texttt{trivial\_iid} &6.444 & &\textbf{0.075}&\\
		\texttt{bit\_flip\_conditioning} & 31.030 && \textbf{0.322}&\\
		\midrule
		\texttt{dueling\_cowboys\_param} & 6.147&$\text{for any}~p,q$&\textbf{0.065}&$\text{for any}~p,q$\\
		\texttt{geometric\_param} & 4.888&$p=\frac{1}{3}$&\textbf{0.262}&$p=\frac{1}{3}$ \\
		\texttt{ky\_die\_param} & 36.619&$p=\frac 2 3 , q=\frac 1 2$ & \textbf{1.298}&$p=\frac 2 3 , q=\frac 1 2$ \\
		\texttt{negative\_binomial\_param} &2.814&$\text{for any}~p$& \textbf{0.047}&$\text{for any}~p$\\
		\texttt{n\_geometric\_param} & 5.365&$p=\frac{q}{3}$ & \textbf{0.133}&$p = \frac{q}{3}$\\
		\texttt{random\_walk\_param} & 5.114&$p=\frac{1}{2}$& \textbf{0.274}&$p=\frac{1}{2}$\\
		\texttt{bit\_flip\_cond\_param} & 58.599&$p=\frac{13}{28} , q = \frac 3 7, r= \frac 2 7$ & \textbf{0.887}&$p=\frac{13}{28} , q = \frac 3 7, r= \frac 2 7$\\
		\texttt{brp\_obs\_param} & TO && \textbf{77.732}&$p=10^{-10}$\\
		\bottomrule
	\end{tabular}
\end{table}

\paragraph{Results for loopy programs.}
\Cref{tab:benchmarks_loop} depicts the empirical results for loopy programs.
The column Program lists the benchmarks.
The columns \sympy and \ginac report their run-times in seconds when used as backend of \toolname{Prodigy}.
The timing in boldface marks the fastest variant.
As these benchmarks all include loops, they are not supported by $\lambda$-PSI and \textsc{Genfer}.

Recall that reasoning about loops involves an equivalence check against a user-specified invariant program.
Finding the right invariant (if it exists in the loop-free \credip fragment) is intricate.
We support the user in discovering such invariants by allowing symbolic parameters for distributions, e.g., one can write $\geometric{p}$ where $p$ is a symbolic parameter.
 For benchmarks subject to parameter synthesis, we also provide the anticipated parameter constraints (or values) inferred automatically by \toolname{Prodigy}. 
Whenever this is the case, we point out that for the \ginac timings, discharging the resulting equation systems is achieved using \sympy solvers, which is due to the missing functionality of \ginac to solve these equation systems. Overall, \ginac is faster than \sympy by about two orders of magnitude, as is similar to the loop-free benchmarks.

It is also worth noting that \toolname{Prodigy} is potentially applicable to practical randomized algorithms beyond toy programs like random walks. These applications include loop-free benchmarks such as \texttt{digitRecognition} for recognizing written digits based on observed data samples, as well as the unbounded loopy program modeling the bounded retransmission protocol (\texttt{brp\_obs\_param}):
\begin{example}[Bounded Retransmission Protocol]\label{ex:brp}
	Prog.~\ref{prog:brp} describes a conditioned variant of the bounded retransmission protocol (BRP) \cite{DBLP:conf/papm/DArgenioJJL01,DBLP:conf/tacas/BatzCJKKM23} which attempts to transmit $\progvar{s}$ packets over a lossy channel, where each individual packet gets lost with probability 1\%. The transmission is considered successful, if \emph{none} of the packets needs more than 4 retransmissions.
	Additionally, we observe that all but the last 9 packets are received successfully without any additional resends.
	\cref{prog:mc_brp} illustrates the protocol as a Markov chain.
	Notice that the number of packets to be sent is parametrized by the (possibly infinite-support) initial distribution of $\progvar{s}$ and $\progvar{f}$ -- modeling an \emph{infinite family} of finite-state Markov chains -- and hence renders techniques like probabilistic model checking \cite{DBLP:conf/lics/Katoen16} infeasible.

Provided with a suitable invariant (cf.~\cref{apx:prodigy}) with parameter $p$ in the probabilities, \toolname{Prodigy} infers that, with $p = 10^{-10}$, Prog.~\ref{prog:brp} is equivalent to this invariant, thereby yielding the exact output distribution (for any initial distribution of $\progvar{s}$ with rational closed form) in the form of a PGF.
From this PGF, we can derive, e.g., with input $\progvar{s} \sim \geometric{\sfrac{1}{2}}$, the \emph{transmission-failure probability} of BRP, i.e., the probability that Prog.~\ref{prog:brp} terminates with $f > 4$ is around $9.9789 \times 10^{-11}$ (see \cref{apx:prodigy}).

From a syntactic point of view, the BRP may seem intricate.
Yet semantically, it represents the structure of the original program's underlying Markov chain (\cref{prog:mc_brp}) in a straightforward manner.
For all but the last 9 packets, no transmission attempt is allowed to fail.
If starting with at most 9 packets to send in total, the initial state might already indicate some failed attempts for the first packet to transmit.
In this case, the first packet sent has less than 5 retries to successfully complete the transmission.
Afterwards, for each of the remaining packets, transmission either fails with some probability $p$ or is successful and continues with the next packet.
\end{example}

\begin{program}[t]
	\centering
	\begin{minipage}[b]{.37\textwidth}
		\begin{align*}
		&\WHILE{\progvar{s} > 0 \wedge \progvar{f} \leq 4} \\
		& \quad \pcomment{packet loss}\\
		& \quad	\{ \OBSERVE{\progvar{s} \leq 9}\fatsemi \ASSIGN{\progvar{f}}{\progvar{f} + 1}\}\\
		& \quad [\sfrac{1}{100}]\\
		& \quad \pcomment{packet received}\\
		&\quad \{\ASSIGN{\progvar{f}}{0}\fatsemi \ASSIGN{\progvar{s}}{\progvar{s}-1}\}\\
		&\}
		\end{align*}
		\caption{A conditioned variant of BRP.}
		\label{prog:brp}
	\end{minipage}
	\hfill
	\begin{minipage}[b]{.62\textwidth}
		\resizebox{\textwidth}{!}{
		\begin{tikzpicture}[every initial by arrow/.style={>=stealth'}]
		\node (s0) [state] {$(0,0)$};
		\node (s1) [state, right=of s0] {$(1,0)$}; 
		\node (s2) [state, right=of s1] {$(2,0)$};
		\node (s3) [state, draw=none, right=of s2] {$\ldots$};
		\node (s4) [state, right=of s3] {$(9,0)$};
		
		\node (s1f1) [state, above=of s1] {$(1,1)$};
		\node (s2f1) [state, above=of s2] {$(2,1)$};
		\node (s4f1) [state, above=of s4] {$(9,1)$};
		
		\node (s1f2) [state, draw=none, above=of s1f1] {$\vdots$};
		\node (s1f5) [state, above=of s1f2] {$(1,5)$};
		
		\node (s2f2) [state, draw=none, above=of s2f1] {$\vdots$};
		\node (s2f5) [state, above=of s2f2] {$(2,5)$};
		
		\node (s4f2) [state, draw=none, above=of s4f1] {$\vdots$};
		\node (s4f5) [state, above=of s4f2] {$(9,5)$};
		
		\node (s6) [state, right=of s4, scale=0.9] {$(10,0)$};
		\node (s7) [state, draw=none, right=of s6] {$\ldots$};
		\node (s8) [state, right=of s7, scale=0.8] {$(s\!-\!1, 0)$};
		\node (s9) [initial right, initial text=,state, above=of s8] {$(s,f)$};
		
		\draw[->,>=stealth'] (s0) edge[loop left] node {$1$} (s0);
		\draw[->,>=stealth'] (s1) edge node[above] {$99\%$} (s0); 
		\draw[->,>=stealth'] (s2) edge node[above] {$99\%$} (s1);
		\draw[->,>=stealth'] (s3) edge node[above] {$99\%$} (s2);
		\draw[->,>=stealth'] (s4) edge node[above] {$99\%$} (s3);
		\draw[->,>=stealth'] (s1f1) edge node[above] {$99\%$} (s0);
		\draw[->,>=stealth'] (s2f1) edge node[above] {$99\%$} (s1);
		\draw[->,>=stealth'] (s4f1) edge node[above] {$99\%$} (s3);
		\draw[->,>=stealth'] (s1f2) edge[dashed, in=60, out=210] (s0);
		\draw[->,>=stealth'] (s2f2) edge[dashed, in=60, out=210] (s1);
		\draw[->,>=stealth'] (s4f2) edge[dashed, in=60, out=210] (s3);
		\draw[->,>=stealth'] (s1f5) edge[loop above] node {$1$} (s1f5);
		\draw[->,>=stealth'] (s2f5) edge[loop above] node {$1$} (s2f5);
		\draw[->,>=stealth'] (s4f5) edge[loop above] node {$1$} (s4f5);
		\draw[->,>=stealth'] (s1) edge node[left] {$1\%$}  (s1f1);
		\draw[->,>=stealth'] (s2) edge node[left] {$1\%$}  (s2f1);
		\draw[->,>=stealth'] (s4) edge node[left] {$1\%$}  (s4f1);
		\draw[->,>=stealth'] (s1f1) edge node[left] {$1\%$}  (s1f2);
		\draw[->,>=stealth'] (s2f1) edge node[left] {$1\%$}  (s2f2);
		\draw[->,>=stealth'] (s4f1) edge node[left] {$1\%$}  (s4f2);
		\draw[->,>=stealth'] (s1f2) edge node[left] {$1\%$}  (s1f5);
		\draw[->,>=stealth'] (s2f2) edge node[left] {$1\%$}  (s2f5);
		\draw[->,>=stealth'] (s4f2) edge node[left] {$1\%$}  (s4f5);

		\draw[->,>=stealth'] (s9) edge[loop above] node {$1\%$} (s9);
		\draw[->,>=stealth'] (s9) edge[bend left] node[right] {$99\%$}  (s8);
		\draw[->,>=stealth'] (s8) edge[bend left] node[left] {$1\%$} (s9);
		\draw[->,>=stealth'] (s8) edge node[below] {$99\%$} (s7);
		\draw[->,>=stealth'] (s7) edge node[below] {$99\%$} (s6);
		\draw[->,>=stealth'] (s7) edge [dashed,bend left, in=135] (s9);
		\draw[->,>=stealth'] (s6) edge[bend left, in=125] node[above]{$1\%$} (s9);
		\draw[->,>=stealth'] (s6) edge node[below] {$99\%$} (s4);

		\end{tikzpicture}
		}
		
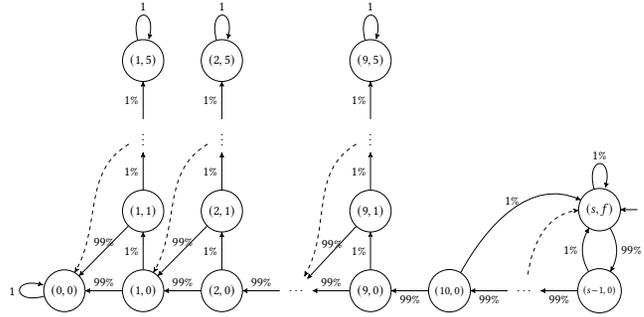
\captionof{figure}{The Markov chain illustrating Prog.~\ref{prog:brp}.}
		\label{prog:mc_brp}
	\end{minipage}
\end{program}

	\section{Limitations of Exact Inference using eFPS}
\label{sec:bottlenecks}


We discuss some limitations of the presented inference approach considering guard evaluations, non-rational probabilities and scalability.
Prog.\ \ref{prog:collatz} models a variant of the famous Collatz algorithm~\cite{andrei1998collatz}.
The Collatz conjecture states that for all positive integers $m$ there exists $n \in \N$ such that for the Collatz function $C(m) \coloneqq \sfrac{n}{2}$ for $n \equiv (0 \text{ mod } 2)$ and $3n+1$ otherwise; the $n$-th fold iteration of the function is $C^n(m) = 1$.
We have adapted the program syntax slightly and make use of the \codify{loop} statement to represent the $n$-fold repetition of a code block.
The program basically behaves as the usual Collatz function with the only exception that in the case where a number is divisible by two, we have a small chance not dividing $\progvar{x}$ by 2 but instead executing the \codify{else} branch. Note that the instruction $\progvar{x} \equiv_2 0 (\text{mod}~2 )$ still preserves rational closed forms as we can compute its semantics by $\frac{F(\bvec{X}) + F(-\bvec{X})}{2}$.
When analyzing the run-times of our tool on this program we observe surprising results: for $(n=1)$ we obtain a result in 0.010631 seconds; $(n=2)$ is computed in 0.049891 seconds and for $(n=3)$ it suddenly increases to 88.689832 seconds.
We think that this phenomenon arises from the fact that evaluating expressions like $\progvar{x} \equiv 0 ~(\textup{mod}~ 2)$ repeatedly, gets increasingly difficult as it is implemented in \toolname{Prodigy} by means of arithmetic progressions.

Another challenge is guard evaluation, i.e., filtering out the corresponding terms of a formal power series such as $\constrain{F}{B}$ for \IFSYMBOL-statements.
In case we are interested in the relation between two variables (like $\progvar{x} = \progvar{y}$) when both have marginal distributions with infinite support, \toolname{Prodigy} cannot compute the result. As an approximation heuristic it computes under-approximations of the \emph{exact} posterior distribution. Note that if either $x$ or $y$ has a finite-support marginal distribution, the posterior is computed by enumeration.
An interesting example why one cannot even strive for such a potential closed-form operation preserving rational closed forms is Prog.\ \ref{prog:non_algebraic}.
For this program, its variable $\progvar{r}$ evaluates to 1 with \emph{non-rational}, not even algebraic probability $\sfrac{1}{\pi}$ after termination \cite{DBLP:conf/soda/FlajoletPS11} -- thus beyond \credip capabilities.
An interesting open question is to determine what syntactic restrictions \emph{exactly capture rational} closed forms.

As a final observation we emphasize that \toolname{Prodigy}'s performance is proportional to the size of constants in the programs.
Assume, e.g., a guard $\progvar{x} > n$, where $n$ is a constant. For larger $n$, the closed-form operation of computing the $n$-th formal derivative takes an increasing amount of time.

\begin{program}[t]
	\begin{minipage}[b]{.49\linewidth}
		\begin{align*}
			&\ASSIGN{\progvar{x}}{\geometric{\sfrac{1}{4}}}\fatsemi\\
			&\ASSIGN{\progvar{y}}{\geometric{\sfrac{1}{4}}}\fatsemi\\
			&\ASSIGN{\progvar{t}}{\progvar{x} + \progvar{y}}\fatsemi\\
			&\PCHOICE{ \ASSIGN{ \progvar{t} }{ \progvar{t}+1 } }{ \sfrac{5}{9} }{\pskip}\fatsemi\\
			&\ASSIGN{\progvar{r}}{1}\fatsemi\\
			&\codify{loop}(3)\{\\
			& \quad \ASSIGN{\progvar{s}}{\iid{\bernoulli{\sfrac{1}{2}}}{2\progvar{t}}}\fatsemi\\
			& \quad \IF{\progvar{s} \neq \progvar{t}} \ASSIGN{\progvar{r}}{0}\}\\
			&\}
		\end{align*}
		\caption{Non-algebraic probabilities.}
		\label{prog:non_algebraic}
	\end{minipage}
	\hfill
	\begin{minipage}[b]{.49\linewidth}
		\begin{align*}	
			&\ASSIGN{\progvar{x}}{\geometric{\sfrac{1}{2}}}\fatsemi\\
			&\codify{loop}(n)\{\\
			& \quad \IF{x \equiv 0 ~(\textup{mod}~ 2)}\\
			& \quad \quad \left\{\, {\ASSIGN{\progvar{x}}{3*\progvar{x} + 1}} \,\right\}\mathrel{\left[\,\sfrac{1}{10}\,\right]}\\
			&\quad \quad\quad\left\{\, {\ASSIGN{\progvar{x}}{\sfrac{1}{2} * \progvar{x}}} \,\right\}\\
			&\quad \ELSE\\
			& \quad \quad \ASSIGN{\progvar{x}}{3 * \progvar{x} + 1}\\
			& \quad \}\\
			& \}
		\end{align*}
		\caption{Probabilistic Collatz program.}
		\label{prog:collatz}
	\end{minipage}
\end{program}


	\section{Related Work}
\label{sec:related_work}

We review a non-exhaustive list of related work in probabilistic inference, ranging from invariant-based verification techniques 
to inference techniques based on sampling and symbolic methods.

\paragraph{\bf Invariant-based verification.}
As a means to avoid intractable fixed point computations, the correctness of loopy probabilistic programs can often be established by inferring specific (inductive) bounds on expectations, called \emph{quantitative loop invariants} \cite{DBLP:series/mcs/McIverM05}. There are a variety of results on synthesizing quantitative invariants, including (semi-)automated techniques based on \emph{martingales} \cite{DBLP:conf/cav/BartheEFH16,DBLP:conf/cav/ChakarovS13,DBLP:conf/sas/ChakarovS14,DBLP:conf/popl/ChatterjeeNZ17,chatterjeeFOPP20,DBLP:journals/toplas/TakisakaOUH21}, \emph{recurrence solving} \cite{DBLP:conf/atva/BartocciKS19,DBLP:conf/tacas/BartocciKS20}, 
\emph{invariant learning} \cite{DBLP:conf/cav/BaoTPHR22},  
and \emph{constraint solving} \cite{DBLP:conf/sas/KatoenMMM10,DBLP:conf/qest/GretzKM13,DBLP:conf/atva/FengZJZX17,DBLP:conf/cav/ChenHWZ15}, particularly via \emph{satisfiability modulo theories} (SMT) \cite{DBLP:conf/tacas/BatzCJKKM23,DBLP:conf/cav/CKKMS20,DBLP:conf/cav/BatzJKKMS20}.

Alternative state-of-the-art verification approaches include \emph{bounded model checking} \cite{DBLP:conf/atva/0001DKKW16} for verifying probabilistic programs with nondeterminism and conditioning
as well as various forms of \emph{value iteration} \cite{DBLP:conf/cav/Baier0L0W17,DBLP:conf/cav/QuatmannK18,DBLP:conf/cav/HartmannsK20} for determining reachability probabilities in finite Markov models.

\paragraph{\bf Sampling-based inference.}
Most existing probabilistic programming languages implement \emph{sampling}-based inference algorithms rooted in the principles of Monte Carlo \cite{doi:10.1080/01621459.1949.10483310}, thereby yielding numerical approximations of the exact results, see, e.g., \cite{gram2021extending}. Such languages include Anglican \cite{DBLP:conf/aistats/WoodMM14}, BLOG \cite{DBLP:conf/ijcai/MilchMRSOK05}, BUGS \cite{Spiegelhalter1995BUGSB}, Infer.NET \cite{InferNET18}, R2 \cite{DBLP:conf/aaai/NoriHRS14}, Stan \cite{stan2022}, etc. In contrast, we are concerned with inference techniques that produce \emph{exact} results.

\paragraph{\bf Symbolic inference.}
In response to the aforementioned challenges \ref{issue:loop} and \ref{issue:support} in exact probabilistic inference, \citet{LOPSTR} proposed a program semantics based on \emph{probability generating functions}. This PGF-based semantics allows for exact quantitative reasoning for, e.g., deciding probabilistic equivalence \cite{CAV22} and proving non-almost-sure termination \cite{LOPSTR} for certain probabilistic programs \emph{without conditioning}.

Extensions of PGF-based approaches to programs with conditioning have been initiated in \cite{Klinkenberg23,zaiser2023exact}; the latter suggested the use of automatic differentiation in the evaluation of PGFs, but the paper addresses \emph{loop-free programs only}.
Combining conditioning and possibly non-terminating behaviors (introduced through loops) substantially complicates the computation of final probability distributions and normalization constants.
Another difference is that Zaiser et al.\ provide truncated posterior distributions together with the first four centralized moments.
We, in contrast, develop a symbolic representation of the full posterior distribution.

As an alternative to PGFs, many probabilistic systems employ \emph{probability density function} (PDF) representations of distributions, e.g., ($\lambda$)\textsc{PSI} \cite{DBLP:conf/cav/GehrMV16,DBLP:conf/pldi/GehrSV20}, \textsc{AQUA} \cite{DBLP:conf/atva/HuangDM21} and \textsc{Hakaru} \cite{DBLP:conf/flops/NarayananCRSZ16}, as well as the density compiler in \cite{DBLP:conf/popl/BhatAVG12,DBLP:journals/lmcs/BhatBGR17}. These systems are dedicated to inference for programs encoding joint (discrete-)continuous distributions with conditioning. Reasoning about the underlying PDF representations, however, amounts to resolving complex integral expressions in order to answer inference queries.
Furthermore, ($\lambda$)\textsc{PSI} admits \emph{only bounded looping behaviors}. \textsc{Dice} \cite{DBLP:journals/pacmpl/HoltzenBM20} employs weighted model counting to enable potentially scalable exact inference for discrete probabilistic programs, yet is also confined to statically bounded loops.
\citet{DBLP:conf/lics/SteinS21} proposed a denotational semantics based on Markov categories for continuous probabilistic programs with exact conditioning and bounded looping behaviors. A similar direction is taken by \citet{DBLP:conf/esop/BichselGV18}. They investigate the connections between \codify{observe}-violations, non-termination, and errors raised by, e.g., division by zero; their semantics is based on Markov kernels.
A recently proposed language \textsc{PERPL} \cite{DBLP:journals/pacmpl/ChiangMS23} compiles probabilistic programs with unbounded recursion into systems of polynomial equations and solves them directly for least fixed points using numerical methods.
A related approach by \citet{DBLP:journals/corr/StuhlmuellerG12} uses dynamic programming techniques transforming probabilistic programs with unbounded recursion into factored sum-product networks, i.e., a particular way of representing an equation system. However, this technique cannot handle infinite-support distributions.
The tool \textsc{Mora} \cite{DBLP:conf/ictac/BartocciKS20,DBLP:conf/tacas/BartocciKS20} supports exact inference for various types of Bayesian networks, but relies on a restricted form of intermediate representation known as prob-solvable loops, whose behaviors can be expressed by a system of C-finite recurrences admitting closed-form solutions.

Finally, we refer interested readers to \cite{DBLP:conf/nips/WinnerS16,DBLP:conf/icml/WinnerSS17,DBLP:conf/icml/SheldonWS18} for a related line of research from the machine learning community, which exploits PGF-based exact inference -- not for probabilistic programs -- but for dedicated types of graphical models with latent count variables.
	\section{Conclusion}
\label{sec:conclusion}
We have presented an exact Bayesian inference approach for probabilistic programs with (possibly unbounded) loops and conditioning.
The core of this approach is a denotational semantics that symbolically encodes distributions as probability generating functions.
We showed how our PGF-based exact inference facilitates (semi-)automated inference, equivalence checking, and invariant synthesis of probabilistic programs.
Our implementation \ifanonymous\else in \toolname{Prodigy}\fi shows promise: It can do exact inference for  various infinite-state loopy programs and exhibits comparable performance to state-of-the-art exact inference tools over loop-free benchmarks.

The possibility to incorporate symbolic parameters in GF representations can enable the application
of well-established optimization methods, e.g., maximum-likelihood estimations and parameter fitting,
to probabilistic inference. Characterizing the family of programs and invariants which admit a potentially complete \esop-based synthesis approach would be of particular interest.
Additionally, future research directions include extending exact inference to continuous distributions by utilizing characteristic functions as the continuous counterpart to PGFs. Furthermore, there is an intriguing connection to be explored between quantitative reasoning about loops and the positivity problem of recurrence sequences \cite{DBLP:conf/icalp/OuaknineW14}, which is induced by loop unfolding.

%
%
%
%

	\ifanonymous
	\else
		\begin{acks}
			%
			Lutz Klinkenberg and Joost-Pieter Katoen are supported by ERC AdG Grant 787914; Darion Haase is supported by the DFG RTG 2236 UnRAVeL; Mingshuai Chen is supported by the ZJNSF Major Program under grant No.~LD24F020013 and by the ZJU Education Foundation's Qizhen Talent program. The authors would like to thank the anonymous reviewers for their constructive feedback on this article and Leo Mommers for his assistance in producing the benchmark results and his work on part of the implementation.
		\end{acks}
	\fi

	\section*{Data-Availability Statement}
	The software that supports~\cref{sec:prodigy} is available on Zenodo~{\ifanonymous(link hidden for anonymity)\else\cite{klinkenberg_2024_10782412}\fi}.

	\bibliography{references}

	\newpage
	\appendix
\section*{Appendix}

\section{Domain Theory}
\label{apx:domaintheory}

\textbf{Notation.}
The set of natural numbers, including 0 is denoted by $\N$.
$\R_{\geq 0}$ denotes the set of non-negative real numbers.
For any sets $D$ and $D'$, we write $(D \to D')$ as the set of functions $\{f\colon D \to D'\}$.
We write vectors in bold-face notations like $\bvec{X}$ for $(X_1, \ldots, X_k)$ and $\bvec{1} = (1,\ldots, 1)$ where the dimension is clear from the context.
We sometimes use Lambda calculus notations describing anonymous functions, e.g. we write $\lambda x.~ x^2$ for a function that maps $x \mapsto x^2$.
Multivariate partial derivatives are compactly denoted by $\partial_x^i f \coloneqq \frac{\partial^i f}{\partial x^i}$.

\begin{definition}[Partial Order]
	\label{apx:partial_order}
	A \emph{partial order} $(D,\sqsubseteq)$ is a set $D$ along with a binary relation $\sqsubseteq ~~ \subseteq (D \times D)$ fulfilling the following properties:
	\begin{enumerate}
		\item Reflexivity: $\forall d \in D.~ d \sqsubseteq d$.
		\item Antisymmetry: $\forall d, d' \in D.~ d\sqsubseteq d' \wedge d' \sqsubseteq d \implies d = d'$.
		\item Transitivity: $\forall d, d', d'' \in D.~ d\sqsubseteq d' \wedge d' \sqsubseteq d'' \implies d\sqsubseteq d''$. 
	\end{enumerate}
\end{definition}

\begin{definition}[$\omega$-Complete Partial Order]
	\label{a[x:complete_lattice]}
	An \emph{$\omega$-complete partial order} (\emph{$\omega$-cpo}) is a partial order $(D, \sqsubseteq)$ such that
	\begin{itemize}
		\item there is a least element $\bot \in D$, and
		\item for all ascending $\omega$-chains, i.e., every set $S = \{s_n \mid n \in \N\} \subseteq D$ such that $s_0 \sqsubseteq s_1 \sqsubseteq s_2 \sqsubseteq \ldots, $ $S$ has a \emph{supremum} denoted by $\sup S \in D$ (sometimes also $\bigsqcup S$).
	\end{itemize}
	An element of the domain $D$ is called an \emph{upper bound} of $S$ if and only if $\forall s \in S.~ s \sqsubseteq d$.
	Further, $d$ is the \emph{least} upper bound of $S$ if and only if $d \sqsubseteq d'$ for every upper bound $d'$ of $S$.
\end{definition}

\begin{definition}[Monotonic Function]
	\label{apx:mono_func}
	Let $(D, \sqsubseteq)$ and $(D', \sqsubseteq')$ be partial orders.
	A function $f \colon D \to D'$ is \emph{monotonic} if and only if:
	\begin{align*}
		\forall d, d' \in D.~ d\sqsubseteq d' \impliesqq f(d) \sqsubseteq' f(d')~. 
	\end{align*}
\end{definition}

\begin{definition}[Continuous Functions]
	\label{apx:cont_fun}
	Let $(D, \sqsubseteq)$ and $(D', \sqsubseteq')$ be $\omega$-cpos.
	A function $f \colon D \to D'$ is \emph{Scott-continuous} if and only if for every $\omega$-chain $S$, it holds that:
	\begin{align*}
		\sup \{f(s) \mid s \in S\} \quad =\quad  f(\sup S)~. 
	\end{align*}
\end{definition}

\begin{lemma}[Continuous Functions are Monotone]
	\label{apx:cont_are_mono}
	Let $(D, \sqsubseteq)$ and $(D', \sqsubseteq')$ be $\omega$-cpos, and $f\colon D \to D'$ be a continuous function.
	Then $f$ is monotonic.
\end{lemma}
\begin{proof}
	Let $d, d' \in D$ such that $d \sqsubseteq d'$.
	\begin{align*}
		& d \sqsubseteq d'\\
		\implies& \sup \{ d, d'\} = d'\\
		\implies& \sup \{f(d), f(d')\} = f(\sup\{d, d'\}) = f(d') \tag{Scott-Cont.\ of $f$}\\
		\implies& f(d) \sqsubseteq' \sup \{f(d), f(d')\} = f(d')\tag*{\qedhere}
	\end{align*}
\end{proof}

\begin{lemma}[Lifting of Partial Orders]
	\label{apx:order_liftings}
	Let $(D', \preceq')$ be a partial order and let $\sqsubseteq'$ be a point-wise lifting of $\preceq'$, i.e., for an arbitrary domain $D$ and any $f,g \in (D \to D')$, $f \sqsubseteq' g$ if and only if $\forall d \in D. ~ f(d) \preceq' g(d)$.
	Then, $(D \to D', \sqsubseteq')$ is a partial order. 
\end{lemma}
\begin{proof}
	Let $f,g,h \in (D \to D')$.
	We need to show that $\sqsubseteq$ is a partial order, i.e., it is reflexive, antisymmetric and transitive.
	\begin{align*}
	\mathbf{Reflexivity:}&& & \forall d \in D. ~ f(d) \preceq' f(d) \tag{refl. of $\prec'$}\\
	&&\implies& f\sqsubseteq' f\\
	\mathbf{Transitivity:}&& &f \sqsubseteq' h ~\text{and}~ h \sqsubseteq' g\\
	&& \implies& \forall d \in D. ~ f(d) \preceq' h(d) ~\text{and}~ h(d) \preceq' g(d)\\
	&& \implies& \forall d \in D.~ f(d) \preceq' g(d) \tag{trans. of $\preceq'$}\\
	&& \implies& f \sqsubseteq g'\\
	\mathbf{Antisymmetry:}&& & f \sqsubseteq' g ~\text{and}~ g\sqsubseteq' g\\
	&& \implies& \forall d \in D.~ f(d) \preceq' g(d) ~ \text{and} ~ g(d) \preceq' f(d) \\
	&& \implies& \forall d \in D.~ f(d) = g(d) \tag{antisym. of $\preceq'$}\\
	&& \implies& f = g\tag*{\qedhere}
	\end{align*}
\end{proof}

\begin{lemma}[Point-Wise Lifting of $\omega$-CPOs]
	\label{apx:lattice_liftings}
	Let $(D', \preceq')$ be a $\omega$-cpo and $\sqsubseteq'$ be a point-wise lifting of $\preceq'$, i.e. for an arbitrary domain $D$ and any $f, g \in (D \to D')$, let 
	\[
		f \sqsubseteq' g \qquad\textnormal{iff}\qquad \forall f \in D.~ f(d) \preceq' g(d)~.
	\]
	Then $(D \to D', \sqsubseteq')$ is an $\omega$-cpo.
\end{lemma}
\begin{proof}
	We claim that every $\omega$-chain $S = \{f_i \mid i \in \N\} \subseteq (D \to D')$ has a least upper bound given by
	\[
		\sup S = \lambda d. \sup S_d, \quad \text{where}\quad S_d \coloneq \{f(d) \mid f \in S \} \subseteq D'~.
	\]
	For every $d \in D, S_d$ is again an $\omega$-chain, because $f_0(d) \preceq' f_1(d) \preceq' \ldots$ by the point-wise definition of $\sqsubseteq'$.
	First we show that $\sup S$ is an upper bound, as for every $f \in S$
	\begin{align*}
		& \forall d \in D. \quad f(d) \preceq' \sup S_d = (\sup S)(d)\\
		\implies & f \sqsubseteq' \sup S
	\end{align*}
	Second, $\sup S$ is the \emph{least} upper bound.
	Therefore, let $\hat{f}$ be an upper bound of $S$.
	\begin{align*}
	& \forall f \in S. \quad f \sqsubseteq' \hat{f}\\
	\implies & \forall f \in S. \quad \forall d \in D.~ f(d) \preceq' \hat{f}(d)\\
	\implies & \forall d \in D. \quad (\sup S)(d) = \sup S_d \preceq' \hat{f}(d)\\
	\implies & \sup S \sqsubseteq' \hat{f} \tag*{\qedhere}
	\end{align*}
\end{proof}

\begin{theorem}[Fixed Point Theorems~\textnormal{\cite{DBLP:journals/ipl/LassezNS82,Abramsky94domaintheory}}]
	\label{apx:fixedpoints}
	Let $f\colon D \to D$ be a continuous function on an $\omega$-cpo $(D, \sqsubseteq)$.
	Then $f$ possesses a least fixed point denoted \lfp $f$, which is given by:
	\[\lfp\ f = \sup \{f^n(\bot) \mid n \in \N\}~,\qquad \text{where}\]
	$f^n$ denotes the $n$-fold application of $f$, and $\bot = \sup \emptyset$ is the least element of $D$.
\end{theorem}

\section{Semantics Using EFPS}
\label{apx:semantics}

\begin{corollary}[Partial Orders over \epgf]
	\label{apx:efps_po}
	$(\epgf, \preceq)$ as well as the point-wise lifting on functions $(\epgf \to \epgf, \sqsubseteq)$ are partial orders.
\end{corollary}
\begin{proof}
	Consider the coefficient function $\coef{\cdot}{F} \in (\N^k \cup \{\lightning\} \to \R_{\geq 0})$ which uniquely determines the ePGF $F$.
	$\abs{F}$ for PGF is bounded by 1, hence in fact $\coef{\cdot}{F} \in (\N^k \cup \{\lightning\} \to [0,1]_\R)$.
	We think of the order $\preceq$ as acting on the domain $(\N^k \cup \{\lightning\} \to [0,1]_\R)$.
	Thus, $\preceq$ can be interpreted as the point-wise lifting of the (total) order $\leq$ on $[0,1]_\R$, i.e., $(\epgf, \preceq)$ is a partial order by applying \Cref{apx:order_liftings}.
	Since $\sqsubseteq$ is a point-wise lifting of $\preceq$, we can argue analogously for $(\epgf \to \epgf, \sqsubseteq)$.
\end{proof}

\begin{corollary}[$\omega$-CPOs over \epgf]
	\label{apx:efps_lattice}
	Both partial orders $(\epgf, \preceq)$ and $(\epgf \to \epgf, \sqsubseteq)$ are $\omega$-cpos.
\end{corollary}
\begin{proof}
	Analogously to the proof of \Cref{apx:efps_po} we note that $\preceq$ is a point-wise lifting of $\leq$ on $[0,1]_\R$, and $\sqsubseteq$ is a point-wise lifting on $\preceq$. 
	Therefore applying \Cref{apx:lattice_liftings} twice yields the claimed result. 
\end{proof}


\begin{lemma}[Continuity of $\Phi_{B,P}$]
	Let $P$ be a cpGCL program and let $B$ be a Boolean guard. The characteristic functional $\Phi_{B,P}$ is continuous on the domain $\left(\epgf\rightarrow\epgf, \sqsubseteq\right)$.
\end{lemma}
\begin{proof}
	\begin{align*}
	    \Phi_{B,P}(\sup S) &= \Phi_{B,P}(\lambda F.~\sup \left\{\psi(F) \mid \psi\in S\right\}) \\
	&= \lambda F.~ \coef{\lightning}{F}\cdot X_\lightning + \angles{F}_{\neg B} \\
	&\qquad + (\lambda F.~  \sup\,\left\{\psi(F) \mid \psi\in S\right\})(\sem{P}(\angles{F}_B)) \tag{Def. $\Phi_{B,P}$}\\
	&= \lambda F.~ \coef{\lightning}{F}\cdot X_\lightning + \angles{F}_{\neg B} + \sup\,\left\{\psi(\sem{P}(\angles{F}_B)) \mid \psi\in S\right\} \tag{Evaluate inner $\lambda$-function}\\
	&= \lambda F.~ \sup\,\left\{\coef{\lightning}{F}\cdot X_\lightning + \angles{F}_{\neg B} + \psi(\sem{P}(\angles{F}_B)) \mid \psi\in S\right\} \tag{Include constants in $\sup$}\\
	&= \sup\,\left\{\lambda F.~ \coef{\lightning}{F}\cdot X_\lightning + \angles{F}_{\neg B} + \psi(\sem{P}(\angles{F}_B)) \mid \psi\in S\right\} \tag{$\sup$ defined point-wise}\\
	&= \sup\,\left\{\Phi_{B,P}(\psi) \mid \psi \in S\right\} \tag{Def. $\Phi_{B,P}$}
\end{align*}
\end{proof}

\begin{lemma}[Continuity of Auxiliary Functions]\label{lem:contaux}
	For all $\sigma\in\N^k \cup \{\lightning\}$ and Boolean guards $B$, the following functions are continuous:
	\begin{enumerate}
		\item the coefficient function $\coef{\sigma}{}$ 
		\item the restriction $\angles{\cdot}_B$ 
		\item the mass $\abs{\cdot}$ 
	\end{enumerate}
\end{lemma}

\begin{proof}
	1 and 2 follow directly from the coefficient-wise definition of $\sup$ on $\epgf$.
	For 3, let $S = \left\{F_i \mid i\in\N\right\}\subseteq \epgf$ be an $\omega$-chain with $F_0 \preceq F_1 \preceq F_2 \preceq \dots$. Then:
		\begin{align*}
		\abs{\sup S} &= \sum_{\sigma\in\N^k} \coef{\sigma}{\sup S} \\
		&= \sum_{\sigma\in\N^k} \sup\left\{\coef{\sigma}{F_i} \mid i\in \N\right\} \\
		&= \sup \left\{\sum_{\sigma\in\N^k} \coef{\sigma}{F_i} ~\big|~ i\in\N\right\} \tag{Monotone Convergence Theorem}\\
		&= \sup \left\{\abs{F_i} \mid i\in\N\right\}
		\end{align*}
\end{proof}

\begin{theorem}[Continuity of $\sem{\cdot}$] \label{thm:contsem}
	For every cpGCL program $P$, $\sem{P}$ is continuous on the domain $\left(\epgf\rightarrow\epgf\right)$.
\end{theorem}

\begin{proof}
	Let $S\subseteq \epgf$ be an $\omega$-cpo. The proof proceeds by induction over the structure of $P$:
	
	\paragraph{Case $P = \pskip$:}
	\begin{align*}
	\sem{P}(\sup S) &= \sup S = \sup \left\{F \mid F\in S\right\} = \sup \left\{\sem{P}(F) \mid F\in S\right\}
	\end{align*}
	
	\paragraph{Case $P = \ASSIGN{\progvar{x}_i}{E}$:}
	\begin{align*}
	\sem{P}(\sup S) &= \sem{P}\left(\coef{\lightning}{\sup S} \cdot X_\lightning + \sum_{\sigma\in\N^k} \coef{\sigma}{\sup S} \cdot\mathbf{X}^\sigma\right) \\
	&= \coef{\lightning}{\sup S}\cdot X_\lightning + \sum_{\sigma\in\N^k} \coef{\sigma}{\sup S} \cdot X_1^{\sigma_1}\cdots X_i^{\eval{\sigma}{E}}\cdots X_k^{\sigma_k} \\
	&= \sup_{F \in S} \left\{\coef{\lightning}{F}\cdot X_\lightning + \sum_{\sigma\in\N^k} \coef{\sigma}{F} \cdot X_1^{\sigma_1}\cdots X_i^{\eval{\sigma}{E}}\cdots X_k^{\sigma_k} \right\} \\
	&= \sup_{F \in S} \left\{\sem{P}(F) \right\}
	\end{align*}
	
	\paragraph{Case $P = \observe{B}$:}
	\begin{align*}
	\sem{P}(\sup S) &= \left(\coef{\lightning}{\sup S} + \abs{\angles{\sup S}_{\neg B}}\right) \cdot X_\lightning + \angles{\sup S}_B \\
	&= \left(\coef{\lightning}{\sup S} + \sup \left\{\abs{\angles{F}_{\neg B}}\mid F\in S\right\}\right) \cdot X_\lightning \\
	&\qquad + \sup \left\{\angles{F}_B \mid F\in S\right\} \tag{Cont. of $\abs{\cdot}$, $\angles{\cdot}_B$} \\
	&= \sup \left\{\left(\coef{\lightning}{F} + \abs{\angles{F}_{\neg B}}\right)\cdot X_\lightning + \angles{F}_B\right\} \\
	&= \sup \left\{\sem{P}(F) \mid F\in S\right\}
	\end{align*}
	
	\paragraph{Case $P = \PCHOICE{P_1}{p}{P_2}$:}
	\begin{align*}
	\sem{P}(\sup S) &= p \cdot \sem{P_1}(\sup S) + (1-p) \cdot \sem{P_2}(\sup S) \\
	&= p \cdot \sup \left\{\sem{P_1}(F) \mid F\in S\right\} \\
	&\qquad + (1-p) \cdot \sup \left\{\sem{P_2}(F) \mid F\in S\right\} \tag{I.H. on $P_1$ and $P_2$} \\
	&= \sup \left\{p \cdot \sem{P_1}(F) + (1-p) \cdot \sem{P_2}(F) \mid F\in S\right\} \\
	&= \sup \left\{\sem{P}(F) \mid F\in S\right\}
	\end{align*}
	
	\paragraph{Case $P = \ITE{B}{P_1}{P_2}$:}
	\begin{align*}
	\sem{P}(\sup S) &= \coef{\lightning}{\sup S}\cdot X_\lightning + \sem{P_1}(\angles{\sup S}_P) + \sem{P_2}(\angles{\sup S}_{\neg B}) \\
	&= \coef{\lightning}{\sup S}\cdot X_\lightning + \sup \left\{\sem{P_1}(\angles{F}_P) \mid F\in S\right\}\\
	&\qquad + \sup \left\{\sem{P_2}(\angles{F}_{\neg B}) \mid F\in S\right\} \tag{I.H. on $P_1$ and $P_2$} \\
	&= \sup \left\{\coef{\lightning}{F}\cdot X_\lightning + \sem{P_1}(\angles{F}_P) + \sem{P_2}(\angles{F}_{\neg B} \mid F\in S\right\} \\
	&= \sup \left\{\sem{P}(F) \mid F\in S\right\}
	\end{align*}
	
	\paragraph{Case $P = P_1\fatsemi P_2$:}
	\begin{align*}
	\sem{P}(\sup S) &= \sem{P_2}\left(\sem{P_1}(\sup S)\right) \\
	&= \sem{P_2}\left(\sup \left\{\sem{P_1}(F)\mid F\in S\right\}\right) \tag{I.H. on $P_1$} \\
	&= \sup \left\{\sem{P_2}(\sem{P_1}(F)) \mid F\in S\right\} \tag{I.H. on $P_2$} \\
	&= \sup \left\{\sem{P}(F) \mid F\in S\right\}
	\end{align*}
	
	\paragraph{Case $P = \WHILEDO{B}{P_1}$:}
	In this case, we use that for all $n\in\N$, $\Phi_{B,P_1}^n(\bot)$ is continuous, which we prove by induction:
	\paragraph{Base case:} $n = 0$.
	\begin{align*}
	\Phi_{B,P_1}^0(\bot)(\sup S) &= 0 = \sup \left\{\Phi_{B,P_1}^0(\bot)(F) \mid F\in S\right\}
	\end{align*}
	
	\paragraph{Induction step:}
	\begin{align*}
	\Phi_{B,P_1}^{n+1}(\bot)(\sup S) &= \Phi_{B,P_1}\left(\Phi_{B,P_1}^n(\bot)\right)(\sup S) \\
	&= \coef{\lightning}{\sup S} + \angles{\sup S}_{\neg B} \\
	&\qquad + \Phi_{B,P_1}^n(\bot)(\sem{P_1}(\angles{\sup S}_B)) \tag{Def. $\Phi_{B,P_1}$} \\
	&= \coef{\lightning}{\sup S} + \sup \left\{\angles{F}_{\neg B} \mid F\in S\right\} \\
	&\qquad + \Phi_{B,P_1}^n(\bot)\left(\sup \left\{\sem{P_1}(\angles{F}_B) \mid F\in S\right\}\right) \tag{Cont. of $\angles{\cdot}_B$, outer I.H. on $P_1$} \\
	&= \coef{\lightning}{\sup S} + \sup \left\{\angles{F}_{\neg B} \mid F\in S\right\} \\
	&\quad + \sup \left\{\Phi_{B,P_1}^n(\bot)\left(\sem{P_1}(\angles{F}_B)\right) \mid F\in S\right\} \tag{Inner I.H.} \\
	&= \sup_{F \in S} \left\{\coef{\lightning}{F} + \angles{F}_{\neg B} + \Phi_{B,P_1}^n(\bot)\left(\sem{P_1}(\angles{F}_B)\right)\right\}\\
	&= \sup \left\{\Phi_{B,P_1}^{n+1}(\bot)(F) \mid F\in S\right\}
	\end{align*}
	
	With this, it follows:
	
	\begin{align*}
	\sem{P}(\sup S) &= \left(\sup_{n\in \N} \Phi_{B,P_1}^n(\bot) \right)(\sup S) \\
	&= \sup \left\{\Phi_{B,P_1}^n (\bot)(\sup S) \mid n\in\N\right\} \\
	&= \sup \left\{\sup \left\{\Phi_{B,P_1}^n(\bot)(F) \mid F\in S\right\}\mid n\in\N\right\} \tag{Cont. of $\Phi_{B,P_1}^n(\bot)$} \\
	&= \sup \left\{\sup \left\{\Phi_{B,P_1}^n(\bot)(F) \mid n\in\N\right\}\mid F\in S\right\} \tag{swap suprema} \\
	&= \sup \left\{\sup \left\{\Phi_{B,P_1}^n(\bot) \mid n\in\N\right\}(F) \mid F\in S\right\} \\
	&= \sup \left\{\sem{P}(F)\mid F\in S\right\}
	\end{align*}
\end{proof}

\begin{lemma}[Linearity of Auxiliary Functions]\label{lem:linaux}
	For all $\sigma\in\N^k \cup \{\lightning\}$, $a \in\R_{\geq 0}$, $F,G\in\efps$ and Boolean guards $B$, the following functions are linear:
	\begin{enumerate}
		\item The coefficient function $\coef{\sigma}{}$, i.e.\ $\coef{\sigma}{a F + G} = a\cdot\coef{\sigma}{F} + \coef{\sigma}{G}$. \label{lem:lincoeff}
		\item The restriction $\angles{\cdot}_B$, i.e.\ $\angles{a F + G}_B = a \cdot\angles{F}_B + \angles{G}_B$.\label{lem:linrestr}
		\item The mass $\abs{\cdot}$, i.e.\ $\abs{a F + G} = a\cdot\abs{F} + \abs{G}$.
	\end{enumerate}
\end{lemma}

\begin{proof}
	
	\begin{enumerate}
		\item follows from coefficient-wise addition and scalar multiplication on eFPS:
		\begin{align}
		a F + G &= a\cdot\left(\coef{\lightning}{F}\cdot X_\lightning + \sum_{\sigma\in\N^k}\coef{\sigma}{F}\cdot\mathbf{X}^\sigma\right) \\
		&\qquad + \left(\coef{\lightning}{G}\cdot X_\lightning + \sum_{\sigma\in\N^k}\coef{\sigma}{G}\cdot\mathbf{X}^\sigma\right) \nonumber\\
		&= \left(a\cdot\coef{\lightning}{F} + \coef{\lightning}{G}\right)\cdot X_\lightning + \sum_{\sigma\in\N^k}\left(a\cdot\coef{\sigma}{F} + \coef{\sigma}{G}\right)\cdot\mathbf{X}^\sigma \label{eq:coeff1}
		\intertext{By \cref{def:efps}:}
		a F+G &= \coef{\lightning}{a F + G}\cdot X_\lightning + \sum_{\sigma\in\N^k}\coef{\sigma}{a F+G}\cdot\mathbf{X}^\sigma \label{eq:coeff2}
		\end{align}
		
		Comparing coefficients of \cref{eq:coeff1} and \cref{eq:coeff2} yields $\coef{\sigma}{a F + G} = a\cdot \coef{\sigma}{F} + \coef{\sigma}{G}$ for all $\sigma\in\N^k \cup \{\lightning\}$.
		
		\item Using the result of \ref{lem:lincoeff}:
		\begin{align*}
		\angles{a F + G}_B &= \sum_{\sigma\models B} \coef{\sigma}{a F+G}\cdot\mathbf{X}^\sigma \\
		&= \sum_{\sigma\models B} \left(a\cdot\coef{\sigma}{F} + \coef{\sigma}{G}\right)\cdot\mathbf{X}^\sigma \\
		&= a\cdot\sum_{\sigma\models B}\coef{\sigma}{F}\cdot\mathbf{X}^\sigma + \sum_{\sigma\models B}\coef{\sigma}{G}\cdot\mathbf{X}^\sigma \\
		&= a\cdot\angles{F}_B + \angles{G}_B
		\end{align*}
		
		\item follows directly from the linearity of the coefficient function \ref{lem:lincoeff}:
		\begin{align*}
					\abs{a F + G} &= \sum_{\sigma \in \N^k}{\coef{\sigma}{a F + G}} \\
						&= \sum_{\sigma \in \N^k}\left(a\cdot\coef{\sigma}{F} + \coef{\sigma}{G}\right) \\
						&= a \cdot \sum_{\sigma \in \N^k}{\coef{\sigma}{F}} + \sum_{\sigma \in \N^k}{\coef{\sigma}{G}} \\
						&= a\cdot\abs{F} + \abs{G}
		\end{align*}
	\end{enumerate}
\end{proof}

\begin{lemma}[$\Phi_{B,P}$ Preserves Linearity]\label{lem:linpreserve}
	Let $\psi\colon\epgf\rightarrow\epgf$ be a linear function, i.e., for all $a \in \R, F, G \in \epgf$, $aF + G \in \epgf$ implies $\psi(aF+G) = a\cdot\psi(F)+\psi(G)$.
	If $\sem{P}$ is linear then $\Phi_{B,P}(\psi)$ is linear as well.
\end{lemma}

\begin{proof}
	\begin{align*}
	\;& \Phi_{B,P}(\psi)(a F + G) \\
	=\;& \left(\lambda F.~ \coef{\lightning}{F} X_\lightning + \angles{F}_{\neg B} + \psi(\sem{P}(\angles{F}_B))\right)(a F + G) \\
	=\;& \coef{\lightning}{a F + G} X_\lightning + \angles{a F + G}_{\neg B} + \psi(\sem{P}(\angles{a F+G}_B)) \\
	=\;& (a\coef{\lightning}{F} + \coef{\lightning}{G})\cdot X_\lightning + a\angles{F}_{\neg B} + \angles{G}_{\neg B} + \psi(\sem{P}(a\angles{F}_B + \angles{G}_B)) \tag{Lin. of $\angles{\cdot}_B$ (\cref{lem:linaux})} \\
	=\;& (a\coef{\lightning}{F} + \coef{\lightning}{G})\cdot X_\lightning + a\angles{F}_{\neg B} + \angles{G}_{\neg B} \\
	&\qquad + \psi(a\sem{P}(\angles{F}_B) + \sem{P}(\angles{G}_B)) \tag{Lin. of $\sem{P}$} \\
	=\;&  (a\coef{\lightning}{F} + \coef{\lightning}{G})\cdot X_\lightning + a\angles{F}_{\neg B} + \angles{G}_{\neg B} \\
	&\qquad+ a\cdot\psi(\sem{P}(\angles{F}_B))  + \psi(\sem{P}(\angles{G}_B)) \tag{Lin. of $\psi$} \\
	=\;& a \cdot \left(\coef{\lightning}{F} X_\lightning + \angles{F}_{\neg B} + \psi(\sem{P}(\angles{F}_B))\right) \\
	&\qquad + \left(\coef{\lightning}{G} X_\lightning + \angles{G}_{\neg B} + \psi(\sem{P}(\angles{G}_B))\right) \\
	=\;& a \cdot \Phi_{B,P}(\psi)(F) + \Phi_{B,P}(\psi)(G)
	\end{align*}
\end{proof}

\begin{corollary}\label{cor:linphin}
	If $\sem{P}$ is linear, then $\Phi_{B,P}^n(\bot)$ is linear for all $n\in\N$, i.e.
	\[\Phi_{B,P}^n(\bot)(a F+G) = a\cdot\Phi_{B,P}^n(\bot)(F) + \Phi_{B,P}^n(\bot)(G).\]
\end{corollary}

\begin{proof}
	By induction:
	
	\paragraph{Base case:} $n = 0$.
	$\Phi_{B,P}^0(\bot) = \bot$ is linear, as $\bot(a F+G) = 0 = a\cdot\bot(F) + \bot(G)$.
	
	\paragraph{Induction step:}
	By the induction hypothesis $\Phi_{B,P}^n(\bot)$ is a linear function. Therefore, $\Phi_{B,P}^{n+1}(\bot) = \Phi_{B,P}(\Phi_{B,P}^n(\bot))$ is also linear by \cref{lem:linpreserve}.
\end{proof}

\begin{theorem}[Linearity of $\sem{\cdot}$] \label{thm:linsem}
	The semantics transformer $\sem{\cdot}$ is linear, i.e.\ for any cpGCL program P
	\[\sem{P}(a F + G) = a \cdot \sem{P}(F) + \sem{P}(G).\]
\end{theorem}

\begin{proof}
	By induction over the structure of $P$:
	
	\paragraph{Case $P = \pskip$:}
	\begin{align*}
	\sem{P}(a F+G) &= a F + G \\
	&= a \sem{P}(F) + \sem{P}(G)
	\end{align*}
	
	\paragraph{Case $P = \ASSIGN{\progvar{x}_i}{E}$:}
	\begin{align*}
	\;& \sem{P}(a F+G)\\
	&= \sem{P}\big((a \coef{\lightning}{F} + \coef{\lightning}{G})X_\lightning \\
	&\qquad + \sum_{\sigma\in\N^k} (a\coef{\sigma}{F} + \coef{\sigma}{G})\mathbf{X}^\sigma\big) \\
	&= (a \coef{\lightning}{F} + \coef{\lightning}{G})X_\lightning \\
	&\qquad+ \sum_{\sigma\in\N^k} (a\coef{\sigma}{F} + \coef{\sigma}{G})X_1^{\sigma_1}\cdots X_i^{\eval{\sigma}{E}}\cdots X_k^{\sigma_k} \\
	&= a\cdot\left(\coef{\lightning}{F} X_\lightning + \sum_{\sigma\in\N^k} \coef{\sigma}{F} X_1^{\sigma_1}\cdots X_i^{\eval{\sigma}{E}}\cdots X_k^{\sigma_k} \right) \\
	&\qquad + \left(\coef{\lightning}{G} X_\lightning + \sum_{\sigma\in\N^k} \coef{\sigma}{G} X_1^{\sigma_1}\cdots X_i^{\eval{\sigma}{E}}\cdots X_k^{\sigma_k}\right) \\
	&= a\cdot\sem{P}(F) + \sem{P}(G)
	\end{align*}
	
	\paragraph{Case $P = \observe{B}$:}
	\begin{align*}
	\;& \sem{P}(a F+G)\\
	=\;& (a\coef{\lightning}{F} + \coef{\lightning}{G} + \abs{\angles{a F + G}_{\neg B}})X_\lightning + \angles{a F + G}_B \\
	=\;& (a\coef{\lightning}{F} + \coef{\lightning}{G} + a\abs{\angles{F}_{\neg B}} + \abs{\angles{G}_{\neg B}})X_\lightning \\
	&\qquad + a\angles{F}_B + \angles{G} \tag{Lin. of $\angles{\cdot}_B$ (\cref{lem:linaux})}\\
	=\;& a\left((\coef{\lightning}{F} + \abs{\angles{F}_{\neg B}})X_\lightning + \angles{F}_B\right) \\
	& + \left((\coef{\lightning}{G} + \abs{\angles{G}_{\neg B}})X_\lightning + \angles{G}_B\right) \\
	=\;& a \sem{P}(F) + \sem{P}(G)
	\end{align*}
	
	\paragraph{Case $P = \PCHOICE{P_1}{p}{P_2}$:}
	\begin{align*}
	\;& \sem{P}(a F+G)\\
	=\;& p\cdot \sem{P_1}(a F + G) + (1-p) \cdot \sem{P_2}(a F + G)\\
	=\;& p \cdot (a\sem{P_1}(F) + \sem{P_1}(G)) + (1-p) \cdot (a\sem{P_2}(F) + \sem{P_2}(G)) \tag{I.H.} \\
	=\;& a \left(p \cdot\sem{P_1}(F) + (1-p) \cdot \sem{P_2}(F)\right) \\
	&\qquad + \left(p \cdot\sem{P_1}(G) + (1-p) \cdot \sem{P_2}(G)\right) \\
	=\;& a \sem{P}(F) + \sem{P}(G)
	\end{align*}
	
	\paragraph{Case $P = \ITE{B}{P_1}{P_2}$:}
	\begin{align*}
	\;& \sem{P}(a F+G)\\
	=\;& (a\coef{\lightning}{F} + \coef{\lightning}{G}) X_\lightning + \sem{P_1}(\angles{a F + G}_B) + \sem{P_2}(\angles{a F + G}_{\neg B})\\
	=\;& (a\coef{\lightning}{F} + \coef{\lightning}{G}) X_\lightning + \sem{P_1}(a\angles{F}_B + \angles{G}_B) \\
	&\qquad + \sem{P_2}(a\angles{F}_{\neg B} + \angles{G}_{\neg B}) \tag{Lin. of $\angles{\cdot}_B$ (\cref{lem:linaux})}\\
	=\;&  (a\coef{\lightning}{F} + \coef{\lightning}{G}) X_\lightning + a\sem{P_1}(\angles{F}_B) + \sem{P_1}(\angles{G}_B) \\
	& \qquad + a\sem{P_2}(\angles{F}_{\neg B}) + \sem{P_2}(\angles{G}_{\neg B}) \tag{I.H.} \\
	=\;& a\left(\coef{\lightning}{F} X_\lightning + \sem{P_1}(\angles{F}_B) + \sem{P_2}(\angles{F}_{\neg B}\right)  \\
	&\qquad + \left(\coef{\lightning}{G} X_\lightning + \sem{P_1}(\angles{G}_B) + \sem{P_2}(\angles{G}_{\neg B}\right) \\
	=\;& a \sem{P}(F) + \sem{P}(G)
	\end{align*}
	
	\paragraph{Case $P = P_1\fatsemi P_2$:}
	\begin{align*}
	\;& \sem{P}(a F+G)\\
	=\;& \sem{P_2}(\sem{P_1}(a F + G)) \\
	=\;& \sem{P_2}(a \sem{P_1}(F) + \sem{P_1}(G)) \tag{I.H.}\\
	=\;& a \sem{P_2}(\sem{P_1}(F)) + \sem{P_2}(\sem{P_1}(G)) \tag{I.H.}\\
	=\;& a \sem{P}(F) + \sem{P}(G)
	\end{align*}
	
	\paragraph{Case $P = \WHILEDO{B}{P_1}$:}
	\begin{align*}
	\;& \sem{P}(a F+G)\\
	=\;& (\lfp\; \Phi_{B,P_1})(a F + G) \\
	=\;& \left(\sup \left\{\Phi_{B,P_1}^n(\bot)\mid n\in\N\right\}\right)(a F + G) \\
	=\;& \sup \left\{\Phi_{B,P_1}^n(\bot)(a F + G) \mid n\in\N\right\}\\
	=\;& \sup \left\{a\cdot\Phi_{B,P_1}^n(\bot)(F) + \Phi_{B,P_1}^n(\bot)(G) \mid n\in\N\right\} \tag{\cref{cor:linphin}, $\sem{P_1}$ lin. by I.H.} \\
	=\;& a \cdot \sup \left\{\Phi_{B,P_1}^n(\bot)(F) \mid n\in\N\right\} + \sup \left\{\Phi_{B,P_1}^n(\bot)(G) \mid n\in\N\right\}\\
	=\;& a \cdot \left(\sup \left\{\Phi_{B,P_1}^n(\bot)\mid n\in\N\right\}\right)(F)  + \left(\sup \left\{\Phi_{B,P_1}^n(\bot) \mid n\in\N\right\}\right)(G) \\
	=\;& a \cdot (\lfp\; \Phi_{B,P_1})(F) + (\lfp\; \Phi_{B,P_1})(G) \\
	=\;& a \sem{P}(F) + \sem{P}(G)
	\end{align*}
\end{proof}

\begin{lemma}[Error Term Pass-Through]
	For every program $P$ and every $F\in\epgf$, the error term $\coef{\lightning}{F} X_\lightning$ passes through the transformer unaffected, i.e.
	\[\sem{P}(F) = \sem{P}\left(\sum_{\sigma\in\N^k} \coef{\sigma}{F} \mathbf{X}^\sigma\right) + \coef{\lightning}{F} X_\lightning.\]
\end{lemma}

\begin{proof} 
	By linearity of $\sem{P}$, we get:
	\begin{align*}
	\sem{P}(F) &= \sem{P}\left(\sum_{\sigma\in\N^k} \coef{\sigma}{F} \mathbf{X}^\sigma + \coef{\lightning}{F} X_\lightning\right)\\
	&= \sem{P}\left(\sum_{\sigma\in\N^k} \coef{\sigma}{F} \mathbf{X}^\sigma\right) + \coef{\lightning}{F} \cdot \sem{P}(X_\lightning)
	\end{align*}
	
	It therefore remains to be shown that $\sem{P}(X_\lightning) = X_\lightning$ by induction over the structure of $P$:
	
	\paragraph{Case $P = \pskip$:}
	\begin{align*}
	\sem{P}(X_\lightning) &= X_\lightning
	\end{align*}
	
	\paragraph{Case $P = \ASSIGN{\progvar{x}_i}{E}$:}
	\begin{align*}
	\sem{P}(X_\lightning) &= X_\lightning + \sum_{\sigma\in\N^k} 0\cdot X_1^{\sigma_1}\cdots X_i^{\eval{\sigma}{E}}\cdots X_k^{\sigma_k} \\
	&= X_\lightning
	\end{align*}
	
	\paragraph{Case $P = \observe{B}$:}
	\begin{align*}
	\sem{P}(X_\lightning) &= \left(1 + \abs{\angles{X_\lightning}_{\neg B}}\right)X_\lightning + \angles{X_\lightning}_B \\
	&= X_\lightning
	\end{align*}
	
	\paragraph{Case $P = \PCHOICE{P_1}{p}{P_2}$:}
	\begin{align*}
	\sem{P}(X_\lightning) &= p\cdot \sem{P_1}(X_\lightning) + (1-p)\cdot \sem{P_2}(X_\lightning) \\
	&= p\cdot X_\lightning + (1-p)\cdot X_\lightning \tag{I.H. on $P_1$ and $P_2$}\\
	&= X_\lightning
	\end{align*}
	
	\paragraph{Case $P = \ITE{B}{P_1}{P_2}$:}
	\begin{align*}
	\sem{P}(F) &= X_\lightning + \sem{P_1}(\angles{X_\lightning}_B) + \sem{P_2}(\angles{X_\lightning}_{\neg B}) \\
	&= X_\lightning + \sem{P_1}(0) + \sem{P_2}(0) \\
	&= X_\lightning
	\end{align*}
	
	\paragraph{Case $P = P_1\fatsemi P_2$:}
	\begin{align*}
	\sem{P}(X_\lightning) &= \sem{P_2}(\sem{P_1}(X_\lightning)) \\
	&= \sem{P_2}(X_\lightning) \tag{I.H. on $P_1$}\\
	&= X_\lightning \tag{I.H. on $P_2$}
	\end{align*}
	
	\paragraph{Case $P = \WHILEDO{B}{P_1}$:}\ 
	
	\noindent We show that $\forall n\in\N: \Phi_{B,P_1}^{n+1}(\bot)(X_\lightning) = X_\lightning$:
	
	\begin{align*}
	\Phi_{B,P_1}^{n+1}(\bot)(X_\lightning) &= \Phi_{B,P_1}(\Phi_{B,P_1}^{n}(\bot))(X_\lightning) \\
	&= X_\lightning + \angles{X_\lightning}_{\neg B} + \Phi_{B,P_1}^n(\bot)(\sem{P_1}(\angles{X_\lightning}_B)) \tag{def. $\Phi_{B,P_1}$}\\
	&= X_\lightning + \Phi_{B,P_1}^n(\bot)(0) \\
	&= X_\lightning \tag{$\Phi_{B,P_1}^n(\bot)(0) \leq \sem{\WHILEDO{B}{P_1}}(0) = 0$}
	\end{align*}
	
	From this, it follows:
	
	\begin{align*}
	\sem{P}(X_\lightning) &= \sup\,\left\{\Phi_{B,P_1}^n(\bot)(X_\lightning) \mid n \in \N\right\} \\
	&= \sup\,\left\{0, X_\lightning\right\} \tag{$\forall n\in\N: \Phi_{B,P_1}^{n+1}(\bot)(X_\lightning) = X_\lightning$} \\
	&= X_\lightning 
	\end{align*}
	
\end{proof}

\begin{lemma}[Alternative Representation]
	\label{apx:lem:altrepr}
	\begin{align*}
	\sem{\WHILEDO{B}{P}}(G) &= \sum_{i=0}^\infty \left(\coef{\lightning}{\varphi_{B,P}^i(G)}X_\lightning + \angles{\varphi_{B,P}^i(G)}_{\neg B}\right) \\
	& \quad \text{where} \quad \varphi_{B,P}(G) \coloneqq \sem{P}(\angles{G}_B).
	\end{align*}
\end{lemma}

\begin{proof}
First, we show by induction that for all $n\in\N$: \[\Phi_{B,P}^n(\bot)(G) = \sum_{i=0}^{n-1} \left(\coef{\lightning}{\varphi_{B,P}^i(F)}X_\lightning + \angles{\varphi_{B,P}^i(F)}_{\neg B}\right).\]

\paragraph{Base case:} $n = 0$.
\[\Phi_{B,P}^0(\bot)(G) = 0 = \sum_{i=0}^{-1} \left(\coef{\lightning}{\varphi_{B,P}^i(G)}X_\lightning + \angles{\varphi_{B,P}^i(G)}_{\neg B}\right)\]

\paragraph{Induction step:}
\begin{align*}
\Phi_{B,P}^{n+1}(\bot)(G) &= \Phi_{B,P}(\Phi_{B,P}^{n}(\bot))(G) \\
&= \coef{\lightning}{G} X_\lightning + \angles{G}_{\neg B} + \Phi_{B,P}^n(\bot)(\sem{P}(\angles{G}_B)) \tag{Def. $\Phi_{B,P}$}\\
&= \coef{\lightning}{G} X_\lightning + \angles{G}_{\neg B} + \Phi_{B,P}^n(\bot)(\varphi_{B,P}(G)) \tag{Def. $\varphi_{B,P}$} \\
&= \coef{\lightning}{G} X_\lightning + \angles{G}_{\neg B} \\
&\quad + \sum_{i=0}^{n-1} \left(\coef{\lightning}{\varphi_{B,P}^i(\varphi_{B,P}(G))}X_\lightning + \angles{\varphi_{B,P}^i(\varphi_{B,P}(G))}_{\neg B}\right) \tag{I.H.} \\
&= \coef{\lightning}{\varphi_{B,P}^0(G)} X_\lightning + \angles{\varphi_{B,P}^0(G)}_{\neg B} \\
&\quad+ \sum_{i=1}^{n} \tag{$\varphi_{B,P}^0(G)=G$, index shift} \coef{\lightning}{\varphi_{B,P}^i(G)}X_\lightning + \angles{\varphi_{B,P}^i(G)}_{\neg B} \\
&=\sum_{i=0}^{n} \coef{\lightning}{\varphi_{B,P}^i(G)}X_\lightning + \angles{\varphi_{B,P}^i(G)}_{\neg B}
\end{align*}

From this, it follows:
\begin{align*}
\sem{\WHILEDO{B}{P}}(G) &= \left(\sup_{n \in \N}\Phi_{B,P}^n(\bot)\right)(G) \\
&= \sup\left\{\Phi_{B,P}^n(\bot)(G)\mid n\in\N\right\} \\
&= \sup_{n \in \N}\left\{\sum_{i=0}^{n-1} \coef{\lightning}{\varphi_{B,P}^i(G)}X_\lightning + \angles{\varphi_{B,P}^i(G)}_{\neg B}\right\} \\
&= \sum_{i=0}^{\infty} \coef{\lightning}{\varphi_{B,P}^i(G)}X_\lightning + \angles{\varphi_{B,P}^i(G)}_{\neg B}
\end{align*}
\end{proof}


\subsection{Coincidence to Operational Semantics}
We refer to the operational semantics for \cpgcl programs described in \cite{DBLP:journals/toplas/OlmedoGJKKM18}.
We show that the Markov chain $\mathcal{R}_\sigma\sem{P}$ precisely reflects the non-normalized PGF semantics $\sem{P}(\bvec{X}^\sigma)$ for any $p \in \cpgcl$ with initial state valuation $\sigma \in \N^k$.
\Cref{apx:operational} shows that the probabilities $\mcterm{P}$ for all $\sigma'$ and $\mcerr{P}$ arising from the Markov chain correspond to the coefficients of $\sem{P}(\bvec{X}^\sigma)$.
It is further shown that modifying these probabilities to the conditional probabilities $\text{Pr}^{\mathcal{R}_\sigma\sem{P}}(\diamondsuit \langle \downarrow, \sigma' \rangle \mid \neg \diamondsuit \lightning)$ has the same effect as applying the normalization function $\normalize$, thus concluding that the two semantics coincide (cf.~\Cref{apx:operational_equivalence}).

\begin{definition}[Markov Chain Semantics of \cpgcl]
	For any \cpgcl program $P$ and any starting state valuation $\sigma\in\N^k$, the \emph{operational Markov chain} is
	\[\mathcal{R}_\sigma\sem{P} \ddefeq \!\left(\mathcal{S}, \angles{P,\sigma}, \mathcal{P}\right),\]
	where:
	\begin{itemize}
		\item $\angles{P,\sigma}$ is the starting state
		\item the set of states $\mathcal{S}$ is the smallest set such that:
		\begin{itemize}
			\item $\mathcal{S}$ contains the starting state $\angles{P, \sigma}$
			\item if $s\in \mathcal{S}$ and $s$ has an outgoing transition to $s'$ according to \cref{fig:mcrules}, then $s'\in \mathcal{S}$
		\end{itemize}
		\item $\mathcal{P}\colon\mathcal{S}\times\mathcal{S} \rightarrow [0,1]$ is the transition matrix with
		\begin{itemize}
			\item $\mathcal{P}(s, s') = p$ if $s \overset{p}{\rightarrow} s'$ can be derived according to \cref{fig:mcrules}
			\item $\mathcal{P}(s, s') = 0$ otherwise
		\end{itemize}
	\end{itemize}
\end{definition}

\begin{figure}[h]
	\label{apx:mcconstruct}
		\raggedright
		\scriptsize
		
		\AxiomC{\vphantom{$\sigma\models B$}}
		\LeftLabel{(skip)}
		\UnaryInfC{$\angles{\pskip,\sigma} \longrightarrow \angles{\downarrow,\sigma}$}
		\DisplayProof
		\qquad
		\AxiomC{\vphantom{$\sigma\models B$}}
		\LeftLabel{(asgn)}
		\UnaryInfC{$\angles{\ASSIGN{\progvar{x}}{E},\sigma} \longrightarrow \angles{\downarrow, \sigma[\progvar{x}\leftarrow\eval{\sigma}{E}]}$}
		\DisplayProof
		\\ \vspace{2mm} 
		\AxiomC{$\sigma\models B$}
		\LeftLabel{(obs-t)}
		\UnaryInfC{$\angles{\observe{B},\sigma} \longrightarrow \angles{\downarrow, \sigma}$}
		\DisplayProof
		\qquad
		\AxiomC{$\sigma\not\models B$}
		\LeftLabel{(obs-f)}
		\UnaryInfC{$\angles{\observe{B},\sigma} \longrightarrow \angles{\lightning}$}
		\DisplayProof
		\\ \vspace{2mm} 
		\AxiomC{\vphantom{$\sigma\models B$}}
		\LeftLabel{(seq-1)}
		\UnaryInfC{$\angles{\downarrow\fatsemi Q, \sigma} \longrightarrow \angles{Q,\sigma}$}
		\DisplayProof
		\qquad
		\AxiomC{$\angles{P,\sigma} \longrightarrow \angles{\lightning}$}
		\LeftLabel{(seq-2)}
		\UnaryInfC{$\angles{P\fatsemi Q, \sigma} \longrightarrow \angles{\lightning}$}
		\DisplayProof
		\\ \vspace{4mm}
		\AxiomC{$\angles{P,\sigma} \overset{p}{\longrightarrow} \angles{P',\sigma'}$}
		\LeftLabel{(seq-3)}
		\UnaryInfC{$\angles{P\fatsemi Q, \sigma} \overset{p}{\longrightarrow} \angles{P'\fatsemi Q, \sigma'}$}
		\DisplayProof
		\\ \vspace{2mm} 
		\AxiomC{\vphantom{$\sigma\models B$}}
		\LeftLabel{(choice-l)}
		\UnaryInfC{$\angles{\pchoice{P}{p}{Q}, \sigma} \overset{p}{\longrightarrow} \angles{P, \sigma}$}
		\DisplayProof
		\qquad
		\AxiomC{\vphantom{$\sigma\models B$}}
		\LeftLabel{(choice-r)}
		\UnaryInfC{$\angles{\pchoice{P}{p}{Q}, \sigma} \overset{1-p}{\longrightarrow} \angles{Q, \sigma}$}
		\DisplayProof
		\\ \vspace{2mm} 
		\AxiomC{$\sigma\models B$}
		\LeftLabel{(if-t)}
		\UnaryInfC{$\angles{\ITE{B}{P}{Q}, \sigma} \longrightarrow \angles{P,\sigma}$}
		\DisplayProof
		\qquad
		\AxiomC{$\sigma\not\models B$}
		\LeftLabel{(if-f)}
		\UnaryInfC{$\angles{\ITE{B}{P}{Q}, \sigma} \longrightarrow \angles{Q,\sigma}$}
		\DisplayProof
		\\ \vspace{2mm} 
		\AxiomC{$\sigma\models B$}
		\LeftLabel{(while-t)}
		\UnaryInfC{$\angles{\WHILEDO{B}{P}, \sigma} \longrightarrow \angles{P \fatsemi \WHILEDO{B}{P}, \sigma} $}
		\DisplayProof
		\qquad
		\AxiomC{$\sigma\not\models B$}
		\LeftLabel{(while-f)}
		\UnaryInfC{$\angles{\WHILEDO{B}{P}, \sigma} \longrightarrow \angles{\downarrow,\sigma}$}
		\DisplayProof
		\\ \vspace{2mm} 
		\AxiomC{\vphantom{$\sigma\models B$}}
		\LeftLabel{(terminal)}
		\UnaryInfC{$\angles{\downarrow, \sigma} \longrightarrow \angles{\sink}$}
		\DisplayProof
		\qquad
		\AxiomC{\vphantom{$\sigma\models B$}}
		\LeftLabel{(undesired)}
		\UnaryInfC{$\angles{\lightning} \longrightarrow \angles{\sink}$}
		\DisplayProof
		\qquad
		\AxiomC{\vphantom{$\sigma\models B$}}
		\LeftLabel{(sink)}
		\UnaryInfC{$\angles{\sink} \longrightarrow \angles{\sink}$}
		\DisplayProof
		
		\caption{Construction rules for the operational Markov chain. $\sigma[\progvar x \leftarrow \eval{\sigma}{E}]$ denotes the program state valuation $\sigma$ with the value of $\progvar x$ replaced by $\eval{\sigma}{E}$. Whenever a transition has no annotated weight above its arrow, it has a weight of $1$.}
		\label{fig:mcrules}
\end{figure}

\begin{lemma}
	\label{apx:operational}
	For every $P \in \cpgcl$ and every two $\sigma, \sigma' \in \N^k$
	\begin{enumerate}
		\item $\textup{Pr}^{\mathcal{R}_\sigma \sem{P}}(\diamondsuit\langle\downarrow, \sigma'\rangle) ~ = ~ \coef{\sigma'}{\sem{P}(\bvec{X}^\sigma)}$
		\item $\textup{Pr}^{\mathcal{R}_\sigma \sem{P}}(\diamondsuit\lightning) ~ = ~ \coef{\lightning}{\sem{P}(\bvec{X}^\sigma)}$
	\end{enumerate}
\end{lemma}
\begin{proof}
	We prove the statements (1) and (2) simultaneously by structural induction over a \cpgcl program $P$.
	
\noindent\textbf{Case} $P = \pskip$: In this case, the Markov chain $\mathcal{R}_\sigma \sem{P}$ looks as follows:
	\begin{figure}[h]
		\begin{tikzpicture}[->,>=stealth, shorten >=2pt, line width=0.5pt, node distance=2cm, initial text=$ $]
		\node[initial] (1) {$\langle \pskip, \sigma \rangle$};
		\node[right of=1] (2) {$\langle \downarrow, \sigma \rangle$};
		\node[right of=2] (3) {$\langle \mathit{sink} \rangle$};
		
		\draw (1) edge (2)
				  (2) edge (3)
				  (3) edge[loop right] (3);
		\end{tikzpicture}
	\end{figure}
	Its PGF semantics yields:
	\begin{align*}
		\sem{P}(\bvec{X}^\sigma) ~=~ \bvec{X}^\sigma
	\end{align*}
	
	Thus:
	\begin{align*}
		\mcterm{P} &=
		\begin{cases}
			1,& \text{if}~ \sigma' = \sigma\\
			0,& \text{else}
		\end{cases}
		~=~ \coef{\sigma'}{\sem{P}(\bvec{X}^\sigma)}\\
		\mcerr{P} &= 0 ~=~ \coef{\lightning}{\sem{P}(\bvec{X}^\sigma)}
	\end{align*}\\
	
\noindent\textbf{Case} $P = \progvar{x}_i \coloneq E$:
	\begin{figure}[h]
		\begin{tikzpicture}[->,>=stealth, shorten >=2pt, line width=0.5pt, node distance=3cm, initial text=$ $]
		\node[initial] (1) {$\langle \assign{\progvar{x}_i}{E}, \sigma \rangle$};
		\node[right of=1] (2) {$\langle \downarrow, \sigma[\progvar{x}_i \leftarrow \eval{\sigma}{}] \rangle$};
		\node[right of=2] (3) {$\langle \mathit{sink} \rangle$};
		
		\draw (1) edge (2)
		(2) edge (3)
		(3) edge[loop right] (3);
		\end{tikzpicture}
	\end{figure}

	Its PGF semantics yields:
	\begin{align*}
		\sem{P}(\bvec{X}^\sigma) &= X_1^{\sigma_1}\cdots X_i^{\eval{\sigma}{}}\cdots X_k^{\sigma_k}
	\end{align*}
	
	Thus:
	\begin{align*}
		\mcterm{P} &=
		\begin{cases}
		1,& \text{if}~ \sigma' = \sigma[\progvar{x}_i \leftarrow \eval{\sigma}{}]\\
		0,& \text{else}
		\end{cases}\\
		&=~ \coef{\sigma'}{\sem{P}(\bvec{X}^\sigma)}\\
		&\\
		\mcerr{P} &= 0 ~=~ \coef{\lightning}{\sem{P}(\bvec{X}^\sigma)}
	\end{align*}\\

\noindent\textbf{Case} $P = \observe{B}$:
	We do a case distinction whether $\sigma \models B$.
	
	\emph{Observe passed}:
	\begin{figure}[h]
		\begin{tikzpicture}[->,>=stealth, shorten >=2pt, line width=0.5pt, node distance=2.5cm, initial text=$ $]
		\node[initial] (1) {$\langle \observe{B}, \sigma \rangle$};
		\node[right of=1] (2) {$\langle \downarrow, \sigma \rangle$};
		\node[right of=2] (3) {$\langle \mathit{sink} \rangle$};
		
		\draw (1) edge (2)
		(2) edge (3)
		(3) edge[loop right] (3);
		\end{tikzpicture}
	\end{figure}

	The PGF semantics yields:
	\begin{align*}
		\sem{P}(\bvec{X}^\sigma) &= \constrain{\bvec{X}^\sigma}{B} + \left(\abs{\constrain{\bvec{X}^\sigma}{\neg B} } + \coef{\lightning}{\bvec{X}^\sigma}\right)X_\lightning = \bvec{X}^\sigma
	\end{align*}
	
	Thus:
	\begin{align*}
		\mcterm{P} &=
		\begin{cases}
		1,& \text{if}~ \sigma' = \sigma\\
		0,& \text{else}
		\end{cases}
		=~ \coef{\sigma'}{\sem{P}(\bvec{X}^\sigma)}\\
		\mcerr{P} &= 0 ~=~ \coef{\lightning}{\sem{P}(\bvec{X}^\sigma)}
	\end{align*}

	\emph{Observe failed}:
	\begin{figure}[h]
		\begin{tikzpicture}[->,>=stealth, shorten >=2pt, line width=0.5pt, node distance=2.5cm, initial text=$ $]
		\node[initial] (1) {$\langle \observe{B}, \sigma \rangle$};
		\node[right of=1] (2) {$\langle \lightning \rangle$};
		\node[right of=2] (3) {$\langle \mathit{sink} \rangle$};
		
		\draw (1) edge (2)
		(2) edge (3)
		(3) edge[loop right] (3);
		\end{tikzpicture}
	\end{figure}

	The PGF semantics yields:
	\begin{align*}
		\sem{P}(\bvec{X}^\sigma) &= \constrain{\bvec{X}^\sigma}{B} + \left(\abs{\constrain{\bvec{X}^\sigma}{\neg B} } + \coef{\lightning}{\bvec{X}^\sigma}\right)X_\lightning = X_\lightning
	\end{align*}
	
	Thus:
	\begin{align*}
		\mcterm{P} &= 0 ~=~ \coef{\sigma'}{\sem{P}(\bvec{X}^\sigma)}\\
		\mcerr{P} &= 1 ~=~ \coef{\lightning}{\sem{P}(\bvec{X}^\sigma)}
	\end{align*}\\
	
\noindent\textbf{Case} $P = \pchoice{P_1}{p}{P_2}$:
	
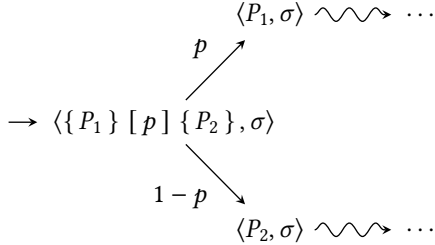
\begin{figure}[h]
		\begin{tikzpicture}[->,>=stealth, shorten >=2pt, line width=0.5pt, node distance=2cm, initial text=$ $]
		\node[initial] (1) {$\langle \pchoice{P_1}{p}{P_2}, \sigma \rangle$};
		\node[above right of=1] (2) {$\langle P_1,\sigma \rangle$};
		\node[below right of=1] (3) {$\langle P_2, \sigma \rangle$};
		\node[right of=2] (4) {$\ldots$};
		\node[right of=3] (5) {$\ldots$};
		
		\draw (1) edge node[above left] {$p$} (2)
		(1) edge node[below left] {$1-p$} (3)
		(2) edge[decorate, decoration={snake, post length=1mm}] (4)
		(3) edge[decorate, decoration={snake, post length=1mm}] (5);
		\end{tikzpicture}
	\end{figure}
	
	The PGF semantics yields:
	\begin{align*}
	\sem{P}(\bvec{X}^\sigma) &= p \cdot \sem{P_1}(\bvec{X}^\sigma) + (1-p) \cdot \sem{P_2}(\bvec{X}^\sigma)
	\end{align*}
	
	Thus:
	\begin{align*}
	\mcterm{P} &= p \cdot \mcterm{P_1}  \\
	& \qquad + (1-p) \cdot \mcterm{P_2} \\
	&= p \cdot \coef{\sigma'}{\sem{P_1}(\bvec{X}^\sigma)} + (1-p) \cdot \coef{\sigma'}{\sem{P_2}(\bvec{X}^\sigma)}\tag{by I.H.}\\
	&= \coef{\sigma'}{\sem{P}(\bvec{X}^\sigma)}\\
	&\\
	\mcerr{P} &= p \cdot \mcerr{P_1} + (1-p) \cdot \mcerr{P_2}\\
	&= p\cdot \coef{\lightning}{\sem{P_1}(\bvec{X}^\sigma)} + (1-p) \cdot \coef{\lightning}{\sem{P_2}(\bvec{X}^\sigma)}\tag{by I.H.}\\
	&= \coef{\lightning}{\sem{P}(\bvec{X}^\sigma)}
	\end{align*}\\
	
\noindent\textbf{Case} $P = \ITE{B}{P_1}{P_2}$:
	We do a case distinction on $\sigma \models B$.\\
	
	\emph{Condition is satisfied}:
	\begin{figure}[h]
		\begin{tikzpicture}[->,>=stealth, shorten >=2pt, line width=0.5pt, node distance=3.5cm, initial text=$ $]
		\node[initial] (1) {$\langle \ITE{B}{P_1}{P_2}, \sigma \rangle$};
		\node[right of=1] (2) {$\langle P_1,\sigma \rangle$};
		\node[right of=2, xshift=-1cm] (3) {$\ldots$};
		
		\draw (1) edge (2)
		(2) edge[decorate, decoration={snake, post length=1mm}] (3);
		\end{tikzpicture}
	\end{figure}

	The PGF semantics yields:
	\begin{align*}
	\sem{P}(\bvec{X}^\sigma) &= \sem{P_1}(\constrain{\bvec{X}^\sigma}{B}) + \sem{P_2}(\constrain{\bvec{X}^\sigma}{\neg B}) + \coef{\lightning}{\bvec{X}^\sigma}X_\lightning ~=~ \sem{P_1}(\bvec{X}^\sigma)
	\end{align*}
	
	Thus:
	\begin{align*}
	\mcterm{P} &= \mcterm{P_1}\\
	&= \coef{\sigma'}{\sem{P_1}(\bvec{X}^\sigma)}\tag{by I.H.}\\
	&\\
	\mcerr{P} &= \coef{\lightning}{\sem{P_1}(\bvec{X}^\sigma)}\\
	&=\coef{\lightning}{\sem{P}(\bvec{X}^\sigma)}
	\end{align*}

	\emph{Condition not satisfied}:
	\begin{figure}[h]
		\begin{tikzpicture}[->,>=stealth, shorten >=2pt, line width=0.5pt, node distance=3.5cm, initial text=$ $]
		\node[initial] (1) {$\langle \ITE{B}{P_1}{P_2}, \sigma \rangle$};
		\node[right of=1] (2) {$\langle P_2,\sigma \rangle$};
		\node[right of=2, xshift=-1cm] (3) {$\ldots$};
		
		\draw (1) edge (2)
		(2) edge[decorate, decoration={snake, post length=1mm}] (3);
		\end{tikzpicture}
	\end{figure}

	The PGF semantics yields:
	\begin{align*}
		\sem{P}(\bvec{X}^\sigma) &= \sem{P_1}(\constrain{\bvec{X}^\sigma}{B}) + \sem{P_2}(\constrain{\bvec{X}^\sigma}{\neg B}) + \coef{\lightning}{\bvec{X}^\sigma}X_\lightning ~=~ \sem{P_2}(\bvec{X}^\sigma)
	\end{align*}
	
	Thus:
	\begin{align*}
		\mcterm{P} &= \mcterm{P_2}\\
		&= \coef{\sigma'}{\sem{P_2}(\bvec{X}^\sigma)}\tag{by I.H.}\\
		&\\
		\mcerr{P} &= \coef{\lightning}{\sem{P_2}(\bvec{X}^\sigma)}\\
		&=\coef{\lightning}{\sem{P}(\bvec{X}^\sigma)}
	\end{align*}\\
	
\noindent\textbf{Case} $P = \COMPOSE{P_1}{P_2}$:
	\begin{figure}[h]
		\begin{tikzpicture}[->,>=stealth, shorten >=2pt, line width=0.5pt, node distance=2.3cm, initial text=$ $]
		\node[initial] (1) {$\langle \COMPOSE{P_1}{P_2}, \sigma \rangle$};
		\node[above right of=1] (2) {$\langle \lightning \rangle$};
		\node[right of=1] (3) {$\langle \COMPOSE{\downarrow}{P_2}, \sigma''\rangle$};
		\node[below right of=1] (4) {$\ldots$};
		\node[right of=3] (5) {$\langle P_2, \sigma'' \rangle$};
		\node[right of=5] (6) {$\ldots$};
		
		\draw (1) edge[decorate, decoration={snake, post length=1mm}] (2)
				  (1) edge[decorate, decoration={snake, post length=1mm}] (3)
				  (1) edge[decorate, decoration={snake, post length=1mm}] (4)
				  (3) edge (5)
				  (5) edge[decorate, decoration={snake, post length=2mm}] (6);
		\end{tikzpicture}
	\end{figure}
	
\noindent\textbf{Case} $P = \WHILEDO{B}{P_1}$:\\

	\emph{Condition not fulfilled} $(\sigma \not\models B)$:
	\begin{figure}[h]
		\begin{tikzpicture}[->,>=stealth, shorten >=2pt, line width=0.5pt, node distance=2.2cm, initial text=$ $]
		\node[initial] (1) {$\langle \WHILEDO{B}{P_1}, \sigma \rangle$};
		\node[right of=1, xshift=.5cm] (2) {$\langle \downarrow, \sigma\rangle$};
		\node[right of=2] (3) {$\langle \mathit{sink} \rangle$};
		
		\draw (1) edge (2)
		(2) edge (3)
		(3) edge[loop right] (3);
		\end{tikzpicture}
	\end{figure}
	
	For the PGF semantics, consider the following, for all $n \in \N$:
	\begin{align*}
	\Phi_{B,P_1}^n(\bot)(\bvec{X}^\sigma) &= \coef{\lightning}{\bvec{X}^\sigma} + \constrain{\bvec{X}^\sigma}{\neg B} + \Phi_{B,P_1}^n(\bot)(\sem{P_1}(\constrain{\bvec{X}^\sigma}{B}))\\
	&= \coef{\lightning}{\bvec{X}^\sigma} + \Phi_{B,P_1}^n(\bot)(\sem{P_1}(\constrain{\bvec{X}^\sigma}{B})) \tag{$\sigma \not\models B$}\\
	&= 0 + \bvec{X}^\sigma + 0\\
	\sem{P}(\bvec{X}^\sigma) &= \lfp\; \Phi_{B,P_1}(\bvec{X}^\sigma)\\
	&= \sup_{n \in \N}\left\{\Phi_{B,P_1}^n(\bot)(\bvec{X}^\sigma) \mid n \in \N\right\}\\
	&= \bvec{X}^\sigma
	\end{align*}
	
	Thus:
	\begin{align*}
		\mcterm{P} &= 
		\begin{cases}
			1,& \text{if}~ \sigma' = \sigma\\
			0,& \text{otherwise}
		\end{cases}
		\quad =\quad  \coef{\sigma'}{\sem{P}(\bvec{X}^\sigma)}\\
		\mcerr{P} &= 0 \quad = \coef{\lightning}{\sem{P}(\bvec{X}^\sigma)}
	\end{align*}\\
	
	\noindent \emph{Condition is satisfied} $(\sigma \models B)$:\\
	
	At least one loop iteration is performed.
	In order for the program to terminate in some state valuation $\sigma'$, some (non-zero) number of loop iterations must be performed.
	The termination probability can therefore be partitioned into the following infinite sum of probabilities:
	\begin{align*}
	\mcterm{P} &= \sum_{n=1}^\infty \text{Pr}^{\mathcal{R}_\sigma \sem{P}}(\diamondsuit_{= n}\angles{\downarrow, \sigma'})\\
	\mcerr{P} &= \sum_{n=1}^\infty \text{Pr}^{\mathcal{R}_\sigma \sem{P}}(\diamondsuit_{=n}\angles{\lightning})~,
	\end{align*}
	where $\text{Pr}^{\mathcal{R}_\sigma \sem{P}}(\diamondsuit_{= n}\angles{\downarrow, \sigma'})$ denotes the probability to reach state $\angles{\downarrow, \sigma'}$ after \emph{exactly} $n$ loop iterations and $\text{Pr}^{\mathcal{R}_\sigma \sem{P}}(\diamondsuit_{=n}\angles{\lightning})$ denotes the probability to reach state $\angles{\lightning}$ in exactly $n$ loop iterations.
	
	By \Cref{apx:lem:altrepr}, the PGF semantics can be represented as follows:
	\begin{align*}
	\sem{P}(\mathbf{X}^\sigma) &= \sum_{n=0}^\infty \left(\coef{\lightning}{\varphi_{B,P_1}^n(\mathbf{X}^\sigma)}X_\lightning + \angles{\varphi_{B,P_1}^n(\mathbf{X}^\sigma)}_{\neg B}\right), \quad \text{where}\\
	&\qquad \varphi_{B,P_1}(F) = \sem{P_1}(\angles{F}_B).\\
	\intertext{By the assumption that $\sigma\models B$, the $0$-th term of this series must be $0$, and thus:}
	\sem{P}(\mathbf{X}^\sigma) &= \sum_{n=1}^\infty \left(\coef{\lightning}{\varphi_{B,P_1}^n(\mathbf{X}^\sigma)}X_\lightning + \angles{\varphi_{B,P_1}^n(\mathbf{X}^\sigma)}_{\neg B}\right)
	\end{align*}
	
	We can therefore restate the initial claims of \Cref{apx:operational} as the following (stricter) conditions:
	
	\begin{enumerate}[label=\arabic*.]
		\item For all $n\in\N_{>0}$:
		\[\qquad{\Pr}^{\mathcal{R}_\sigma\sem{P}}(\diamondsuit_{=n}\angles{\downarrow, \sigma'}) = \coef{\sigma'}{\angles{\varphi_{B,P_1}^n(\mathbf{X}^\sigma)}_{\neg B}}\]
		\item For all $n\in\N_{>0}$:
		\[\qquad{\Pr}^{\mathcal{R}_\sigma\sem{P}}(\diamondsuit_{=n}\angles{\lightning}) = \coef{\lightning}{\varphi_{B,P_1}^n(\mathbf{X}^\sigma)}\]
	\end{enumerate}
	
	For both parts, we make use of the following observation, which follows from the linearity of $\varphi$ and the assumption $\sigma \models B$:
	\begin{align}
	\varphi_{B,P_1}^{n+1}(\mathbf{X}^\sigma) &= \varphi_{B,P_1}^n\left(\varphi_{B,P_1}(\mathbf{X}^\sigma)\right) \nonumber \\
	&= \varphi_{B,P_1}^n\left(\sum_{\sigma''\in \N^k}\coef{\sigma''}{\sem{P_1}(\mathbf{X}^\sigma)} \mathbf{X}^{\sigma''}\right) \nonumber \\
	&= \sum_{\sigma''\in \N^k} \coef{\sigma''}{\sem{P_1}(\mathbf{X}^\sigma)} \cdot \varphi_{B,P_1}^n(\mathbf{X}^{\sigma''}) \label{obsphi}
	\end{align}
	
	\begin{enumerate}[label=\arabic*.]
		\item First, note that a loop can never terminate in $\sigma'$ if $\sigma'\models B$. Accordingly, the construction rules of the Markov chain semantics (cf.\ \Cref{fig:mcrules}) contain the rule (while-f) as the only way of reaching a terminating state from a loop, which is only applicable if $\sigma' \not\models B$. We therefore have (for all $n\in\N^k_{>0}$):
		\[{\Pr}^{\mathcal{R}_\sigma\sem{P}}(\diamondsuit_{=n}\angles{\downarrow, \sigma'}) = \coef{\sigma'}{\angles{\varphi_{B,P_1}^n(\mathbf{X}^\sigma)}_{\neg B}} = 0\]
		
		We show the case $\sigma'\not\models B$ by induction:
		
		\paragraph{Base case:} $n = 1$.
		\begin{align*}
		& \quad {\Pr}^{\mathcal{R}_\sigma\sem P}(\diamondsuit_{=1}\angles{\downarrow, \sigma'}) \\
		&= {\Pr}^{\mathcal{R}_\sigma\sem {P_1}}(\diamondsuit\angles{\downarrow,\sigma'}) \\
		&= \coef{\sigma'}{\sem{P_1}(\mathbf{X}^\sigma)} \tag{outer I.H.}\\
		&= \coef{\sigma'}{\angles{\sem{P_1}(\mathbf{X}^\sigma)}_{\neg B}} \tag{$\sigma'\not\models B$} \\
		&= \coef{\sigma'}{\angles{\sem{P_1}(\angles{\mathbf{X}^\sigma}_B)}_{\neg B}} \tag{$\sigma\models B$} \\
		&= \coef{\sigma'}{\angles{\varphi_{B,P_1}(\mathbf{X}^\sigma)}_{\neg B}} \\
		\end{align*}
		
		\paragraph{Induction step:}
		In order for the loop to terminate in $n+1$ iterations, the first execution of the loop body must terminate in some state valuation $\sigma''$, from which the loop then terminates in $n$ iterations, i.e.,
		\begingroup
		\allowdisplaybreaks
		\begin{align*}
		& \quad {\Pr}^{\mathcal{R}_\sigma\sem{P}}(\diamondsuit_{=n+1}\angles{\downarrow, \sigma'}) \\
		&= \sum_{\sigma''\in\N^K} {\Pr}^{\mathcal{R}_\sigma\sem{P_1}}(\diamondsuit\angles{\downarrow,\sigma''} \cdot {\Pr}^{\mathcal{R}_{\sigma''}\sem P}(\diamondsuit_{=n}\angles{\downarrow, \sigma'}) \\
		&= \sum_{\sigma''\in\N^k} \coef{\sigma''}{\sem{P_1}(\mathbf{X}^\sigma)} \cdot {\Pr}^{\mathcal{R}_{\sigma''}\sem{P}}(\diamondsuit_{=n}\angles{\downarrow, \sigma'}) \tag{outer I.H.}\\
		&= \sum_{\sigma''\models B} \coef{\sigma''}{\sem{P_1}(\mathbf{X}^\sigma)}\cdot {\Pr}^{\mathcal{R}_{\sigma''}\sem{P}}(\diamondsuit_{=n}\angles{\downarrow, \sigma'}) \tag{$0$ if $\sigma''\not\models B$} \\
		&= \sum_{\sigma''\models B} \coef{\sigma''}{\sem{P_1}(\mathbf{X}^\sigma)} \cdot \coef{\sigma'}{\angles{\varphi_{B,P_1}^n(\mathbf{X}^{\sigma''})}_{\neg B}} \tag{inner I.H.} \\
		&= \sum_{\sigma''\in\N^k} \coef{\sigma''}{\sem{P_1}(\mathbf{X}^\sigma)} \cdot \coef{\sigma'}{\angles{\varphi_{B,P_1}^n(\mathbf{X}^{\sigma''})}_{\neg B}} \tag{$0$ if $\sigma''\not\models B$} \\
		&= \sum_{\sigma''\in\N^k} \coef{\sigma'}{\angles{\coef{\sigma''}{\sem{P_1}(\mathbf{X}^\sigma)} \cdot \varphi_{B,P_1}^n(\mathbf{X}^{\sigma''})}_{\neg B}} \tag{Lin. of $\coef{\sigma'}{}$ and $\angles{\cdot}_B$} \\
		&= \coef{\sigma'}{\angles{\sum_{\sigma''\in\N^k} \coef{\sigma''}{\sem{P_1}(\mathbf{X}^\sigma)} \cdot \varphi_{B,P_1}^n(\mathbf{X}^{\sigma''})}_{\neg B}} \tag{Lin. of $\coef{\sigma'}{}$ and $\angles{\cdot}_B$} \\
		&= \coef{\sigma'}{\angles{\varphi_{B,P_1}^{n+1}(\mathbf{X}^\sigma)}_{\neg B}} \tag{by \Cref{obsphi}}
		\end{align*}
	\endgroup
		
		\item By induction:
		
		\paragraph{Base case:} $n = 1$.
		\begin{align*}
		& \quad {\Pr}^{\mathcal{R}_\sigma\sem{P}}(\diamondsuit_{=1}\angles{\lightning}) \\
		&= {\Pr}^{\mathcal{R}_\sigma\sem{P_1}}(\diamondsuit\lightning) \\
		&= \coef{\lightning}{\sem{P_1}(\mathbf{X}^\sigma)} \tag{outer I.H.}\\
		&= \coef{\lightning}{\sem{P_1}(\angles{\mathbf{X}^\sigma}_B)} \tag{$\sigma\models B$} \\
		&= \coef{\lightning}{\varphi_{B,P_1}(\mathbf{X}^\sigma)} \\
		\end{align*}
		
		\paragraph{Induction step:}
		In order for the loop to reach $\angles\lightning$ in the $(n+1)$-th iteration, the first execution of the loop body must terminate in some state valuation $\sigma''$, from where $\angles\lightning$ is then reached in the $n$-th iteration, i.e.,
		\begingroup
		\allowdisplaybreaks
		\begin{align*}
		& \quad {\Pr}^{\mathcal{R}_\sigma\sem{P}}(\diamondsuit_{=n+1}\angles\lightning) \\
		&= \sum_{\sigma''\in\N^k} {\Pr}^{\mathcal{R}_\sigma\sem{P_1}}(\diamondsuit\angles{\downarrow,\sigma''} \cdot {\Pr}^{\mathcal{R}_{\sigma''}\sem{P}}(\diamondsuit_{=n}\angles\lightning) \\
		&= \sum_{\sigma''\in\N^k} \coef{\sigma''}{\sem{P_1}(\mathbf{X}^\sigma)} \cdot {\Pr}^{\mathcal{R}_{\sigma''}\sem{P}}(\diamondsuit_{=n}\angles\lightning) \tag{outer I.H.}\\
		&= \sum_{\sigma''\models B} \coef{\sigma''}{\sem{P_1}(\mathbf{X}^\sigma)} \cdot {\Pr}^{\mathcal{R}_{\sigma''}\sem{P}}(\diamondsuit_{=n}\angles\lightning) \tag{$0$ if $\sigma''\not\models B$} \\
		&= \sum_{\sigma''\models B} \coef{\sigma''}{\sem{P_1}(\mathbf{X}^\sigma)} \cdot \coef{\lightning}{\varphi_{B,P_1}^n(\mathbf{X}^{\sigma''})} \tag{inner I.H.} \\
		&= \sum_{\sigma''\in\N^k} \coef{\sigma''}{\sem{P_1}(\mathbf{X}^\sigma)} \cdot \coef{\lightning}{\varphi_{B,P_1}^n(\mathbf{X}^{\sigma''})} \tag{$0$ if $\sigma''\not\models B$} \\
		&= \sum_{\sigma''\in\N^k} \coef{\lightning}{\coef{\sigma''}{\sem{P_1}(\mathbf{X}^\sigma)} \cdot \varphi_{B,P_1}^n(\mathbf{X}^{\sigma''})} \tag{Lin. of $\coef{\lightning}{}$} \\
		&= \coef{\lightning}{\sum_{\sigma''\in\N^k} \coef{\sigma''}{\sem{P_1}(\mathbf{X}^\sigma)} \cdot \varphi_{B,P_1}^n(\mathbf{X}^{\sigma''})}\tag{Lin. of $\coef{\lightning}{}$ and $\angles{\cdot}_B$} \\
		&= \coef{\lightning}{\varphi_{B,P_1}^{n+1}(\mathbf{X}^\sigma)} \tag{by \Cref{obsphi}}
		\end{align*}
		\endgroup
	\end{enumerate}
\end{proof}

\begin{theorem}[Operational Equivalence]
	\label{apx:operational_equivalence}
	For every \cpgcl program $p$ and every $\sigma, \sigma' \in \N^k$
	\begin{align*}
		\text{Pr}^{\mathcal{R}_\sigma}\sem{P}(\diamondsuit\langle\downarrow,\sigma'\rangle\mid\neg\diamondsuit\lightning) \eeq \coef{\sigma'}{\normalize(\sem{P}(\bvec{X}^\sigma))}~.
	\end{align*}
	This includes the case of undefined semantics, i.e., the left-hand side is undefined if and only if the right-hand side is undefined.
\end{theorem}
\begin{proof}
	\begin{align*}
	\text{Pr}^{\mathcal{R}_\sigma\sem{P}}(\diamondsuit\langle\downarrow,\sigma'\rangle\mid\neg\diamondsuit\lightning) &=
	\frac{\text{Pr}^{\mathcal{R}_\sigma\sem{P}}(\diamondsuit\langle\downarrow,\sigma'\rangle\wedge\neg\diamondsuit\lightning)}{\text{Pr}^{\mathcal{R}_\sigma\sem{P}}(\neg\diamondsuit\lightning)}\\
	&= \frac{\mcterm{P}}{\text{Pr}^{\mathcal{R}_\sigma\sem{P}}(\neg\diamondsuit\lightning)} \tag{reaching $\langle\downarrow,\sigma'\rangle$ implies not reaching $\langle\lightning\rangle$}\\
	&=\frac{\mcterm{P}}{1 - \mcerr{P}}\\
	&=\frac{\coef{\sigma'}{\sem{P}(\bvec{X}^\sigma)}}{1 - \coef{\lightning}{\sem{P}(\bvec{X}^\sigma)}}
	\tag{cf. \Cref{apx:operational}}\\
	&= \coef{\sigma'}{\normalize(\sem{P}(\bvec{X}^\sigma))}
	\end{align*}
\end{proof}

\section{Reasoning about Loops}
\label{apx:loops}

\begin{definition}[Admissible eSOP-transformer]
	\label{def:admissibleeSOPTransformer}
	A function $\psi \colon \esop \to \esop$ is called \emph{admissible} if
	\begin{itemize}
		\item $\psi$ is \emph{continuous} on $\esop$.
		\item $\psi$ is \emph{linear} in the following sense:
		For all $F, G \in \esop$ and $p \in [0,1]$
		\[
		pF + G \in \esop
		\quad\text{implies}\quad
		\psi(pF + G)
		\eeq
		p \psi(F) + \psi(G)
		~.
		\]
		\item $\psi$ is \emph{homogeneous w.r.t.\ meta-indeterminates}, i.e., for all $G \in \esop$ and $\tau \in \N^l$,
		\[
		\psi(G \bvec{U}^\tau)
		\eeq
		\psi(G) \bvec{U}^\tau
		~.
		\]
		\item $\psi$ \emph{preserves ePGF}, i.e., $G \in \epgf$ implies $\psi(G) \in \epgf$.
	\end{itemize}
\end{definition}

\begin{theorem}[\sop Semantics]
	Let\, $P$ be a loop-free \credip program.
	Let\, $G = \sum_{\sigma \in \N^k} G_\sigma \bvec{U}^\sigma \in \esop$.
	The \esop semantics $\sem{P}\colon \esop \to \esop$ of $P$ can be computed by
	\[
	\sem{P}(G) \eeq \sum\nolimits_{\sigma \in \N^k} \sem{P}(G_\sigma) \cdot \bvec{U}^\sigma~.
	\]
\end{theorem}
\begin{proof}[Proof outline]
		Note that every eSOP $G$ can be decomposed into
		\[
			\sum_{\sigma \in \N^k} G_\sigma \bvec{U}^\sigma \eeq \underbrace{\sum_{\sigma \in \N^k} \constrain{G_\sigma}{\textit{true}} \bvec{U}^\sigma}_{\in ~SOP} ~+~ \sum_{\sigma \in \N^k} \coef{\lightning}{G_\sigma}X_\lightning\bvec{U}^\sigma~,
		\]
		by simple \esop arithmetic.
		Analogue to \cref{lem:errpassthru} one can show that the observe-violation probabilities pass through the eSOP semantics unaffected, i.e.,
		\begin{equation}
			\label{eqn:sopeqn}
			\sem{P}(G) \eeq \sem{P}\left(\sum_{\sigma \in \N^k} \constrain{G_\sigma}{\textit{true}} \bvec{U}^\sigma\right) ~+~ \sum_{\sigma \in \N^k} \coef{\lightning}{G_\sigma}X_\lightning\bvec{U}^\sigma~.
		\end{equation}
		Using the latter fact, the proof of \Cref{thm:sop_semantics} proceeds along a similar line of reasoning as in \cite{CAV22} by showing that $\sem{P}$ is an admissible \esop transformer.


		All loop-free cases but \codify{observe} coincide with \redip \cite{CAV22} on the distributions where the observe violation probability is zero which is an immediate consequence of \cref{eqn:sopeqn} and the results in \cite[Appendix F]{CAV22}.
		To complete the proof, we show that the semantics of $\observe{\FALSE}$ is also admissible.
		Recall the $\observe{\FALSE}$ semantics: $G[\bvec{X}/\bvec{1}, X_\lightning /1]\cdot X_\lightning$.
		Note that the $\observe{\FALSE}$ semantics is entirely based on the following elementary transformations, which are admissible (by \cite{CAV22}):
		\begin{itemize}
			\item Multiplication by a constant $G\in \textup{ePGF}\colon \lambda F.~ G \cdot F$
			\item Substitution of $X \in \bvec{X}$ by a constant $G\in \textup{ePGF}\colon \lambda F.~ F[X/G]$.
		\end{itemize}
		Thus, $\sem{\observe{\FALSE}}$ is admissible as a composition of admissible transformations.

		We use the fact that admissible transformers allow for \enquote{infinite linearity} applications (see \cite[Appendix F.5]{CAV22}), to conclude
		\begingroup
		\allowdisplaybreaks
		\begin{align*}
			\sem{P}(G) & = \sum_{\sigma \in \N^k} \sem{P}(G_\sigma \bvec{U}^\sigma) \tag{by infinite linearity}                                                                                                                                                                       \\
			           & = \sum_{\sigma \in \N^k} \sem{P}\left(\sum_{\tau \in \N^k} \coef{\tau}{G_\sigma} \bvec{X}^\tau\bvec{U}^\sigma + \coef{\lightning}{G_\sigma} X_\lightning\bvec{U}^\sigma\right) \tag{by \esop arithmetic}                                                     \\
			           & = \sum_{\sigma \in \N^k} \sum_{\tau \in \N^k} \sem{P}\left(\coef{\tau}{G_\sigma} \bvec{X}^\tau\bvec{U}^\sigma + \coef{\lightning}{G_\sigma} X_\lightning\bvec{U}^\sigma\right)\tag{by infinite linearity}                                                    \\
			           & = \sum_{\sigma \in \N^k} \sum_{\tau \in \N^k} \sem{P}\left(\coef{\tau}{G_\sigma} \bvec{X}^\tau\bvec{U}^\sigma\right) + \sem{P}\left(\coef{\lightning}{G_\sigma} X_\lightning\bvec{U}^\sigma\right) \tag{Lin.\ of \sem{P}}                                    \\
			           & = \sum_{\sigma \in \N^k} \sum_{\tau \in \N^k} \coef{\tau}{G_\sigma}\cdot\sem{P}\left( \bvec{X}^\tau\right)\bvec{U}^\sigma + \coef{\lightning}{G_\sigma} X_\lightning\bvec{U}^\sigma \tag{by admissible \sem{P} and \cref{eqn:sopeqn}}                        \\
			           & = \sum_{\sigma \in \N^k} \sum_{\tau \in \N^k} \sem{P}\left(\coef{\tau}{G_\sigma}\cdot \bvec{X}^\tau\right)\bvec{U}^\sigma + \sem{P}\left(\coef{\lightning}{G_\sigma} X_\lightning\right)\bvec{U}^\sigma\tag{by \cref{thm:linsem} and \cref{lem:errpassthru}} \\
			           & = \sum_{\sigma \in \N^k} \sum_{\tau \in \N^k} \left(\sem{P}\left(\coef{\tau}{G_\sigma}\cdot \bvec{X}^\tau\right) + \sem{P}\left(\coef{\lightning}{G_\sigma} X_\lightning\right)\right)\bvec{U}^\sigma \tag{by \esop arithmetic}                              \\
			           & = \sum_{\sigma \in \N^k} \sum_{\tau \in \N^k} \left(\sem{P}\left(\coef{\tau}{G_\sigma}\cdot \bvec{X}^\tau + \coef{\lightning}{G_\sigma} X_\lightning\right) \right)\bvec{U}^\sigma \tag{\cref{thm:linsem}}                                                   \\
			           & = \sum_{\sigma \in \N^k} \sem{P}\left(G_\sigma\right)\bvec{U}^\sigma \tag{by Def. of $G_\sigma$}
		\end{align*}
		\endgroup
\end{proof}

\begin{lemma}[\esop Characterization]
	Let $P_1$ and $P_2$ be loop-free \credip-programs with $\text{Vars}(P_i) \subseteq \{\progvar{x}_1,\ldots, \progvar{x}_k\}$ for $i \in \{1, 2\}$. Further, consider a vector $\bvec{U} = (U_1,\ldots, U_k)$ of meta-indeterminates, and let $\hat{G}$ be the eSOP
	$(1-X_1U_1)^{-1}\cdots(1-X_kU_k)^{-1} \in \R[[\bvec{X}, \bvec{U}]].$
	Then, 
	\[
	\forall G \in \textup{\epgf}. ~ \sem{P_1}(G)~=~\sem{P_2}(G) \quad\iff\quad \sem{P_1}(\hat{G})~=~\sem{P_2}(\hat{G}).
	\]
\end{lemma}
\begin{proof}
	We observe that 
	$
	\hat{G}
	~=~
	\sum_{\sigma \in \N^k} \bvec{X}^\sigma \bvec{U}^\sigma
	$. Then we have 
	\begingroup
	\allowdisplaybreaks
	\begin{align*}
	\allowdisplaybreaks
	& \sem{P_1}(\hat{G}) = \sem{P_2}(\hat{G}) \\
	\iff\quad & \sem{P_1}(\hat{G}) - \sem{P_2}(\hat{G}) = 0 \\
	\iff\quad & \sem{P_1}(\sum_{\sigma \in \N^k} \bvec{X}^\sigma\bvec{U}^\sigma) - \sem{P_2}(\sum_{\sigma \in \N^k} \bvec{X}^\sigma \bvec{U}^\sigma) = 0 \\
	\iff\quad & \sum_{\sigma \in \N^k} \sem{P_1}(\bvec{X}^\sigma) \bvec{U}^\sigma  - \sum_{\sigma \in \N^k} \sem{P_2}(\bvec{X}^\sigma) \bvec{U}^\sigma = 0 \tag{By \Cref{thm:sop_semantics}} \\
	\iff\quad & \sum_{\sigma \in \N^k} ( \sem{P_1}(\bvec{X}^\sigma) -  \sem{P_2}(\bvec{X}^\sigma)) \bvec{U}^\sigma = 0 \tag{rewriting} \\
	\iff\quad & \forall\sigma \in \N^k \colon \sem{P_1}(\bvec{X}^\sigma) -  \sem{P_2}(\bvec{X}^\sigma) = 0 \tag{By definition of the 0-FPS in $\R[[\bvec{X}, X_\lightning, \bvec{U}]]$}\\
	\iff\quad & \forall\sigma \in \N^k\colon \sem{P_1}(\bvec{X}^\sigma) = \sem{P_2}(\bvec{X}^\sigma) \\
	\iff\quad & \sem{P_1} = \sem{P_2} \tag{by \citet{Kozen} and \cref{lem:errpassthru}}
	\end{align*}
	\endgroup
\end{proof}

\section{Parameter Synthesis}
\label{apx:synthesis}

\begin{theorem}[Decidability of Parameter Synthesis]
	Let $W$ be a \credip \WHILESYMBOL\ loop and $I_{\bvec{p}}$ be a parametrized loop-free \credip program.
	It is decidable whether there exist parameter values $\bvec{\rho}$ such that the instantiated template $I_\bvec{\rho}$ is an invariant, i.e.,
	\[
	\exists \bvec{p} \in \R^l. \quad \sem{W} = \sem{I_{\bvec{p}}}
	~.
	\]
\end{theorem}
\begin{proof}
	Let $W$ and $I_{\bvec{p}}$ be given as described. Also, let $\hat{G} = (1-X_1U_1)^{-1} \cdots (1-X_kU_k)^{-1} \in \esop$ which is a rational closed form.
	\begingroup
	\allowdisplaybreaks
	\begin{align*}
	&\exists \bvec{p} \in \R^l.~ \quad \quad \sem{ I_{\bvec{p}} } = \sem{\Phi_{B,P} ( I_{\bvec{p}} ) }\\
	\Leftrightarrow ~& \exists \bvec{p} \in \R^l.~ \quad \sem{ I_{\bvec{p}} }(\hat{G}) = \sem{\Phi_{B,P} ( I_{\bvec{p}} ) }(\hat{G}) \tag{\cref{lem:sopequiv}}\\
	\Leftrightarrow~&  \exists \bvec{p} \in \R^l.~ \quad \frac{F_\bvec{p}}{H_\bvec{p}} = \frac{\hat{F}_\bvec{p}}{\hat{H}_\bvec{p}} \tag{loop-free \credip preserves rational functions}\\
	\Leftrightarrow~&  \exists \bvec{p} \in \R^l.~ \quad F_\bvec{p} \hat{H}_\bvec{p} = \hat{F}_\bvec{p} H_\bvec{p} \\
	\Leftrightarrow~&  \exists \bvec{p} \in \R^l.~ \quad F_\bvec{p} \hat{H}_\bvec{p}  - \hat{F}_\bvec{p} H_\bvec{p} = 0
	\end{align*}%
	\endgroup
	In the last step, $F_\bvec{p} \hat{H}_\bvec{p}$ and $\hat{F}_\bvec{p} H_\bvec{p}$ are polynomials in $\R[\bvec{p}][\bvec{X}, X_\lightning, \bvec{U}]$ ($\bvec{p}$ can only occur as probabilities in $I_\bvec{p}$).
	Using the results about quantifier elimination in the theory of non-linear real arithmetic (by Cylindrical Algebraic Decomposition \cite{caviness2012quantifier}), we have a decision procedure of $\mathcal{O}\left(2^{2^{\abs{X} + \abs{U} + 1}}\right)$ worst-case complexity to decide whether the formula can be satisfied. 
\end{proof}

\section{Benchmarks and Additional Examples}
\label{apx:prodigy}

\begin{example}[The Invariant for Prog.~\ref{prog:brp}]
	\begingroup
	\allowdisplaybreaks
	\begin{align*}
	&\IF{\progvar{s} > 0 \wedge \progvar{f} < 5}\\
	& \quad \pcomment{no more transmission failures allowed, except for last 9 packets} \\
	& \quad \IF{\progvar{s} \geq 10}\\
	& \quad \quad \IFINLINE{\iid{\bernoulli{\sfrac{1}{100}}}{\progvar{s} - 9} = 0}{\pskip} \ELSESYMBOL~ \{\OBSERVE{\FALSE}\}\fatsemi\\
	& \quad \quad \ASSIGN{\progvar{s}}{9}\fatsemi\\
	& \quad \quad \ASSIGN{\progvar{f}}{0} \}\\
	& \quad \pcomment{$\leq 9$ packets left; state can have failed attempts for first packet} \\
	& \quad \IF{\iid{\bernoulli{\sfrac{99}{100}}}{5-\progvar{f}} > 0} \\
	& \quad \quad \ASSIGN{\progvar{f}}{0}\fatsemi\\
	& \quad \quad \ASSIGN{\progvar{s}}{\progvar{s}-1}\fatsemi\\
	& \quad \quad \pcomment{each remaining packet fails the transmission with probability $p$} \\
	& \quad \quad \IF{\progvar{s} = 8} \\
	& \quad \quad \quad \PCHOICE{\ASSIGN{\progvar{f}}{5}}{p}{\ASSIGN{\progvar{f}}{0}\fatsemi \ASSIGN{\progvar{s}}{\progvar{s}-1}}\\
	& \quad \quad \}\fatsemi\\
	& \quad \quad \IF{\progvar{s} = 7} \\
	& \quad \quad \quad \PCHOICE{\ASSIGN{\progvar{f}}{5}}{p}{\ASSIGN{\progvar{f}}{0}\fatsemi \ASSIGN{\progvar{s}}{\progvar{s}-1}}\\
	& \quad \quad \}\fatsemi\\
	& \quad \quad \IF{\progvar{s} = 6} \\
	& \quad \quad \quad \PCHOICE{\ASSIGN{\progvar{f}}{5}}{p}{\ASSIGN{\progvar{f}}{0}\fatsemi \ASSIGN{\progvar{s}}{\progvar{s}-1}}\\
	& \quad \quad \}\fatsemi\\
	& \quad \quad \IF{\progvar{s} = 5} \\
	& \quad \quad \quad \PCHOICE{\ASSIGN{\progvar{f}}{5}}{p}{\ASSIGN{\progvar{f}}{0}\fatsemi \ASSIGN{\progvar{s}}{\progvar{s}-1}}\\
	& \quad \quad \}\fatsemi\\
	& \quad \quad \IF{\progvar{s} = 4} \\
	& \quad \quad \quad \PCHOICE{\ASSIGN{\progvar{f}}{5}}{p}{\ASSIGN{\progvar{f}}{0}\fatsemi \ASSIGN{\progvar{s}}{\progvar{s}-1}}\\
	& \quad \quad \}\fatsemi\\
	& \quad \quad \IF{\progvar{s} = 3} \\
	& \quad \quad \quad \PCHOICE{\ASSIGN{\progvar{f}}{5}}{p}{\ASSIGN{\progvar{f}}{0}\fatsemi \ASSIGN{\progvar{s}}{\progvar{s}-1}}\\
	& \quad \quad \}\fatsemi\\
	& \quad \quad \IF{\progvar{s} = 2} \\
	& \quad \quad \quad \PCHOICE{\ASSIGN{\progvar{f}}{5}}{p}{\ASSIGN{\progvar{f}}{0}\fatsemi \ASSIGN{\progvar{s}}{\progvar{s}-1}}\\
	& \quad \quad \}\fatsemi\\
	& \quad \quad \IF{\progvar{s} = 1} \\
	& \quad \quad \quad \PCHOICE{\ASSIGN{\progvar{f}}{5}}{p}{\ASSIGN{\progvar{f}}{0}\fatsemi \ASSIGN{\progvar{s}}{\progvar{s}-1}}\\
	& \quad \quad \}\\
	& \quad \ELSE \ASSIGN{\progvar{f}}{5} \} \}
	\end{align*}%
	\endgroup
	The inferred transmission-failure probability is
	\[
		\frac{51601999994933400000469819999961544000002572999999869540000004609999999899900000001}{51711\cdot 10^{88}}~.
	\]
\end{example}

\begin{example}
	This example of two programs $I$ and $J$ shows the step-by-step computation of the invariant and the modified invariant to prove the actual equivalence.
	\label{apx:ex:complete}
	\begingroup
	\allowdisplaybreaks
	\begin{align*}
	I: \qquad & \annotate{(1-XU)^{-1}(1-YV)^{-1}} \\
	& \IF{\progvar{y} = 1}\\
	& \qquad \annotate{(1-XU)^{-1}YV} \\
	& \qquad \incrasgn{\progvar{x}}{\iid{\geometric{\sfrac{1}{2}} + 1}{\progvar{y}}} \fatsemi \\
	& \qquad \annotate{(1-XU)^{-1}X(2-X)^{-1}YV} \\
	& \qquad \ASSIGN{\progvar{y}}{0}\fatsemi \\
	& \qquad \annotate{(1-XU)^{-1}X(2-X)^{-1}V} \\
	& \qquad \observe{\progvar{x} < 3} \\
	& \qquad \annotate{(\sfrac{1}{2}X + \sfrac{1}{4} X^2 + \sfrac{1}{4}X_\lightning)V} \\
	&\qquad \qquad \annocolor{+ \left((\sfrac{1}{2} X^2 + \sfrac{1}{2}X_\lightning)U + X_\lightning U^2(1-U)^{-1}\right)V} \\
	& \} \\
	& \annotate{(1-XU)^{-1}(1-YV)^{-1} - (1-XU)^{-1}YV} \\
	& \qquad\annocolor{+ (\sfrac{1}{2}X + \sfrac{1}{4} X^2 + \sfrac{1}{4}X_\lightning)V} \\
	&\qquad \annocolor{+ \left((\sfrac{1}{2} X^2 + \sfrac{1}{2}X_\lightning)U + X_\lightning U^2(1-U)^{-1}\right)V}
	\end{align*}%
	\endgroup
	
	We want to show that $\sem{J}(\hat{G})$ --- where $J = \IF{\progvar{y}=1}P\fatsemi I\ELSE \pskip\}$ --- yields the same result:
	
	\begingroup
	\allowdisplaybreaks
	\begin{align*}
	J: \qquad & \annotate{(1-XU)^{-1}(1-YV)^{-1}} \\
	& \IF{\progvar{y} = 1} \\
	& \qquad \annotate{(1-XU)^{-1}YV} \\
	& \qquad \PCHOICE{\ASSIGN{\progvar y}{0}}{\sfrac{1}{2}}{\ASSIGN{\progvar y}{1}} \fatsemi \\
	& \qquad \annotate{\sfrac{1}{2}(1-XU)^{-1}(Y+1)V} \\
	& \qquad \ASSIGN{\progvar x}{\progvar x + 1} \fatsemi \\
	& \qquad \annotate{\sfrac{1}{2}X(1-XU)^{-1}(Y+1)V} \\
	& \qquad \observe{\progvar x < 3}\fatsemi \\
	& \qquad \annotate{\left(\sfrac{1}{2}(X + X^2U)(Y+1) + X_\lightning U^2(1-U)^{-1}\right)V} \\
	& \qquad \IF{\progvar{y} = 1}\\
	& \qquad\qquad \annotate{\sfrac{1}{2}(X + X^2U)YV} \\
	& \qquad\qquad \incrasgn{\progvar{x}}{\iid{\geometric{\sfrac{1}{2}} + 1}{\progvar{y}}} \fatsemi \\
	& \qquad\qquad \annotate{\sfrac{1}{2}(X^2+X^3U)(2-X)^{-1}YV} \\
	& \qquad\qquad \ASSIGN{\progvar{y}}{0}\fatsemi \\
	& \qquad\qquad \annotate{\sfrac{1}{2}(X^2+X^3U)(2-X)^{-1}V} \\
	& \qquad\qquad \observe{\progvar{x} < 3} \\
	& \qquad\qquad \annotate{\sfrac{1}{2}\left(\sfrac{1}{2}X^2 + \sfrac{1}{2}X_\lightning + X\lightning U\right)V} \\
	& \qquad \} \\
	& \qquad \annotate{\left((\sfrac{1}{2}X + \sfrac{1}{4}X^2 + \sfrac{1}{4}X_\lightning) + (\sfrac{1}{2}X^2 + \sfrac{1}{2}X_\lightning)U\right)V}\\
	&\qquad \qquad \annocolor{+ X_\lightning U^2 (1-U)^{-1}V} \\
	& \} \\
	& \annotate{(1-XU)^{-1}(1-YV)^{-1} - (1-XU)^{-1}YV} \\
	& \qquad \annocolor{+ \left((\sfrac{1}{2}X + \sfrac{1}{4} X^2 + \sfrac{1}{4}X_\lightning) + (\sfrac{1}{2} X^2 + \sfrac{1}{2}X_\lightning)U\right)V}\\
	&\qquad \annocolor{+ X_\lightning U^2(1-U)^{-1}V}
	\end{align*}%
	\endgroup
\end{example}

\end{document}